\documentclass[a4paper,11pt]{book}
\usepackage[utf8]{inputenc}
\usepackage{bbold}
\usepackage{amsmath}
\usepackage{amssymb}
\usepackage{amsthm}
\usepackage{hyperref}
\usepackage{enumitem,kantlipsum}
\usepackage{geometry}
\usepackage{setspace}

\newtheorem{theorem}{Theorem}[section]
\newtheorem{thm}[theorem]{Theorem}
\newtheorem{prop}[theorem]{Proposition}
\newtheorem{lem}[theorem]{Lemma}
\newtheorem{cor}[theorem]{Corollary}

\theoremstyle{definition}
\newtheorem{exmp}[theorem]{Example}
\newtheorem{defn}[theorem]{Definition}

\theoremstyle{remark}
\newtheorem{rem}[theorem]{Remark}
\newtheorem*{ntn}{Notation}

\newcommand{\scp}[2]{\langle #1, #2 \rangle}
\newcommand{\jap}[1]{\langle #1 \rangle}
\newcommand*\diff{\mathop{}\!\mathrm{d}}
\newcommand*\supp{\mathop{}\!\mathrm{supp}}
\newcommand*\one{\mathop{}\!\mathbb{1}}
\newcommand*\e{\mathop{}\!\mathrm{e}}
\newcommand*\I{\mathop{}\!\mathrm{i}}
\newcommand*\Op{\mathop{}\!\mathrm{Op}}
\newcommand*\singsupp{\mathop{}\!\mathrm{sing \ supp}}
\newcommand{\norm}[1]{\lVert#1\rVert}

\begin{document}

\author{Janik Kruse}
\title{The Nelson Model on Static Spacetimes}
\date{\today}

\frontmatter
%\thispagestyle{empty}
%\maketitle

\newpage
\thispagestyle{empty}

\begin{titlepage}
	\setstretch{1.5}
	\begin{center}
		\thispagestyle{empty}
		\LARGE{\textsc{Ludwig-Maximilians-Universität München}}\\
		\large{Fakultät für Physik}\\
		\large{Fakultät für Mathematik, Informatik und Statistik} \\[1ex]
		\LARGE{\textsc{Technische Universität München}}\\
		\large{Fakultät für Physik}\\
		\large{Fakultät für Mathematik}
		\vspace{0.7cm}
		\begin{center}
			\vspace{1.2cm}
			\textbf{\LARGE{Master Thesis}}
			\medskip\par
			\textbf{\normalsize{submitted for the degree of}} \\[2ex]
			\textbf{\Large{Master of Science}}\\
			\vspace{1.6cm}
			\Large{\textbf{The Nelson Model on Static Spacetimes}}\\[-0.5ex]
			\bigskip\par
			\large{by} \par
			\large{\textsc{Janik Kruse}}
			\vspace{3.2cm}
		\end{center}
		\medskip
	\end{center}
	\begin{tabular}{ll}
		Submission Date:  & June 25, 2021 \\[-1ex]
		Supervisors: & Dr. Jonas Lampart (Université de Bourgogne) \\[-1ex]
		& Prof. Dr. Phan Thành Nam
	\end{tabular}
	\vspace{0.7cm}
	\begin{center}
		Revised version of \today.
	\end{center}
\end{titlepage}

\thispagestyle{empty}

\chapter*{Preface}

The Nelson model describes the interaction between nonrelativistic quantum particles and a relativistic quantum field of scalar bosons (e.g. interaction between nonrelativistic nucleons and relativistic mesons in an atomic nucleus). In the present work, we discuss the regularisation of this semi-relativistic model. The formal Nelson Hamiltonian $H$ acting on the Hilbert space $\mathfrak{H} = L^2(\mathbb{R}^3, \diff X) \otimes \Gamma_\mathrm{s}(L^2(\mathbb{R}^3, \diff x))$ is given by
\begin{align*}
	H = -\Delta_X \otimes \mathbb{1} + \mathbb{1} \otimes \diff \Gamma(\omega) + \Phi(\omega^{-\frac{1}{2}} \delta_X),
\end{align*}
where $\omega = \sqrt{-\Delta + m^2}$ is the relativistic dispersion relation with bosonic field mass $m$, $\diff\Gamma(\omega)$ is its second quantisation on the symmetric Fock space $\Gamma_\mathrm{s}(L^2(\mathbb{R}^3, \diff x))$, $\Phi(\omega^{-\frac{1}{2}} \delta_X)$ is a field operator with form factor $v_X = \omega^{-\frac{1}{2}} \delta_X/\sqrt{2}$, and $\delta_X$ is the Dirac delta centred at the particle position $X$. The field operator $\Phi(f) = (a^*(f) + a(f))/\sqrt{2}$ is a sum of creation and annihilation operators. However, the operator $H$ is not well-defined because the form factor $v_X \notin L^2(\mathbb{R}^3,\diff x)$ is too singular. 

We present two different techniques of regularising the Nelson Hamiltonian in this work: renormalisation and regularisation by interior boundary conditions. To re\-nor\-ma\-lise the Hamiltonian $H$, we demonstrate that putting an ultraviolet cut-off $\Lambda$ on the interaction (technically this amounts to replacing $\delta_X$ by some smoother function $\rho_{\Lambda,X}$) leads to a well-defined self-adjoint cut-off Hamiltonian $H_\Lambda$. The cut-off Hamiltonian has a finite vacuum energy $-E_\Lambda$. We establish that once the vacuum energy is removed from the cut-off Hamiltonian, a well-defined limit of $H_\Lambda + E_\Lambda$ as $\Lambda \to \infty$ exists. This limit operator is the renormalised Nelson Hamiltonian.

A weakness of this method is that it provides little information about the renormalised Nelson Hamiltonian. The second regularisation method we present is the method of interior boundary conditions (IBC).

The main idea behind the IBC method is to interpret the quantity $a^*(v_X) \Psi$ for suitable $\Psi \in \mathfrak{H}$ as a distribution. In return, we also consider $H_0\Psi$ as a distribution and choose $\Psi \in \mathfrak{H}$ in such a way that the singular parts of $a^*(v_X) \Psi$ and $H_0 \Psi$ exactly cancel each other. Then, $L=H_0 + a^*(v_X)$ is a well-defined operator. The main difficulty of this approach is to find a physically reasonable extension $A$ of the annihilation operator $a(v_X)$ to the domain $D(L)$ of $L$. The IBC Hamiltonian $L+A$ constructed in this way is well defined, self-adjoint, and equivalent to the re\-nor\-ma\-lised Nelson Hamiltonian. However, this time, we obtain an explicit description of the renormalised Nelson Hamiltonian and its domain, which was not available before.

In the present work, we generalise both regularisation techniques to the Nelson model on static spacetimes. In the Nelson model on static spacetimes the Laplace operator is replaced by the Laplace--Beltrami operator, which is the correct generalisation of the Laplace operator on pseudo-Riemannian manifolds. We revisit a paper by Gérard \textit{et al.} \cite{gerard2012}, who translated Nelson's renormalisation technique to the Nelson model on static spacetimes. The original part of this work is to apply the IBC method to the Nelson model on static spacetimes. 

Usually, the Nelson model (on Euclidean spacetime) is discussed in Fourier space because the Laplace operator is a simple multiplication operator in Fourier space. The main difficulty of the present work lies in the fact that no simple representation of the Nelson Hamiltonian on static spacetimes in Fourier space exists because the Laplace--Beltrami operator is a partial differential operator with variable coefficients. Instead, we rely on the pseudo-differential calculus.

%The structure of this work is as follows. In the first chapter we introduce the reader to the Fock space formulation and second quantisation. Then we define the Nelson model on Euclidean as well as on static spacetimes. In the second chapter we provide an overview of the pseudo-differential calculus with the most important theorems and technique we need for the present work. In chapter three we regularise the Nelson Hamiltonian on static spacetimes by renormalisation, and in the final chapter we apply the IBC method.

%The reader is supposed to be familiar with functional analysis, in particular with operator theory, and to have some familiarity with mathematical quantum mechanics. 

\subsubsection*{Acknowledgements}

This work is the result of a thesis project in the master's degree program 'Theoretical and Mathematical Physics' at the Ludwig--Maximilians--Universität München and Tech\-nische Universität München and was conducted in cooperation with the Laboratoire Interdisciplinaire Carnot de Bourgogne (ICB) at the Université de Bourgogne in Dijon. I very much appreciate the hospitality of the ICB during my stay from October 2020 to April 2021. 

I thank my supervisor Jonas Lampart for his helpful hints and discussions, and for his help to find my way in Dijon. I also thank my second referee Phan Thàn Nam who introduced me to mathematical quantum mechanics through his excellent lectures. A special word of thanks for helpful comments on the first drafts of this work goes to my fellow-students and friends Jakob Oldenburg, Jonas Peteranderl, and Markus Frankenbach.

\vspace{0.8cm}
\noindent Paderborn, June 2021

\noindent Janik Kruse
%\section*{Selbstständigkeitserklärung}
%
%Hiermit versichere ich, dass ich die vorliegende Arbeit selbstständig und ohne fremde Hilfe angefertigt habe. Ich habe keine anderen Hilfsmittel als die angegebenen Quellen hierfür verwendet und Zitate kenntlich gemacht.
%
%
%\vspace*{1.2cm}
%
%Paderborn, \today
%
%\vspace*{1.4cm}
%
%Janik Kruse
%
%\vspace*{4cm}

\section*{Statement of Authorship}

I herewith assure that I wrote the present thesis independently and that I have used no other means than the ones indicated. I have indicated all parts of the work in which sources are used according to their wording or to their meaning.

\vspace*{1.2cm}

Paderborn, \today

\vspace*{1.4cm}

Janik Kruse

\newpage

\tableofcontents

\mainmatter
\chapter{Nelson Model} 
\label{ch:NelsonModel}

\begin{verse}
	\textit{First quantisation is a mystery, but second quantisation is a functor. --- Edward Nelson}
\end{verse}

The Nelson model describes a physical system of nonrelativistic quantum particles coupled to a relativistic quantum field of scalar bosons. Nelson introduced this model in 1964 \cite{nelson1964} and rigorously demonstrated the existence of a well-defined Hamiltonian by renormalisation. In the present work,  we are interested in a generalised model on static spacetimes.

This introductory chapter is organised as follows. First, we explain the Fock space formalism and second quantisation, which is essential for the present work. Next, we give a mathematical description of the Nelson model on Euclidean spacetime and collect some results from the literature. In the last section, we define the Nelson model on static spacetimes.

\section{Fock Space Formalism}
\label{sec:FockSpaceFormalism}

In this section, we present the basic concepts of the Fock space formalism and second quantisation. More elaborate and detailed expositions are available in the literature \cite{bratteli2002, lorinczi2020}. Furthermore, we include  some results from the appendices of \cite{griesemer2016,griesemer2018}.

\subsection{Fock Space}
Let $\mathfrak{h}$ denote a complex separable Hilbert space. The state of a quantum system with $n$ particles is usually described by an element of the tensor product space $\mathfrak{h}^{\otimes n}$. If the number of particles is not fixed, the \textbf{Fock space} $\Gamma(\mathfrak{h})$ is introduced as the direct sum of all $n$-particle subspaces:
\begin{align}
	\Gamma(\mathfrak{h}) = \bigoplus_{n \in \mathbb{N}} \mathfrak{h}^{\otimes n}.
\end{align}
It is understood that $\mathfrak{h}^{\otimes 0} = \mathbb{C}$ and that the direct (Hilbert space) sum is to be completed.

An element $\Psi \in \Gamma(\mathfrak{h})$ is a sequence $\Psi = (\Psi^{(0)}, \Psi^{(1)}, \dots)$ with $\Psi^{(n)} \in \mathfrak{h}^{\otimes n}$. One particular element is the \textbf{vacuum vector} $\Omega = (1, 0, 0, \dots)$. 

The Fock space $\Gamma(\mathfrak{h})$ is a Hilbert space with the scalar product
\begin{align}
	\scp{\Psi_1}{\Psi_2}_{\Gamma(\mathfrak{h})} = \sum_{n=0}^{\infty} \scp{\Psi_1^{(n)}}{\Psi_2^{(n)}}_{\mathfrak{h}^{\otimes n}}
\end{align}
and the induced norm
\begin{align}
	\norm{\Psi}^2_{\Gamma(\mathfrak{h})} = \sum_{n=0}^{\infty} \norm{\Psi^{(n)}}^2_{\mathfrak{h}^{\otimes n}}.
\end{align}
With respect to this norm, a dense subspace of $\Gamma(\mathfrak{h})$ is given by the \textbf{finite-particle vectors}:
\begin{align}
	\Gamma_{\mathrm{fin}}(\mathfrak{h}) = \{ \Psi \in \Gamma(\mathfrak{h}) \mid \exists M \geq 0, \Psi^{(n)} = 0 \ \forall n \geq M \}.
\end{align}
For the rest of this section, let $f, f_1,\dots, f_n, g \in \mathfrak{h}$.

In physical applications, further restrictions must usually be imposed on the Fock space due to quantum statistics. To extract a subspace relevant for the description of bosonic systems, we introduce the \textbf{symmetrisation operator} $P_+$. On pure tensor states, it is defined as
\begin{align}
	P_+ (f_1 \otimes \dots \otimes f_n) = \frac{1}{n!} \sum_{\sigma \in S_n} f_{\sigma(1)} \otimes \dots \otimes f_{\sigma(n)},
\end{align}
where the sum is taken over all permutations $\sigma \in S_n$. The definition is extended by linearity and continuity to $\Gamma(\mathfrak{h})$. It is easy to verify that $P_+$ is an orthogonal projection operator. The \textbf{symmetric Fock space} $\Gamma_\mathrm{s}(\mathfrak{h})$, which is the relevant Fock space for bosonic systems, is the image of $P_+$, that is,
\begin{align}
	\Gamma_\mathrm{s}(\mathfrak{h}) = P_{+} \Gamma(\mathfrak{h}),
\end{align}
and its symmetric $n$-particle subspaces are $\mathfrak{h}^{\otimes_\mathrm{s} n} = P_{+}\mathfrak{h}^{\otimes n}$. An antisymmetric Fock space can be defined by similar means. However, in the following, we exclusively use the symmetric Fock space.

\subsection{Second Quantisation}

An operator $h$ defined on the Hilbert space $\mathfrak{h}$ can be lifted to an operator $\diff\Gamma(h)$ on the (symmetric) Fock space:
\begin{align}
	\diff \Gamma(h)\big|_{\mathfrak{h}^{\otimes_\mathrm{s} n}} = \sum_{j=1}^{n} \one^{\otimes (j-1)} \otimes h \otimes \one^{\otimes(n-j)}.
\end{align}
If $h$ is self-adjoint, $\diff\Gamma(h)$ is symmetric and hence closable. Furthermore, a dense subset of analytic vectors of $\diff\Gamma(h)$ is given by any finite sum of symmetrised products of analytic vectors of $h$. Thus, the operator $\diff\Gamma(h)$ is essentially self-adjoint according to Nelson's analytic vector theorem (see \cite[Thm.~X.39, Cor.~2]{reed2}). Its self-adjoint closure, for which we continue to write $\diff\Gamma(h)$, is called the \textbf{second quantisation} of $h$. 

\begin{exmp}
	The second quantisation of the identity operator is the \textbf{number operator} $N = \diff\Gamma(\mathbb{1})$ with $N\Psi^{(n)} = n \Psi^{(n)}$.
\end{exmp}

\begin{exmp}
	Let $\mathfrak{h} = L^2(\mathbb{R}^d)$ and let $\omega = \sqrt{-\Delta + m^2}$ be the \textbf{relativistic dispersion relation} of a boson with mass $m\geq 0$. Then the action of the second quantisation $\diff\Gamma(\omega)$ on $\Psi \in D(\diff\Gamma(\omega))$ in Fourier space is given by
	\begin{align}
		\widehat{(\diff \Gamma(\omega) \Psi)}^{(n)}(\xi_1, \dots, \xi_n) = \sum_{j=1}^{n} \omega(\xi_j) \widehat{\Psi}^{(n)}(\xi_1, \dots, \xi_n)
	\end{align}
	with $\omega(\xi) = \sqrt{|\xi|^2 + m^2}$. The second quantisation $\diff \Gamma(\omega)$ describes the total energy of a bosonic system. In the Nelson model, the bosonic quantum field is such a system; therefore, we call $\diff \Gamma(\omega)$ the \textbf{free field Hamiltonian}.
	\label{exmp:FreeFieldHamiltonian}
\end{exmp}

Another important set of linear operators $a(f):\mathfrak{h}^{\otimes_\mathrm{s} n} \to \mathfrak{h}^{\otimes_\mathrm{s} (n-1)}$ (set $\mathfrak{h}^{\otimes_\mathrm{s}-1} = 0$) and $a^{*}(f):\mathfrak{h}^{\otimes_\mathrm{s} n} \to \mathfrak{h}^{\otimes_\mathrm{s} (n+1)}$ is defined as follows. If $n=0$, we initially define $a(f)\Omega = 0$ and $a^{*}(f)\Omega = f$. For $n\geq 1$, we set
\begin{align}
	a(f) f_1 \otimes \dots \otimes f_n &= \sqrt{n}\scp{f}{f_1}_{\mathfrak{h}} f_2 \otimes \dots \otimes f_n, \\
	a^*(f)\Psi^{(n)} &= \sqrt{n+1}P_{+} \left(f \otimes \Psi^{(n)}\right). \label{def:CreationOperator}
\end{align}
These operators can be extended to operators on $\Gamma_\mathrm{s}(\mathfrak{h})$. Because  $a(f)$ decreases the number of particles (i.e.  $Na(f)\Psi^{(n)} = (n-1)a(f)\Psi^{(n)}$), the operator $a(f)$ is called \textbf{annihilation operator}. The operator $a^{*}(f)$ increases the number of particles, and therefore is called \textbf{creation operator}.

The \textbf{commutation relations} are well known and easy to verify: 
\begin{align}
	[a(f), a(g)] &= 0, \label{eq:CommutatorAnni} \\
	[a^*(f), a^*(g)] &= 0, \\
	[a(f), a^*(g)] &= \scp{f}{g}. 
	\label{eq:CommutatorAnniCrea}
\end{align}
Moreover, if $h$ is a self-adjoint operator on $\mathfrak{h}$ and $f \in D(h)$, the commutators of the second quantisation $\diff\Gamma(h)$ with annihilation and creation operators are
\begin{align}
	[\diff\Gamma(h),a^*(f)] &= a^*(h f), \label{eq:SecondQuantisationCreationCommutator} \\
	[\diff \Gamma(h), a(f)] &= -a(h f) \label{eq:SecondQuantisationAnnihilationCommutator}.
\end{align}

It is useful to have an explicit representation of annihilation and creation operators for $\mathfrak{h} = L^2(\mathbb{R}^d)$. Let $x = (x_1, \dots, x_n) \in \mathbb{R}^{dn}$ and let $\hat{x}_j \in \mathbb{R}^{d(n-1)}$ denote the vector $x$ with the entry $j$ removed. Then, for any $\Psi^{(n)} \in L^2(\mathbb{R}^{dn})$, $\Psi^{(n-1)} \in L^2(\mathbb{R}^{d(n-1)})$, annihilation and creation operators may be written as
\begin{align}
	a(f) \Psi^{(n)} (\hat{x}_{n-1}) &= \sqrt{n} \int_{\mathbb{R}^d} \overline{f(x_n)} \Psi^{(n)}(x) \diff x_n, \\
	a^*(f)\Psi^{(n-1)}(x) &= \frac{1}{\sqrt{n}} \sum_{j=1}^{n} f(x_j) \Psi^{(n-1)}(\hat{x}_j).
\end{align}

The following lemma provides useful norm estimates for annihilation and creation operators.
\begin{lem}
	\label{lem:EstimatesAnnihilationCreation}
	Let $h \geq 1$ be a self-adjoint operator on $\mathfrak{h}$, and $f \in D(h^{-\alpha})$ for an $\alpha \geq 1/2$. Then, for every $\Psi \in D(\diff\Gamma(h)^{\alpha})$,
	\begin{align}
		\norm{a(f)\Psi}_{\Gamma(\mathfrak{h})} &\leq \norm{h^{-\alpha} f}_{\mathfrak{h}} \norm{\diff \Gamma(h)^{\alpha} \Psi}_{\Gamma(\mathfrak{h})} \label{eq:FirstEstimateAC}, \\
		\norm{a^*(f)\Psi}_{\Gamma(\mathfrak{h})} &\leq \norm{h^{-\alpha} f}_{\mathfrak{h}} \norm{\diff \Gamma(h)^{\alpha} \Psi}_{\Gamma(\mathfrak{h})} + \norm{f}_{\mathfrak{h}} \norm{\Psi}_{\Gamma(\mathfrak{h})}. 
		\label{eq:SecondEstimateAC}
	\end{align}
	Moreover, if $f,g \in D(h^{-\alpha/2})$, then, for every $\Psi \in D(\diff\Gamma(h)^{\alpha})$,
	\begin{align}
		\norm{(N+1)^{-\frac{1}{2}} a(f) a(g) \Psi}_{\Gamma(\mathfrak{h})} &\leq \norm{h^{-\frac{\alpha}{2}} f}_{\mathfrak{h}} \norm{h^{-\frac{\alpha}{2}} g}_{\mathfrak{h}} \norm{\diff\Gamma(h)^{\alpha} \Psi}_{\Gamma(\mathfrak{h})}.
		\label{eq:ThirdEstimate}
	\end{align}
\end{lem}

\begin{proof}
	We can assume that $h \geq 1$ is a multiplication operator on $L^2(\mathbb{R}^d)$ due to the spectral theorem. Because $\Psi^{(n)}$ is symmetric in $x = (x_1, \dots, x_n) \in \mathbb{R}^{nd}$, we obtain the following equality:
	\begin{align}
		\scp{\Psi^{(n)}}{\diff\Gamma(h^{2\alpha})\Psi^{(n)}}_{\mathfrak{h}^{\otimes n}} &= \sum_{i=1}^{n} \int_{\mathbb{R}^{nd}} h(x_i)^{2\alpha} |\Psi^{(n)}(x)|^2 \diff x \notag \\
		&= n \int_{\mathbb{R}^{nd}} h(x_n)^{2\alpha} |\Psi^{(n)}(x)|^2 \diff x.
	\end{align}
	From the above, it follows that
	\begin{align}
		\norm{a(f)\Psi^{(n)}}^2_{\mathfrak{h}^{\otimes n}} &= n \int_{\mathbb{R}^{(n-1)d}} \left| \int_{\mathbb{R}^d} \overline{f(x_n)} \Psi^{(n)}(x) \diff x_n \right|^2 \diff \hat{x}_{n} \notag \\ 
		&\leq n \norm{h^{-\alpha}f}^2_\mathfrak{h} \int_{\mathbb{R}^{nd}} h(x_n)^{2\alpha} |\Psi^{(n)}(x)|^2 \diff x \notag \\
		&= \norm{h^{-\alpha}f}^2_{\mathfrak{h}} \norm{\diff\Gamma(h^{2\alpha})^{\frac{1}{2}}\Psi^{(n)}}^2_{\mathfrak{h}^{\otimes n}}.
		\label{eq:ProofEstimatesAnnihilationCreation}
	\end{align}
	The condition $2\alpha \geq 1$ ensures that $\diff\Gamma(h^{2\alpha}) \leq \diff\Gamma(h)^{2\alpha}$. This proves \eqref{eq:FirstEstimateAC} if we sum \eqref{eq:ProofEstimatesAnnihilationCreation} over $n$. 
	
	The inequality \eqref{eq:SecondEstimateAC} for the creation operator follows from \eqref{eq:FirstEstimateAC} and
	\begin{align}
		\norm{a^*(f)\Psi}^2_{\Gamma(\mathfrak{h})} = \norm{a(f)\Psi}^2_{\Gamma(\mathfrak{h})} + \norm{f}^2_\mathfrak{h}\norm{\Psi}^2_{\Gamma(\mathfrak{h})},
	\end{align}
	which is a consequence of the commutation relation \eqref{eq:CommutatorAnniCrea}. Inequality \eqref{eq:ThirdEstimate} can be proven using similar techniques.
\end{proof}

\begin{exmp}
	If $h = \mathbb{1}$ in Lemma \ref{lem:EstimatesAnnihilationCreation}, then
	\begin{align}
		\norm{a(f)\Psi}_{\Gamma(\mathfrak{h})} &\leq \norm{f}_{\mathfrak{h}} \norm{N^{\frac{1}{2}}\Psi}_{\Gamma(\mathfrak{h})}, \label{eq:NormEstimateAnnihilation} \\
		\norm{a^*(f)\Psi}_{\Gamma(\mathfrak{h})} &\leq \norm{f}_{\mathfrak{h}} \norm{(N+1)^{\frac{1}{2}} \Psi}_{\Gamma(\mathfrak{h})} \label{eq:NormEstimateCreation}.
	\end{align}
	This shows that $D(a(f)) \supset D(N^{1/2})$ and $D(a^*(f)) \supset D(N^{1/2})$ for the domains of the annihilation and creation operators.
\end{exmp}

\subsection{Field and Weyl Operators}
\label{ssec:FieldWeylOperators}

Annihilation and creation operators are adjoint to each other. Symmetric operators like the \textbf{(Segal) field operator}
\begin{align}
	\Phi(f) = \frac{1}{\sqrt{2}} \left( a^*(f) + a(f) \right)
	\label{eq:FieldOperator}
\end{align}
and the \textbf{conjugate momentum operator}
\begin{align}
	\Pi(f) = \frac{\I}{\sqrt{2}} \left( a^*(f) - a(f) \right)
	\label{eq:ConjugateMomentum}
\end{align}
are constructed by combining annihilation and creation operators. Observe that $\Phi(f)$ and $\Pi(f)$ are $\mathbb{R}$-linear but not $\mathbb{C}$-linear in~$f$. Moreover, $\Pi(f) = \Phi(\I f)$. Considering the following proposition, we denote the self-adjoint closures of $\Phi(f)$ and $\Pi(f)$ by the very same symbol.

\begin{prop}
	For every $f\in \mathfrak{h}$, the field operator $\Phi(f)$ and the conjugate momentum operator $\Pi(f)$ are essentially self-adjoint on $\Gamma_{\mathrm{s,fin}}(\mathfrak{h})$. 
	\label{prop:FieldOperatorSelfAdjoint}
\end{prop}

\begin{proof}
	We prove the proposition for $\Phi(f)$. The proof for $\Pi(f)$ is identical. For every $\Psi \in \Gamma_{\mathrm{s,fin}}(\mathfrak{h})$ and $t\geq 0$, the series
	\begin{align}
		\sum_{n=0}^{\infty} \frac{t^n}{n!} \norm{\Phi(f)^n \Psi}_{\Gamma(\mathfrak{h})} < \infty
	\end{align}
	is finite due to \eqref{eq:NormEstimateAnnihilation} and \eqref{eq:NormEstimateCreation}. Thus, $\Gamma_{\mathrm{s,fin}}(\mathfrak{h})$ is a (dense) set of analytic vectors of $\Phi(f)$. The proposition follows from Nelson's analytic vector theorem (see \cite[Thm.~X.39, Cor.~2]{reed2}).
\end{proof}

It is immediate from the commutation relations \eqref{eq:CommutatorAnni}--\eqref{eq:CommutatorAnniCrea} for annihilation and creation operators that
\begin{align}
	[\Phi(f),\Phi(g)] &= \I \Im\scp{f}{g}_\mathfrak{h}, \\
	[\Pi(f),\Pi(g)] &= \I \Im\scp{f}{g}_\mathfrak{h}, \label{eq:CommutationRelationConjugateMomentum} \\
	[\Phi(f),\Pi(g)] &= \I \Re\scp{f}{g}_\mathfrak{h}, \label{eq:CanonicalFieldCommutator}
\end{align}
where $\Im (z)$ denotes the imaginary part of $z \in \mathbb{C}$, and $\Re (z)$ the real part. 

The \textbf{Weyl operators}
\begin{align}
	V(f) = \e^{\I\Pi(f)}, \ \ W(f) = \e^{\I\Phi(f)}
\end{align}
are unitary because the field operators $\Phi(f)$ are self-adjoint. It is not difficult to verify that the Weyl operators satisfy the following algebraic relations:
\begin{align}
	V(f)V(g) &= \e^{-\frac{\I}{2}\scp{f}{g}_\mathfrak{h}}V(f+g), \\
	W(f)W(g) &= \e^{-\frac{\I}{2}\scp{f}{g}_\mathfrak{h}}W(f+g).
\end{align}
Sometimes these relations are used as the defining property of the \textbf{Weyl algebra} $\{V(f),f\in\mathfrak{h}\} = \{W(f),f\in\mathfrak{h}\}$.  The following proposition collects more interesting properties.

\begin{prop}
	The following statements are true:
	\begin{enumerate}
		\item The Weyl operator $V(g)$ leaves the domain $D(\Phi(f))$ invariant %(i.e. $V(g)D(\Phi(f)) = D(\Phi(f))$) 
		and
		\begin{align}
			V(g) \Phi(f) V(g)^* = \Phi(f) + \Re\scp{f}{g}_\mathfrak{h}. \label{eq:WeylField}
		\end{align}
		\item If $h$ is a self-adjoint operator on $\mathfrak{h}$ and $g \in D(h)$, then $V(g)$ leaves the domain $D(\diff\Gamma(h))$ invariant %(i.e. $V(g) D(\diff\Gamma(h)) = D(\diff\Gamma(h))$) 
		and
		\begin{align}
			V(g)\diff\Gamma(h)V(g)^* = \diff\Gamma(h) + \Phi(h g) + \frac{1}{2}\scp{h g}{g}_\mathfrak{h}.
		\end{align}
		\item If $\Psi \in D(N^{1/2})$, then
		\begin{align}
			\norm{(V(f) - V(g))\Psi}_{\Gamma(\mathfrak{h})} &\leq \norm{\Pi(f-g)\Psi}_{\Gamma(\mathfrak{h})} + \frac{1}{2}|\Im \scp{f}{g}_\mathfrak{h}| \norm{\Psi}_{\Gamma(\mathfrak{h})}.
		\end{align} \label{it:PropWeylOperator3}
	\end{enumerate}
	\label{prop:PropertiesWeylOperator}
\end{prop}

\begin{proof}
	\begin{enumerate}[wide, labelwidth=!, labelindent=0pt]
		\item According to Proposition \ref{prop:FieldOperatorSelfAdjoint} and its proof, every $\Psi \in \Gamma_\mathrm{s,fin}(\mathfrak{h})$ is an analytic vector for the conjugate momentum operator $\Pi(g)$. Hence, we write $\Phi(f) V(g)^* \Psi$ as a power series, that is,
		\begin{align}
			\Phi(f)V(g)^*\Psi = \sum_{k=0}^{\infty} \frac{(-\I)^k}{k!} \Phi(f)\Pi(g)^k \Psi.
		\end{align}
		From the commutation relation \eqref{eq:CanonicalFieldCommutator} and induction, it follows that, for all $k \geq 1$,
		\begin{align}
			[\Phi(f),\Pi(g)^k] = \I k \Re\scp{f}{g}_\mathfrak{h} \Pi(g)^{k-1}.
		\end{align}
		Using this commutator in the power series expansion above, we obtain
		\begin{align}
			\Phi(f)V(g)^* \Psi = V(g)^*(\Phi(f) + \Re\scp{f}{g}_\mathfrak{h}) \Psi \label{eq:WeylProp1}.
		\end{align}
		This proves \eqref{eq:WeylField} on $\Gamma_\mathrm{s,fin}(\mathfrak{h})$. Furthermore, from \eqref{eq:WeylProp1}, we derive the following inequality for $\Psi \in \Gamma_\mathrm{s,fin}(\mathfrak{h})$:
		\begin{align}
			\norm{\Phi(f)V(g)^*\Psi}_{\Gamma(\mathfrak{h})} \leq \norm{\Phi(f) \Psi}_{\Gamma(\mathfrak{h})} + |\Re \scp{f}{g}_\mathfrak{h}| \norm{\Psi}_{\Gamma(\mathfrak{h})}.
			\label{eq:ExtensionInequality}
		\end{align}
		This inequality extends to $\Psi \in D(\Phi(f))$ because $\Gamma_\mathrm{s,fin}(\mathfrak{h})$ is dense in $D(\Phi(f))$.	Because $\Gamma_\mathrm{s,fin}(\mathfrak{h})$ is a core for $\Phi(f)$, there is, for every $\Psi \in D(\Phi(f))$, a sequence $(\Psi_k)_{k\in \mathbb{N}} \subset \Gamma_\mathrm{s,fin}(\mathfrak{h})$ such that $\norm{\Phi(f)(\Psi_k - \Psi)}_{\Gamma(\mathfrak{h})} + \norm{\Psi_k - \Psi}_{\Gamma(\mathfrak{h})} \to 0$ as $k\to\infty$. Hence, $\Phi(f)V(g)^*\Psi_k \to \Phi(f)V(g)^*\Psi$ converges due to \eqref{eq:ExtensionInequality}, and \eqref{eq:WeylProp1} extends to $D(\Phi(f))$. 
		
		Moreover, we deduce that $V(g)^*\Psi \in D(\Phi(f))$ because $\Phi(f)V(g)^*\Psi$ is well defined. If we repeat the same argument but replace $V(g)^*$ with $V(g)$, we obtain $V(g)\Psi \in D(\Phi(f))$ and conclude that $V(g)D(\Phi(f)) = D(\Phi(f))$.
		
		\item The proof is similar to that of the first item; therefore, we only sketch it. For $\Psi \in D(\diff\Gamma(h)) \cap \Gamma_\mathrm{s,fin}(\mathfrak{h})$, we expand $\diff\Gamma(h)V(g)^* \Psi$ in a power series:
		\begin{align}
			\diff\Gamma(h)V(g)^* \Psi = \sum_{k=0}^\infty \frac{(-\I)^k}{k!} \diff\Gamma(h)\Pi(g)^k \Psi.
			\label{eq:WeylSumExpansion}
		\end{align}
		From \eqref{eq:SecondQuantisationCreationCommutator} and \eqref{eq:SecondQuantisationAnnihilationCommutator}, it follows that
		\begin{align}
			[\diff \Gamma(h), \I\Pi(g)] = -\Phi(h g),
		\end{align}
		and, by induction, for every $k\geq 2$,
		\begin{align}
			[\diff \Gamma(h), \I^k\Pi(g)^k] = -k \I^{k-1} \Pi(g)^{k-1} \Phi(h g) + \frac{k(k-1)}{2} \I^{k-2} \Pi(g)^{k-2} \scp{h g}{g}_\mathfrak{h}.
		\end{align}
		Inserting this commutator into \eqref{eq:WeylSumExpansion}, we obtain
		\begin{align}
			\sum_{k=0}^\infty \frac{(-\I)^k}{k!} \diff\Gamma(h)\Pi(g)^k \Psi = V(g)^* (\diff\Gamma(h) - \Phi(h g) + \frac{1}{2}\scp{h g}{g}_\mathfrak{h})\Psi.
		\end{align}
		
		\item For $\Psi \in D(N^{1/2})$, the following equality holds true:
		\begin{align}
			&\norm{(\e^{\I\Pi(f)} - \e^{\I\Pi(g)})\Psi}_{\Gamma(\mathfrak{h})} = \norm{(\e^{-\I\Pi(g)}\e^{\I\Pi(f)} - 1)\Psi}_{\Gamma(\mathfrak{h})} \notag \\
			&\ \ \ = \norm{\int_{0}^{1} \frac{\diff}{\diff t} \left( \e^{-\I t\Pi(g)}\e^{\I t\Pi(f)} \Psi \right) \diff t}_{\Gamma(\mathfrak{h})} \notag \\
			&\ \ \ = \norm{\int_{0}^{1} \e^{-\I t\Pi(g)} \Pi(f-g) \e^{\I t\Pi(f)} \Psi \diff t}_{\Gamma(\mathfrak{h})}.
		\end{align}
		From \eqref{eq:CommutationRelationConjugateMomentum}, we derive the commutator
		\begin{align}
			[\Pi(f-g), \e^{\I t \Pi(f)}] = \e^{\I t \Pi(f)} t \Im \scp{f-g}{f}.
		\end{align}
		Inserting this commutator into the above equality, we obtain the claim:
		\begin{align}
			\norm{(\e^{\I\Pi(f)} - \e^{\I\Pi(g)})\Psi}_{\Gamma(\mathfrak{h})} &= \norm{\int_{0}^{1} \e^{-\I t\Pi(g)} \e^{\I t\Pi(f)} (\Pi(f-g) + t \Im \scp{f-g}{f}) \Psi \diff t}_{\Gamma(\mathfrak{h})} \notag \\
			& \ \ \ \leq \int_0^1 \norm{\Pi(f-g)\Psi + t \Im\scp{f-g}{f}_\mathfrak{h} \Psi}_{\Gamma(\mathfrak{h})} \diff t \notag \\
			&\ \ \ \leq \norm{\Pi(f-g)\Psi}_{\Gamma(\mathfrak{h})} + \frac{1}{2}|\Im \scp{f}{g}_\mathfrak{h}| \norm{\Psi}_{\Gamma(\mathfrak{h})}.
		\end{align}
	\end{enumerate}
\end{proof}

\section{Nelson Model with Constant Coefficients} 
\label{sec:NelsonModelConstantCoefficients}

After introducing the Fock space formalism, we can now describe the Nelson model mathematically. As explained earlier, the Nelson model describes the coupling of nonrelativistic quantum particles to a relativistic field of bosons. We only discuss the case of a single nonrelativistic particle (referred to simply as 'particle' in the following text) because the general case offers no more difficulties rather than complicating the notation. Furthermore, we assume that the bosonic field mass $m > 0$ is strictly positive to avoid infrared divergences.

The underlying Hilbert space $\mathfrak{H}$ of the Nelson model is a tensor product space. It consists of the Hilbert space $\mathfrak{K} = L^2(\mathbb{R}^3, \diff X)$ for the particle and the symmetric Fock space $\Gamma_\mathrm{s}(\mathfrak{h})$ with $\mathfrak{h} = L^2(\mathbb{R}^3, \diff x)$ for the field:
\begin{align}
	\mathfrak{H} = \mathfrak{K} \otimes \Gamma_\mathrm{s}(\mathfrak{h}) = L^2(\mathbb{R}^3, \diff X) \otimes \Gamma_\mathrm{s}(L^2(\mathbb{R}^3, \diff x)). 
\end{align} 
The $n$-boson subspaces are denoted by $\mathfrak{H}^{(n)} = \mathfrak{K} \otimes \mathfrak{h}^{\otimes_\mathrm{s}n}$.

The free particle Hamiltonian $-\Delta_X$ and the free field Hamiltonian $\diff \Gamma(\omega)$ with dispersion relation $\omega = \sqrt{-\Delta + m^2}$ (see Example~\ref{exmp:FreeFieldHamiltonian}) combine into the \textbf{free Nelson Hamiltonian}:
\begin{align}
	H_{0} = -\Delta_X \otimes \mathbb{1} + \mathbb{1} \otimes \diff \Gamma(\omega).
\end{align}
Because it is the sum of two commuting self-adjoint operators, $H_0$ is self-adjoint on $D(H_0) = D(-\Delta_X \otimes \mathbb{1}) \cap D(\mathbb{1} \otimes \diff\Gamma(\omega))$.

If we formally perturb the free Nelson Hamiltonian by the field operator $\Phi(\omega^{-\frac{1}{2}} \delta_X)$ defined in \eqref{eq:FieldOperator}, where $\delta_X$ is the Dirac delta centred at the particle position $X$, we  obtain the formal \textbf{Nelson Hamiltonian}:
\begin{align}
	H = -\Delta_X \otimes \mathbb{1} + \mathbb{1} \otimes \diff \Gamma(\omega) + \Phi(\omega^{-\frac{1}{2}} \delta_X).
	\label{eq:NelsonHamiltonian}
\end{align}
However, the formal Nelson Hamiltonian is not a well-defined operator. The problem is that the form factor $\omega^{-\frac{1}{2}} \delta_X$ is not an element of $\mathfrak{h} = L^2(\mathbb{R}^3, \diff x)$. If the creation operator $a^*(\omega^{-\frac{1}{2}}\delta_X)$ is defined as in \eqref{def:CreationOperator} even though $\omega^{-\frac{1}{2}}\delta_X \notin \mathfrak{h}$, then its natural domain is $D(a^*(\omega^{-\frac{1}{2}}\delta_X)) = \{ 0 \}$. Hence, the field operator $\Phi(\omega^{-\frac{1}{2}}\delta_X)$ is not densely defined, and $H$ cannot be self-adjoint as a sum of self-adjoint operators.

\begin{rem}
	The annihilation operator is typically easier to handle because the scalar product $\scp{\omega^{-\frac{1}{2}}\delta_X}{f}_\mathfrak{h}$ in the definition of the annihilation operator may be well defined for a dense set of $f\in \mathfrak{h}$, even if $\omega^{-\frac{1}{2}}\delta_X \notin \mathfrak{h}$. 
	If $a(\omega^{-\frac{1}{2}}\delta_X)$ were defined on $D(H_0^{1/2})$, a common method of interpreting \eqref{eq:NelsonHamiltonian} would be in the sense of quadratic forms. However, this is not the case in the Nelson model.
\end{rem}

The Nelson model has yet to be regularised. One possibility is to replace the singular delta function with a smoother function $\rho_\Lambda$ which depends on the scale $\Lambda$ in such a way that $\rho_{\Lambda}$ approaches the delta function for large $\Lambda$. We realise this as follows. Assume that $0 \leq \rho \in \mathcal{S}(\mathbb{R}^3)$ is a nonnegative symmetric Schwartz function with $\int_{\mathbb{R}^3} \rho = 1$ and set
\begin{align}
	\rho_{\Lambda,X}(x) = \Lambda^3 \rho(\Lambda(x-X)), \ \rho_{\Lambda}(x) = \rho_{\Lambda,0}(x).
	\label{eq:CutOffFunction}
\end{align}
Then, for every $s<-3/2$, $\rho_{\Lambda,X}$ converges in $H^{s}(\mathbb{R}^3)$ to $\delta_X$ uniformly in $X$:
\begin{align}
	\norm{\rho_{\Lambda,X}-\delta_X}_{H^{s}}^2 &= \norm{(1+|\cdot|^2)^{\frac{s}{2}}(\widehat{\rho}_{\Lambda} -1)}_{L^2}^2 \notag \\
	&= \int_{\mathbb{R}^3} \frac{\left|\widehat{\rho} \left(\xi/\Lambda\right)-1\right|^2 }{(1+|\xi|^2)^{-s}} \diff \xi \stackrel{\Lambda\to\infty}{\longrightarrow} 0.
	\label{eq:ConvergenceCutOffToDelta}
\end{align}
The convergence in the last step is justified because $\widehat{\rho}(\xi/\Lambda) \to 1$ pointwise and because $s<-3/2$ is sufficiently small to apply the dominated convergence theorem.

The \textbf{cut-off Hamiltonian} $H_\Lambda$ is obtained by replacing $\delta_X$ with $\rho_{\Lambda,X}$:
\begin{align}
	H_{\Lambda} = -\Delta_X \otimes \mathbb{1} + \mathbb{1} \otimes \diff \Gamma(\omega) + \Phi(\omega^{-\frac{1}{2}}\rho_{\Lambda,X}).
\end{align} 
Physically, this replacement suppresses energies above a certain scale $\Lambda$, which are beyond the validity of the model. Therefore, $\rho_\Lambda$ is called an \textbf{ultraviolet cut-off}. Equivalently, the cut-off can be interpreted as assigning a nonzero radius to particles that are usually point-like objects. 

The following proposition proves that $H_{\Lambda}$ is a well-defined self-adjoint operator.

\begin{prop}
	The cut-off Hamiltonian $H_\Lambda$ is self-adjoint on $D(H_\Lambda) = D(H_0)$.
	\label{prop:CutOffHamiltonianSelfAdjoint}
\end{prop}

\begin{proof}
	Observe that the number operator $N$ is dominated by $H_0/m$ due to $\diff\Gamma(\mathbb{1}) \leq \diff\Gamma(\omega)/m$. Furthermore, \eqref{eq:NormEstimateAnnihilation} and \eqref{eq:NormEstimateCreation} imply that $D(N^{1/2}) \subset D(\Phi(\omega^{-\frac{1}{2}}\rho_{\Lambda,X}))$. Thus, $D(H_0) \subset D(H_0^{1/2}) \subset D(N^{1/2}) \subset D(\Phi(\omega^{-\frac{1}{2}}\rho_{\Lambda,X}))$.
	
	According to the Kato--Rellich theorem, it suffices to prove that the perturbation $\Phi(\omega^{-\frac{1}{2}}\rho_{\Lambda,X})$ is relatively bounded with respect to $H_0$ and that the relative bound is strictly smaller than 1. Again, from \eqref{eq:NormEstimateAnnihilation} and \eqref{eq:NormEstimateCreation}, it follows that, for every $\epsilon > 0$,
	\begin{align}
		\norm{\Phi(\omega^{-\frac{1}{2}}\rho_{\Lambda,X})\Psi}_{\mathfrak{H}} &\leq \sup_{X\in \mathbb{R}^3} \norm{\omega^{-\frac{1}{2}}\rho_{\Lambda,X}}_\mathfrak{h} \norm{(N+1)^{\frac{1}{2}} \Psi}_{\mathfrak{H}} \notag \\
		&\leq \sup_{X\in \mathbb{R}^3} \norm{\omega^{-\frac{1}{2}}\rho_{\Lambda,X}}_\mathfrak{h} (\epsilon \norm{H_0\Psi}_{\mathfrak{H}} + C_\epsilon \norm{\Psi}_{\mathfrak{H}}).
		\label{eq:EstimatePertubation}
	\end{align}	
	The factor $\sup_X \norm{\omega^{-\frac{1}{2}}\rho_{\Lambda,X}}_\mathfrak{h} = \norm{\omega^{-\frac{1}{2}}\rho_{\Lambda}}_\mathfrak{h}$ is finite due to the cut-off $\rho_\Lambda$. In the second inequality, we used $N\leq H_0/m$ and Young's inequality. Choosing an $\epsilon$ that is small enough yields the desired bound.
\end{proof}

The cut-off $\rho_\Lambda$ ensures that the cut-off Hamiltonian is well defined and self-adjoint. However, the choice of a particular cut-off is arbitrary, and it is expected that relevant physical quantities (e.g. energy levels) heavily depend on how the cut-off is chosen.

The actual problem we face in the Nelson model without the cut-off is that the Nelson Hamiltonian has infinite vacuum energy, whereas the cut-off Hamiltonian has finite vacuum energy $-E_\Lambda$. It makes sense to remove the infinite vacuum energy as we are usually interested in energy differences rather than absolute values. Nelson \cite{nelson1964} has proven the following theorem, which provides a decent candidate for a well-defined Nelson Hamiltonian.

\begin{thm}[Nelson]
	There are constants $E_\Lambda$ such that $H_\Lambda + E_\Lambda$ converges in the strong resolvent sense to a self-adjoint bounded from below operator $H$ as $\Lambda \to \infty$. 
	\label{thm:Nelson}
\end{thm}

A sequence $(A_n)_{n\in\mathbb{N}}$ of self-adjoint operators on $\mathfrak{H}$ converges in the strong resolvent sense to a self-adjoint operator $A$ if $(A_n - \lambda)^{-1}\Psi \to (A - \lambda)^{-1}\Psi$ for every (or equivalently for one) $\lambda \in \mathbb{C}\backslash\mathbb{R}$ and every $\Psi \in \mathfrak{H}$.

The operator $H$ from Nelson's theorem is the \textbf{renormalised Nelson Hamiltonian}. We prove Nelson's theorem in Chapter \ref{ch:RemovalCutOff} for the Nelson model on static spacetimes, which we introduce in the next section.

\begin{rem}
	A sequence of self-adjoint operators converges in the strong resolvent sense if and only if the associated unitary groups converge strongly (see \cite[Thm.~VIII.21]{reed1}). That is, the convergence in Theorem \ref{thm:Nelson} is equivalent to
	\begin{align}
		\e^{\I t(H_\Lambda+E_\Lambda)} \Psi \to \e^{\I t H}\Psi
	\end{align}
	as $\Lambda \to \infty$ for all $t \in \mathbb{R}$ and $\Psi \in \mathfrak{H}$.	Although  $\e^{\I t H_\Lambda}\Psi$ diverges as $\Lambda \to \infty$, the state determined by it converges in the sense that
	\begin{align}
		\scp{\e^{\I t H_\Lambda}\Psi}{A \e^{\I t H_\Lambda}\Psi} \to \scp{\e^{\I t H}\Psi}{A \e^{\I t H}\Psi}
	\end{align}
	for any bounded self-adjoint operator $A$.
\end{rem}

\begin{rem}
	Recently, Griesemer and Wünsch \cite{griesemer2018} proved that the convergence of $H_\Lambda + E_\Lambda$ in Theorem \ref{thm:Nelson} also holds in the norm resolvent sense.
\end{rem}

In summary, we defined the formal Nelson Hamiltonian in \eqref{eq:NelsonHamiltonian}, realised that the Hamiltonian is not well defined and that putting a cut-off on the interaction solves the problem of self-adjointness. Then, we argued that the cut-off can be removed by extracting the divergent vacuum energy $-E_\Lambda$ of the system.

It remains an open question whether we have any chance of describing the renormalised Hamiltonian $H$ because many interesting properties of the cut-off Hamiltonians $H_\Lambda$ are not expected to survive the limit $\Lambda \rightarrow \infty$. One example is that $D(H_\Lambda) = D(H_0)$ for all $\Lambda < \infty$, but $D(H_0) \cap D(H) = \{0\}$. This property was conjectured by Nelson and proven by Griesemer and Wünsch \cite{griesemer2018} without describing $D(H)$ explicitly. 

In Chapter \ref{ch:IBC}, we obtain  an explicit description of the renormalised Nelson Hamiltonian and its domain due to the novel \textbf{IBC  method} introduced by Teufel and Tumulka \cite{teufel2016, teufel2020}. Lampart and Schmidt \cite{lampart2019} already demonstrated that the IBC method applies to the Nelson-type Hamiltonians. The present work aims to generalise these recently obtained results to the Nelson model on static spacetimes, also known as the Nelson model with variable coefficients.

\section{Nelson Model with Variable Coefficients}
\label{sec:NelonModelVariableCoeff}

The Nelson model with variable coefficients is an extension of the Nelson model, where the Laplace operator is replaced with a second-order partial differential operator with variable coefficients. This extension describes the Nelson model on static spacetimes, which is motivated by the fact that the correct generalisation of the Laplace operator to pseudo-Rie\-mann\-ian manifolds is the \textbf{Laplace--Beltrami operator}. The Nelson Hamiltonian with variable coefficients was introduced in \cite{gerard2011, gerard2012}, and we adopt some of the notation used in these papers.

We begin with redefining the \textbf{free particle Hamiltonian} $K_0$ by
\begin{align}
	K_0 = -\sum_{1\leq j,k\leq 3} \partial_{X_j} g^{jk}(X) \partial_{X_k} = -\partial_X \cdot g(X) \partial_X,
	\label{eq:K0}
\end{align}
where the matrix $g(X) := [g^{jk}(X)]$ is assumed to be symmetric for every $X\in \mathbb{R}^3$, which ensures that $K_0$ is symmetric. However, to obtain self-adjointness, we must impose further regularity conditions on $g(X)$. Later, we assume that the matrix coefficients $g^{jk}$ are smooth, but for now it suffices to assume that they are Lipschitz continuous. Additionally, we impose the following \textbf{uniform ellipticity} condition on the matrix $g(X)$:
\begin{align}
	C_0 \leq g(X) \leq C_1,\ C_0 > 0. 
	\label{eq:EllipticityCondition}
\end{align}
The assumption of ellipticity is useful for applying well-known results from the theory of elliptic operators. Furthermore, it guarantees that $K_0$ is 'similar' to the Laplacian, which is easier to handle than $K_0$.

\begin{prop}
	Assume that, for every $X\in\mathbb{R}^3$, the matrix $g(X)$ is symmetric. Furthermore, assume that $g$ is uniform elliptic and that its matrix coefficients are Lipschitz continuous. Then $K_0$ is self-adjoint on $D(K_0) = D(-\Delta)$.
	\label{prop:K0SelfAdjoint}
\end{prop}

\begin{proof}
	We define a quadratic form $q:H^1(\mathbb{R}^3) \times H^1(\mathbb{R}^3) \to \mathbb{C}$ by
	\begin{align}
		q(u,v) = \sum_{1\leq j,k\leq 3} \scp{\partial_{X_j} v}{g^{jk}(X) \partial_{X_k} u} = \scp{\partial_X u}{g(X)\partial_X v}.
	\end{align}
	The strategy of this proof is to associate a self-adjoint operator $B$ with the quadratic form $q$ and prove that $B=K_0$. Clearly, $q$ is bounded from below by 0:
	\begin{align}
		q(u,u) = \scp{\partial_X u}{g(X)\partial_X u} \geq C_0 \norm{\partial_X u}^2 \geq 0.
	\end{align}
	Furthermore, $q$ is closed, which follows from the fact that the quadratic form norm
	\begin{align}
		\norm{u}_q^2 = q(u,u) + \norm{u}^2
	\end{align}
	is equivalent to the Sobolev $H^1$-norm by condition \eqref{eq:EllipticityCondition}. Thus, by a well-known theorem in the theory of quadratic forms (see \cite[Thm.~2.14]{teschl2014}), a unique self-adjoint operator $B\geq 0$ exists  with form domain $Q(B) = H^1(\mathbb{R}^3)$ such that $q(u,u) = \norm{B^{\frac{1}{2}}u}^2$. More explicitly, $B$ is given by
	\begin{align}
		D(B) &= \{u \in H^1(\mathbb{R}^3) \mid \exists f \in L^2(\mathbb{R}^3), q(u,v) = \scp{f}{v} \ \forall v \in H^1(\mathbb{R}^3)\} \\
		Bu &= f.
	\end{align}	
	The condition $q(u,v) = \scp{f}{v}$ for all $v \in H^1(\mathbb{R}^3)$ is equivalent to calling $u\in H^1(\mathbb{R}^3)$ a weak solution of the partial differential equation $K_0u = f$. From the elliptic regularity theorem (see \cite[Thm.~17.2.7]{hoermander1985} and its proof) we deduce that $u \in H^2(\mathbb{R}^3)$. To apply \cite[Thm.~17.2.7]{hoermander1985}, $g(X)$ must be symmetric and its matrix coefficients Lipschitz continuous.
	
	Hence, $D(B)$ is contained in $H^2(\mathbb{R}^3)$. As the other inclusion is trivial (for $u\in H^2(\mathbb{R}^3)$, choose $f=K_0u$), $D(B) = H^2(\mathbb{R}^3)$. On $H^2(\mathbb{R}^3)$, the action of $B$ and $K_0$ coincide; thus, $B=K_0$.
\end{proof}

We add an exterior particle potential $W(X)$ to the free particle Hamiltonian $K_0$. This potential should be chosen such that $K_0 + W(X)$ is essentially self-adjoint. The self-adjoint closure of this operator is denoted by $K$. %However, in Chapter \ref{ch:RemovalCutOff}, we observe  that the exterior particle potential $W(X)$ is not interesting in the Nelson model with variable coefficients because the renormalised Hamiltonian is determined up to an exterior particle potential.

The free field Hamiltonian on static spacetimes is defined similarly  to the free particle Hamiltonian. We set
\begin{align}
	h_0 &= -\sum_{1\leq j,k \leq 3} \partial_{x_j} g^{jk}(x) \partial_{x_k} = -\partial_x \cdot g(x)\partial_x, \\
	h &= h_0 + \mu(x)^2.
\end{align}
We allow for a position-dependent field mass $0\leq \mu \in L^\infty(\mathbb{R}^3)$. Clearly, $h$ is a self-adjoint operator on $D(h) = H^2(\mathbb{R}^3)$ by the previous proposition and the fact that $\mu$ is bounded. 

The \textbf{dispersion relation} is $\omega = h^{1/2}$,  and the \textbf{free field Hamiltonian} is the second quantisation $\diff\Gamma(\omega)$. We denote by $m$ the infimum of the spectrum of $\omega$. If $m>0$, we regard the model as \textbf{massive}, and \textbf{massless} otherwise. As in the Nelson model with constant co\-ef\-fi\-cients, we assume that the Nelson model with variable coefficients is massive to avoid infrared divergences.

Altogether, the Nelson Hamiltonian with variable co\-ef\-fi\-cients and cut-off is given by
\begin{align}
	H_\Lambda = K \otimes \mathbb{1} + \mathbb{1} \otimes \diff \Gamma(\omega) + \Phi(\omega^{-\frac{1}{2}}\rho_{\Lambda,X}).
	\label{eq:CutOffHamiltonianStaticST}
\end{align}
The cut-off function $\rho \in \mathcal{S}(\mathbb{R}^3)$ is chosen as in the previous section.

What makes the Nelson model with variable coefficients more challenging to analyse than the Nelson model with constant coefficients is that  the free particle Hamiltonian $K_0$ or the operator $h_0$ cannot be replaced by a multiplication operator in Fourier space. Instead, we rely on the \textbf{pseudo-differential calculus}, which is introduced in the next chapter.

\chapter{Pseudo-differential Calculus}
\label{ch:Pseudors}

The pseudo-differential calculus defines an algebraic correspondence between \textbf{classical} and \textbf{quantum observables} known as \textbf{quantisation}. The classical observables in this calculus are represented by \textbf{symbols} that are smooth functions defined on the physical \textbf{phase space}, a subset of $\mathbb{R}^d \times \mathbb{R}^d$. Each symbol defines a pseudo-differential operator, which is the corresponding quantum observable. However, because of the \textbf{ordering ambiguity} in quantum mechanics (classical observables commute, but their quantum mechanical counterparts generally do not), the correspondence between classical and quantum observables is ambiguous. Instead, different quantisation schemes exist.

From a mathematical viewpoint, pseudo-differential operators provide an extended class of differential operators. To understand this, let $p(D)$ be a linear differential operator of order $m$, that is,
\begin{align}
	p(D)u(x) = \sum_{|\alpha| \leq m} g_\alpha D_x^{\alpha}u(x),
\end{align}
where $g_\alpha \in \mathbb{C}$, $D_x = -\I \partial_x$, and the sum is taken over all multi-indices $\alpha \in \mathbb{N}^d$ with $|\alpha| \leq m$. The operator $p(D)$ can alternatively be written as a composition of a Fourier transformation, a multiplication by the symbol
\begin{align}
	p(\xi) = \sum_{|\alpha| \leq m} g_\alpha \xi^{\alpha}
\end{align}
and an inverse Fourier transformation:
\begin{align}
	p(D)u(x) &= \frac{1}{(2\pi)^d} \int_{\mathbb{R}^d} \int_{\mathbb{R}^d} \e^{\I\scp{x-y}{\xi}} p(\xi) u(y) \diff y \diff \xi	\notag\\ 
	&= \frac{1}{(2\pi)^d} \int_{\mathbb{R}^d} \e^{\I\scp{x}{\xi}} p(\xi) \widehat{u}(\xi) \diff \xi.
\end{align}
In Fourier space, the initially complicated differential operator $p(D)$ is a simple multi\-pli\-cation by the symbol $p(\xi)$. Manipulations of the differential operator (e.g. taking powers, inverses, etc.) are equivalent to manipulations of its symbol. 

A pseudo-differential operator $\mathrm{Op}(p)$ is an operator for which the symbol $p = p(x,\xi)$ can depend on two phase space variables, $x$ and $\xi$, representing position and momentum:
\begin{align}
	\mathrm{Op}(p) u(x) &= \frac{1}{(2\pi)^d} \int_{\mathbb{R}^d} \int_{\mathbb{R}^d} \e^{\I\scp{x-y}{\xi}} p(x,\xi) u(y) \diff y \diff \xi.
\end{align} 
The main difficulty in the pseudo-differential calculus is that manipulations of the symbol no longer correspond to the same manipulations of the pseudo-differential operator (e.g. the inverse symbol (if existent) is not necessarily the symbol of the inverse operator). Nevertheless, we demonstrate that the correspondence holds at least in an approximate sense, and we obtain estimates in how far the corres\-pon\-dence fails.

The analysis of pseudo-differential operators is a vast area of mathematics with many different applications. It is impossible to provide a complete overview of this topic, but we explain essential ideas and techniques relevant to the present work. Several excellent references in the literature provide a more detailed exposition. We follow \cite{dimassi1999, hoermander1985, martinez2002, raymond1991, zworski2012} in this chapter.

\section{Symbols}
\label{sec:Symbols}

Generally, a symbol is a smooth function satisfying additional estimates. In the literature, different classes of symbols are more or less useful in different applications. To define appropriate symbol classes for the present work, we first need the notion of an order function.

\begin{defn}
	A smooth function $M:\mathbb{R}^k \to (0,\infty)$ is an \textbf{order function} if constants $C,m$ exist such that $M(z) \leq C\jap{z-w}^m M(w)$ for all $z,w \in \mathbb{R}^{k}$, where $\jap{z} = \sqrt{1+|z|^2}$ is the \textbf{Japanese bracket}.
	\label{def:OrderFunction}
\end{defn}

\begin{ntn}
	In this chapter, we write $z$ or $w$ for an element in $\mathbb{R}^k$ (unless stated otherwise). If $k=2d$, we replace $z$ with $(x,\xi) \in \mathbb{R}^d \times \mathbb{R}^d$ and call $x$ the \textbf{position variable} and $\xi$ the \textbf{momentum variable}. Sometimes we also deal with $k=3d$. In this case, we write $z=(x,y,\xi)$, and $y$ takes the role of an integrated variable.
\end{ntn}

\begin{exmp}
	Clearly, $M\equiv 1$ and $M(z) = \jap{z}$ are order functions, with the latter being a consequence of Peetre's inequality (see Theorem \ref{thm:Peetre}). Moreover, if $k=2d$, then $M(x,\xi) = \jap{x}^\rho \jap{\xi}^\sigma$ is an order function for any $\rho,\sigma \in \mathbb{R}$.
\end{exmp}

For a given order function $M$ and a Riemannian metric $g$, we define the symbol class $S_k(M,g)$. A Riemannian metric on $\mathbb{R}^k$ is a smooth map $g:\mathbb{R}^{k} \ni z \mapsto g_z$ with values in positive definite quadratic forms on $\mathbb{R}^{k}$. Further regularity conditions must be imposed on the Riemannian metric $g$ to obtain a well-behaved theory of pseudo-differential operators (cf.~\cite{bony2013} for a detailed discussion). In the following text, we call these Riemannian metrics \textbf{admissible}.

\begin{defn}
	Let $M$ be an order function and $g$ be an admissible metric on $\mathbb{R}^{k}$. The \textbf{symbol class} $S_k(M,g)$ is the set of all $a\in C^\infty(\mathbb{R}^{k})$ that satisfy the following bound for all $p\in\mathbb{N}$ and $v_1,\dots,v_p \in \mathbb{R}^{k}$:
	\begin{align}
		\left|\prod_{i=1}^{p} (v_i \cdot \partial_{z})a(z)\right| \leq C_p M(z) \prod_{i=1}^p |g_{z}(v_i)|^\frac{1}{2}.
	\end{align}
	We write $S_k(M,g) = S(M,g)$ if the dimension $k$ is irrelevant or clear from the context.
	\label{def:SMg}
\end{defn}

The space $S(M,g)$ is a Fréchet space with semi-norms given by the best constants $C_{p}$ satisfying the above bound. 

It is easy to verify that the product of two symbols is, again, a symbol:

\begin{prop}
	If $a_1 \in S(M_1,g)$ and $a_2 \in S(M_2,g)$, then $a_1 a_2 \in S(M_1M_2,g)$.
\end{prop}

We rarely work with the symbol class $S(M,g)$ in its full generality but restrict our attention to the following two examples.

\begin{defn}
	Let $M$ be an order function. Then $S_k(M) = S_k(M,\diff z^2)$ is the set of all $a\in C^\infty(\mathbb{R}^k)$ satisfying the bound $|\partial_z^\alpha a(z)| \leq C_{\alpha} M(z)$ for every multi-index $\alpha \in \mathbb{N}^k$.
\end{defn}

In (quantum) mechanical applications, derivatives in the momentum variable usually improve the order of a symbol $a\in S_{2d}(M)$. Therefore, we also define the following subclass of $S(M)$.

\begin{defn}
	Let $m\in \mathbb{R}$. Then $S^m = S_{2d}(\jap{\xi}^m,\diff x^2 + \jap{\xi}^{-2}\diff\xi^2) \subset S_{2d}(\jap{\xi}^m)$ is the set of all $a\in C^\infty(\mathbb{R}^d \times \mathbb{R}^d)$ that satisfy the following bound for all multi-indices $\alpha, \beta \in \mathbb{N}^d$:
	\begin{align}
		|\partial_x^\alpha \partial_\xi^\beta a(x,\xi)| \leq C_{\alpha \beta} \jap{\xi}^{m-|\beta|}
	\end{align}
	We call $m$ the \textbf{order} of the symbol class, and write $S^{-\infty}=\bigcap_m S^m$ and $S^{\infty} = \bigcup_{m} S^m$.
\end{defn}

The symbol classes $S^m$ have the practical advantage that we can perform com\-pu\-ta\-tions in $S^n$ modulo $S^l$ because $S^l \subset S^n$ for $l\leq n$. It is also possible to expand elements of $S^m$ in \textbf{asymptotic sums}. To assign a precise meaning to asymptotic sums of symbols, two symbols $a,b \in S^\infty$ are called \textbf{asymptotically equivalent}, written as $a\sim b$, if $a-b\in S^{-\infty}$. A surprising feature is that although a formal series of the type $\sum_{j} a_j$, where the order of $a_j$ is decreasing to $-\infty$ for larger $j$, has no reason to be convergent, a symbol that is asymptotically equivalent to this series can always be found.

\begin{thm}
	Let $a_j \in S^{m-j}$ for any $j \in \mathbb{N}$. Then a symbol $a \in S^m$ uniquely determined up to asymptotic equivalence (i.e. up to a symbol in $S^{-\infty}$) exists such that
	\begin{align}
		a - \sum_{j<l} a_j \in S^{m-l}
		\label{eq:AsymptoticExpansion}
	\end{align}
	for every $l \in \mathbb{N}$.
	\label{thm:AsymptoticExpansion}
\end{thm}

\begin{proof}[Sketch of the Proof]
	If $a$ exists, it is uniquely determined by \eqref{eq:AsymptoticExpansion} up to a symbol in $S^{-\infty}$. To prove the existence of $a$, we define a \textbf{resummation} of the formal series $\sum_{j} a_j$. 
	
	We take a $\chi \in C_0^\infty(\mathbb{R}^d)$, which is equal to $1$ in a neighbourhood of $0$, and set $A_j(x,\xi) = (1-\chi(\epsilon_j\xi))a_j(x,\xi)$, where $\epsilon_j \to 0$ rapidly. Then $A_j \sim a_j$ because $A_j(x,\cdot) = a_j(x,\cdot)$ outside a compact set. Clearly, $a(x,\xi) = \sum_j A_j(x,\xi)$ defines an element in $C^\infty(\mathbb{R}^{2d})$ because the sum is finite near any fixed $\xi_0 \in \mathbb{R}^{d}$. One can verify that the symbol $a$ defined in this way indeed satisfies \eqref{eq:AsymptoticExpansion}. We refer to \cite[Thm. 18.1.3]{hoermander1985} for the full proof.
\end{proof}

\begin{exmp}
	The symbol $\e^{\I \scp{D_x}{D_\xi}} a$ with $a \in S^{\infty}$ has the following asymptotic expansion:
	\begin{align}
		\e^{\I \scp{D_x}{D_\xi}} a(x,\xi) &\sim \sum_j \frac{1}{j!} (\I \scp{D_x}{D_\xi})^j a(x,\xi) = \sum_{\alpha} \frac{1}{\alpha!} \partial_\xi^\alpha D_x^\alpha a(x,\xi). 
		\label{eq:AsymptoticExpansionExponential}
	\end{align}
	A resummation of the sum on the right-hand side is found in \cite[Thm.~18.1.7]{hoermander1985}.
\end{exmp}

\begin{rem}
	The nonzero symbol of the highest order in the asymptotic sum $\sum_{j} a_j$ is called the \textbf{principal symbol}. Typically, it is easy to determine the principal symbol (e.g., in \eqref{eq:AsymptoticExpansionExponential}, the principal symbol of $\e^{\I \scp{D_x}{D_\xi}} a$ is simply $a$), whereas the lower order terms are usually harder to compute.
\end{rem}

It is also possible to define asymptotic expansions for general symbols in $S(M)$ at the cost of introducing an additional \textbf{order parameter} $h$. In this case, we call $a \in S(M)$ asymptotically equivalent to the formal series $\sum_j h^j a_j$ with $a_j \in S(M)$ if
\begin{align}
	\left|\partial^\alpha \left(a-\sum_{j<l} h^j a_j\right)\right| \leq C_{l\alpha} h^l M
\end{align}
for all $l\in\mathbb{N}$, $\alpha \in \mathbb{N}^k$ and all $h>0$ that are sufficiently small. With a similar resummation argument as in the proof of Theorem \ref{thm:AsymptoticExpansion}, we associate each formal series $\sum_j h^j a_j$ with a symbol $a \in S(M)$ that is asymptotically equivalent to this series.

%\begin{exmp}
%	Similar to the previous example, for $a \in S(M)$,
%	\begin{align}
%		\e^{\I h \scp{D_x}{D_\xi}} a(x,\xi) \sim \sum_j \frac{h^j}{j!} (\I \scp{D_x}{D_\xi})^j a(x,\xi).
%	\end{align}
%	This asymptotic expansion is an immediate consequence of the \textbf{stationary phase theorem}.
%\end{exmp}

It is useful to have both concepts of asymptotic sums. For simplicity, we restrict our attention to symbols in $S^\infty$ when we construct asymptotic sums in the following sections.

\section{Pseudo-differential Operators}
\label{sec:Pseudors}

Having defined appropriate classes of symbols, we introduce pseudo-differential operators. We assume that $M$ is an order function on $\mathbb{R}^{3d}$ that depends only on the momentum variable (i.e. $M(x,y,\xi) = M(\xi)$).\footnote{This assumption is not strictly necessary for defining pseudo-differential operators, but it simplifies the discussion. The general case is of no use for the present work.} In this case, an $m \in \mathbb{R}$ exists such that $M(x,\xi) \leq C M(0) \jap{\xi}^m$. 
\begin{defn}
	Let $M$ be an order function specified as above and $a \in S_{3d}(M)$. Then we call the operator $\Op(a)$ defined on $C_\mathrm{c}^\infty(\mathbb{R}^d)$ by
	\begin{align}
		\mathrm{Op}(a) u(x) = \frac{1}{(2\pi)^d} \int_{\mathbb{R}^d} \int_{\mathbb{R}^d} \e^{\I\scp{x-y}{\xi}} a(x,y,\xi) u(y) \diff y \diff \xi
		\label{eq:Pseudor}
	\end{align}
	a \textbf{pseudo-differential operator of order $m$}.
	\label{def:Pseudor}
\end{defn}
We defined $\Op(a)$ on $C_\mathrm{c}^\infty(\mathbb{R}^d)$ to make sense of \eqref{eq:Pseudor} as an \textbf{oscillatory integral}. To verify that the integral \eqref{eq:Pseudor} is well defined, we assume first that $m<-d$. Then,
\begin{align}
	|a(x,y,\xi)| \leq C M(x,y,\xi) \leq C \jap{\xi}^m,
\end{align}
and the integral on the right-hand side of \eqref{eq:Pseudor} is absolutely convergent. If $m\geq -d$, we exploit the oscillatory behaviour of $\e^{\I\scp{x-y}{\xi}}$. We observe the following:
\begin{align}
	L \e^{\I\scp{x-y}{\xi}} = \e^{\I\scp{x-y}{\xi}}, \ L = \frac{1-\xi D_y}{1+|\xi|^2}.
\end{align}
Thus, if $L^T$ is the transpose of $L$, we obtain via integration by parts, for any $k\in \mathbb{N}$,
\begin{align}
	\mathrm{Op}(a) u(x) = \frac{1}{(2\pi)^d} \int_{\mathbb{R}^d} \int_{\mathbb{R}^d} \e^{\I\scp{x-y}{\xi}} (L^T)^k( a(x,y,\xi) u(y) ) \diff y \diff \xi \equiv I_k(a) u(x).
\end{align}
If $k$ is chosen sufficiently large, the integral $I_k(a)u(x)$ is absolutely convergent; hence, $\Op(a)$ is well defined on $C_\mathrm{c}^\infty(\mathbb{R}^d)$.

\begin{thm}
	If $a\in S(M)$, then $\Op(a)$ defines a continuous map from $C_\mathrm{c}^\infty(\mathbb{R}^d)$ to $C^\infty(\mathbb{R}^d)$.
	\label{thm:ContinuityPseudors}
\end{thm}

\begin{proof}
	We take $l\in\mathbb{N}$ and $k=[m] + d + l + 1$, where $[m]$ is the smallest integer such that $m\leq [m]$. Then $\Op(a)u = I_k(a) u$ has $l$ continuous derivatives by the dominated convergence theorem. As $l$ is arbitrary, $\Op(a)u \in C^\infty(\mathbb{R}^d)$.
	
	Furthermore, it is not difficult to bound $\partial^\alpha \Op(a)u$ by a finite sum of semi-norms of $u$ (see \cite[Thm. 2.4.3]{martinez2002}). This proves that $\Op(a): C_\mathrm{c}^\infty(\mathbb{R}^d) \to C^\infty(\mathbb{R}^d)$ is continuous.
\end{proof}

The continuity of $\Op(a)$ on $C_\mathrm{c}^\infty(\mathbb{R}^d)$ and the Schwartz kernel theorem (see Theorem \ref{thm:SchwartzKernel}) imply the existence of a unique distributional kernel $K$ such that $\scp{\Op(a)u}{v} = K(u\otimes v)$ for all $u,v \in C_\mathrm{c}^\infty(\mathbb{R}^d)$. It is suggestive to denote the kernel $K$ by
\begin{align}
	K(x,y) = \frac{1}{(2\pi)^d} \int_{\mathbb{R}^d} \e^{\I\scp{x-y}{\xi}} a(x,y,\xi) \diff\xi \label{eq:PseudorKernel}
\end{align}
and write $\Op(a)u(x) = \int_{\mathbb{R}^d} K(x,y) u(y) \diff y$ as an integral operator.

The following examples of pseudo-differential operators are easy to verify.
\begin{exmp}
	\begin{enumerate}[wide, labelwidth=!, labelindent=0pt]
		\item If $a(x,y,\xi) = \xi^\alpha$, then $\mathrm{Op}(a)u(x) = D^\alpha u(x)$.
		\item If $a(x,y,\xi) = a(x)$ or $a(x,y,\xi) = a(y)$, then $\mathrm{Op}(a)u(x) = a(x)u(x)$.
		\item If $a(x,y,\xi) = \sum_{|\alpha|\leq m} g_\alpha(x) \xi^\alpha$, then $\Op(a)u(x) = \sum_{|\alpha|\leq m} D^\alpha (g_\alpha u)(x)$.		
		\item If $a(x,y,\xi) = \sum_{|\alpha|\leq m} g_\alpha(y) \xi^\alpha$, then $\mathrm{Op}(a)u(x) = \sum_{|\alpha|\leq m} g_\alpha(x) D^\alpha u(x)$.
	\end{enumerate}
	\label{exmp:Pseudors}
\end{exmp}

\begin{rem}
	The last two examples establish that replacing $g_\alpha(x)$ by $g_\alpha(y)$ commutes the order in which the differential $D^\alpha$ and multiplication by $g_\alpha$ act on $u$. This is connected to the ordering ambiguity of quantum mechanics (cf. Section \ref{sec:ChangeOfQuantisation} below).
\end{rem}

Thus far, we defined pseudo-differential operators only on $C_\mathrm{c}^\infty(\mathbb{R}^d)$, which is insufficient for practical purposes. In the rest of this section, we continuously extend pseudo-differential operators to Schwartz functions and, subsequently, to tempered distributions.

\begin{thm}
	If $a\in S(M)$, then $\Op(a)$ defines a continuous map from $\mathcal{S}(\mathbb{R}^d)$ to $\mathcal{S}(\mathbb{R}^d)$.
	\label{thm:PseudorsSchwartzSpaceMapping}
\end{thm}

\begin{proof}
	For simplicity, we assume that the symbol does not depend on the integration variable $y$. This assumption is justified in the next section. The reason for this choice is that $\Op(a)u$ simplifies to
	\begin{align}
		\Op(a)u(x) = \frac{1}{(2\pi)^d} \int_{\mathbb{R}^d} \e^{\I\scp{x}{\xi}} a(x,\xi) \widehat{u}(\xi) \diff\xi.
	\end{align}	
	In this simplified form, it is easy to verify that $x^\beta D_x^\alpha \Op(a)u(x)$ is a linear combination of 
	\begin{align}
		\Op(x^{\beta-\delta} D_x^{\alpha-\gamma}a) (x^{\delta} D_x^{\gamma} u)(x).
	\end{align}
	It follows that $x^\beta D_x^\alpha \Op(a)u(x)$ is bounded by a finite sum of semi-norms of $u$. This proves that $\Op(a)u \in \mathcal{S}(\mathbb{R}^d)$ if $u \in \mathcal{S}(\mathbb{R}^d)$ and that $\Op(a):\mathcal{S}(\mathbb{R}^d)\to \mathcal{S}(\mathbb{R}^d)$ is a continuous map.	
\end{proof}

To establish the same result for the dual space $\mathcal{S}'(\mathbb{R}^d)$, we define the distributional pairing 
\begin{align}
	\scp{\mathrm{Op}(a)u}{v} = \scp{u}{\mathrm{Op}(b) v}
	\label{eq:DistributionalPairing}
\end{align}
for $u \in \mathcal{S}'(\mathbb{R}^d)$, $v\in\mathcal{S}(\mathbb{R}^d)$, where $\Op(b) = \Op(a)^*$ is the formal adjoint of $\Op(a)$. The symbol $b$ of the formal adjoint is defined by $b(x,y,\xi) = \overline{a(y,x,\xi)}$.

\begin{thm}
	If $a\in S(M)$, then $\Op(a)$ defines a continuous map from $\mathcal{S}'(\mathbb{R}^d)$ to $\mathcal{S}'(\mathbb{R}^d)$.
	\label{thm:PseudorsSchwartzDualSpaceMapping}
\end{thm}

\begin{proof}
	If $u\in \mathcal{S}'(\mathbb{R}^d)$, then $\mathrm{Op}(a)u$ defined in \eqref{eq:DistributionalPairing} is an element of $\mathcal{S}'(\mathbb{R}^d)$ because $\mathrm{Op}(b) v \in \mathcal{S}(\mathbb{R}^d)$ for all $v\in \mathcal{S}(\mathbb{R}^d)$ by the previous theorem. Moreover, $\mathrm{Op}(a):\mathcal{S}'(\mathbb{R}^d)\to\mathcal{S}'(\mathbb{R}^d)$ is continuous because $\scp{\Op(a)u_n}{v} = \scp{u_n}{\Op(b)v} \to 0$ for all $v\in \mathcal{S}(\mathbb{R}^d)$ if $u_n \to 0$ in $\mathcal{S}'(\mathbb{R}^d)$.
\end{proof}

\section{Change in Quantisation}
\label{sec:ChangeOfQuantisation}

In Definition \ref{def:Pseudor}, we assumed that the underlying symbol of a pseudo-differential operator depends on the position variable $x$, integration variable $y$, and momentum variable $\xi$ (almost) arbitrarily. In this section, we demonstrate that %the dependence on the integration variable $y$ is somewhat superfluous. More specifically, we illustrate that 
it is not a restriction to replace the arbitrary dependence on $(x,y)$ by the convex combination $tx+(1-t)y$ for any $t\in[0,1]$, where the parameter $t$ distinguishes different quantisations.

For a symbol $a \in S_{2d}(M)$, where the order function $M$ satisfies the same assumptions as in Definition \ref{def:Pseudor}, we set
\begin{align}
	\Op_t(a)u(x) = \frac{1}{(2\pi)^d} \int_{\mathbb{R}^d} \int_{\mathbb{R}^d} \e^{\I\scp{x-y}{\xi}} a(tx+(1-t)y,\xi) u(y) \diff y \diff \xi,
\end{align}
and call $\Op_t(a)$ a pseudo-differential operator in \textbf{quantisation} $t$.

\begin{thm}
	Let $a \in S_{3d}(M)$. Then, for every $t\in[0,1]$, a unique $a_t \in S_{2d}(M)$ exists such that $\Op(a) = \Op_t(a_t)$. The symbol $a_t$ is given by
	\begin{align}
		a_t(x,\xi) &= \frac{1}{(2\pi)^d} \int_{\mathbb{R}^d} \int_{\mathbb{R}^d} \e^{\I\scp{\theta}{\eta-\xi}} a(x+(1-t)\theta, x-t\theta,\eta) \diff\eta \diff\theta \notag \\
		&= \e^{-\I \scp{D_\theta}{D_\xi}}a(x+(1-t)\theta,x-t\theta, \xi) \big|_{\theta = 0}.
	\end{align} \label{thm:PseudorReduction}
\end{thm}

\begin{proof}
	To determine the symbol $a_t$, we compare the integral kernels of $\Op_t(a_t)$ and $\Op(a)$:
	\begin{align}
		\int_{\mathbb{R}^d} \e^{\I\scp{x-y}{\xi}} a_t(tx+(1-t)y,\xi) \diff\xi = \int_{\mathbb{R}^d} \e^{\I\scp{x-y}{\xi}} a(x,y,\xi) \diff\xi.
	\end{align}
	If $\theta = x-y$ and $z=tx + (1-t)y = x-(1-t)\theta$, then
	\begin{align}
		\int_{\mathbb{R}^d} \e^{\I\scp{\theta}{\xi}} a_t(z,\xi) \diff\xi = \int_{\mathbb{R}^d} \e^{\I\scp{\theta}{\xi}} a(z+(1-t)\theta,z-t\theta,\xi) \diff\xi.
	\end{align}
	Hence, an inverse Fourier transformation with respect to $\theta$ on both sides yields 
	\begin{align}
		a_t(z,\eta) = \frac{1}{(2\pi)^d} \int_{\mathbb{R}^d} \int_{\mathbb{R}^d} \e^{\I\scp{\theta}{\xi-\eta}} a(z+(1-t)\theta,z-t\theta,\xi) \diff\xi \diff \theta,
	\end{align}
	which is the desired result if we rename $x \leftrightarrow z$ and $\eta \leftrightarrow \xi$.
\end{proof}

Due to their appearance in applications, three values of $t$ are particularly important:
\begin{itemize}
	\item \textbf{Standard quantisation} ($t=1$): $a(x,D) = \mathrm{Op}(a) = \mathrm{Op}_{1}(a)$;
	\item \textbf{Weyl quantisation} ($t=\frac12$): $a^\mathrm{w}(x,D) = \mathrm{Op}_{\frac{1}{2}}(a)$; 	
	\item \textbf{Right quantisation} ($t=0$): $\mathrm{Op}_0(a)$.
\end{itemize}
The Weyl quantisation has the pleasant property that $a^\mathrm{w}(x,D)$ is a symmetric operator if $a$ is real-valued. However, the standard quantisation and, sometimes, the right quantisation are easier to handle. For example, a pseudo-differential operator in standard quantisation has the following form:
\begin{align}
	a(x,D) u(x) = \frac{1}{(2\pi)^d} \int_{\mathbb{R}^d} \e^{\I\scp{x}{\xi}} a(x,\xi) \widehat{u}(\xi) \diff \xi.
\end{align}
To understand the different quantisations better, we consider the following example relevant for the Nelson model with variable coefficients:
\begin{align}
	a(x,\xi) = \sum_{1\leq j,k \leq d} \xi_j g^{jk}(x) \xi_k \in S^{2}.
\end{align}
Then $\Op_t(a)$ in standard, Weyl, and right quantisation is given by
\begin{align}
	a(x,D)u(x) &= \sum_{1\leq j,k \leq d} g^{jk}(x) D_{x_j} D_{x_k}u(x), \\
	a^\mathrm{w}(x,D)u(x) &= \sum_{1\leq j,k \leq d} D_{x_j} g^{jk}(x) D_{x_k}u(x), \\
	\Op_0(a)u(x) &= \sum_{1\leq j,k \leq d} D_{x_j} D_{x_k} g^{jk}(x) u(x).
\end{align}
We observe that the choice of a particular quantisation (i.e. the choice of a particular $t\in[0,1]$) resolves the ordering ambiguity of the classical position observable $x$ and momentum observable $\xi$. Furthermore, we deduce from the commutator $[D_{x}, g(x)] = (D_{x} g)(x)$ that the difference between the standard, Weyl, and right quanti\-sation is an operator of order $1$, whereas $a$ is a symbol of order $2$. This observation can be generalised. 

\begin{thm}
	Let $a\in S_{2d}(M)$. Then $\Op_t(a_t) = \Op_0(a)$ for $a_t = \e^{\I t \scp{D_x}{D_\xi}} a(x,\xi)$.
	\label{thm:ChangingQuantisations}
\end{thm}

\begin{proof}
	From Theorem \ref{thm:PseudorReduction} it follows that
	\begin{align}
		a_t(x,\xi) &= \e^{-\I \scp{D_\theta}{D_\xi}}a(x-t\theta, \xi) \big|_{\theta = 0} \notag \\
		&= \e^{\I t\scp{D_x}{D_\xi}}a(x-t\theta, \xi) \big|_{\theta = 0} \notag \\
		&= \e^{\I t\scp{D_x}{D_\xi}}a(x,\xi).
	\end{align}
	This proves the theorem.
\end{proof}

\begin{cor}
	If $a\in S^m$ and $s,t\in [0,1]$, then the pseudo-differential operators $\mathrm{Op}_t(a)$ and $\mathrm{Op}_s(a)$ are identical up to a pseudo-differential operator of order $m-1$.
	\label{cor:ChangingQuantisations}
\end{cor}

\begin{proof}
	Theorem \ref{thm:ChangingQuantisations} indicates that $\mathrm{Op}_t(a) = \mathrm{Op}_s(b)$ for 
	\begin{align}
		b(x,\xi) &= \e^{\I (s-t)\scp{D_x}{D_\xi}}a(x,\xi) \notag \\
		&\sim \sum_\alpha \frac{(s-t)^\alpha}{\alpha!} D_x^\alpha \partial_\xi^\alpha a(x,\xi) \notag \\
		&= a(x,\xi) + S^{m-1}.
	\end{align}	
	Thus, $\mathrm{Op}_t(a) = \mathrm{Op}_s(b) = \mathrm{Op}_s(a) + \mathrm{Op}(S^{m-1})$.
\end{proof}

\section{Symbolic Calculus}

Thus far, we formulated the basic notions of symbols and pseudo-differential operators. In this section, we explore how manipulations of a pseudo-differential operator affect its symbol.

\subsection{Adjoint Symbol}

We mentioned at the end of Section \ref{sec:Pseudors} that the adjoint operator $\Op(a)^* = \Op(b)$ is a pseudo-differential operator with symbol $b(x,y,\xi) = \overline{a(y,x,\xi)}$. In terms of the pseudo-differential operator $\Op_t(a)$ in quantisation $t$, this means that
\begin{align}
	\Op_t(a)^* = \Op_{1-t}(\overline{a}).
\end{align}
A slight obstacle is that $\Op_{1-t}(\overline{a})$ is not in the same quantisation as $\Op_t(a)$ except if $t=1/2$. However, based on our results from Section \ref{sec:ChangeOfQuantisation}, it is not difficult to construct the \textbf{adjoint symbol} $a^*$ satisfying $\Op_t(a)^* = \Op_t(a^*)$. The following theorem immediately follows from Theorem \ref{thm:ChangingQuantisations}.
\begin{thm}
	Let $a\in S(M)$ and set
	\begin{align}
		a^*(x,\xi) &= \e^{\I(2t-1) \scp{D_x}{D_\xi}} \overline{a(x,\xi)} \label{eq:AdjointSymbol1}.
	\end{align}
	Then $a^*$ defines a symbol in $S(M)$, and $\Op_t(a)^* = \Op_t(a^*)$.
	\label{prop:AdjointSymbol}
\end{thm}

\subsection{Moyal Product}

Determining a symbol for the composition $\Op_t(a_1)\Op_t(a_2)$ with $a_1 \in S(M_1)$ and $a_2 \in S(M_2)$ is more difficult. For simplicity, we discuss this construction only for the standard quantisation. Expanding the expression for $\Op(a_1)\Op(a_2)u$, we obtain
\begin{align}
	\mathrm{Op}(a_1)\mathrm{Op}(a_2)u(x) &= \frac{1}{(2\pi)^{d}} \int_{\mathbb{R}^{d}} \e^{\I\scp{x}{\eta}} a_1(x,\eta) \mathcal{F}[\Op(a_2)u](\eta) \diff \eta \notag \\	
	&= \frac{1}{(2\pi)^{2d}} \int_{\mathbb{R}^{3d}} \e^{\I\scp{x}{\eta}} \e^{\I\scp{y}{\xi-\eta}} a_1(x,\eta) a_2(y,\xi) \widehat{u}(\xi) \diff y \diff \xi \diff \eta \notag \\
	&= \frac{1}{(2\pi)^{d}} \int_{\mathbb{R}^{d}} \e^{\I\scp{x}{\xi}} a_1\# a_2(x,\xi) \widehat{u}(\xi) \diff \xi \notag \\
	&= \Op(a_1\# a_2)u(x),
\end{align}
where we defined the \textbf{Moyal product}
\begin{align}
	a_1 \# a_2(x,\xi) &= \frac{1}{(2\pi)^d} \int_{\mathbb{R}^{d}} \int_{\mathbb{R}^d} \e^{-\I\scp{x-y}{\xi-\eta}} a_1(x,\eta) a_2(y,\xi) \diff y \diff \eta \notag \\
	&= \e^{\I \scp{D_y}{D_\eta}} a_1(x,\eta) a_2(y,\xi) \big|_{\eta = \xi, y=x}.
\end{align}
A similar construction is possible for the other quantisations. For completeness, we state the following theorem (see \cite[Thm.~2.7.4]{martinez2002}).
\begin{thm}
	Let $a_1 \in S(M_1)$, $a_2 \in S(M_2)$, and set
	\begin{align}
		a_1 \#_t a_2 &= \e^{\I (\scp{D_\eta}{D_v} - \scp{D_u}{D_\xi})} a_1(tx+(1-t)u,\eta) a_2((1-t)x+tv,\xi)  \big|_{u=v=x, \eta=\xi}.
	\end{align}
	Then $a_1 \#_t a_2 \in S(M_1M_2)$ and $\Op_t(a_1)\Op_t(a_2) = \Op_t(a_1 \#_t a_2)$.
	\label{thm:MoyalProduct}
\end{thm}

%The principal symbol of $a_1\#_t a_2$ is $a_1a_2$ (i.e. a simple multiplication). If $a_1$ does not depend on $\xi$ or if $a_2$ does not depend on $x$, the principal symbol is the exact symbol. 

\begin{exmp}
	We construct a symbol for the commutator $[\cdot,\cdot]$ of two pseudo-differential operators in standard quantisation. If $a_1 \in S^{m_1}$ and $a_2 \in S^{m_2}$, then
	\begin{align}
		[\Op(a_1),\Op(a_2)] = \Op(a_1\# a_2 - a_2 \# a_1).
	\end{align}
	In leading order, the asymptotic expansion of $a_1\# a_2 - a_2 \# a_1$ is given by
	\begin{align}
		a_1\# a_2 - a_2 \# a_1 &= -\I (\partial_\xi a_1  \partial_x a_2 - \partial_\xi a_2 \partial_x a_1) + S^{m_1+m_2-2} \notag \\
		&= \I \{ a_1, a_2 \} + S^{m_1+m_2-2},
	\end{align}
	where $\{\cdot,\cdot\}$ is the \textbf{Poisson bracket}. This is the well-known correspondence between the Poisson bracket in classical theory and the commutator in quantum theory.
\end{exmp}

\subsection{Construction of a Parametrix}

We use the Moyal product to construct inverses of elliptic pseudo-differential operators.

\begin{defn}
	A symbol $a\in S_k(M)$ is \textbf{elliptic} if a constant $C>0$ exists such that $|a(z)| \geq C M(z)$ for all $z\in\mathbb{R}^k$. A pseudo-differential operator is elliptic if its underlying symbol is elliptic.
\end{defn}

The inverse function $M^{-1}$ of an order function $M$ is an order function. Hence, if $a\in S(M)$ is elliptic, then $a^{-1} \in S(M^{-1})$. The following theorem reveals that every elliptic symbol has an inverse with respect to the Moyal product.

\begin{thm}
	Let $a \in S(M)$ be elliptic. Then a symbol $b \in S(M^{-1})$ uniquely determined up to asymptotic equivalence exists such that $a\# b \sim 1$ and $b \# a \sim 1$.
	\label{thm:MoyalInverse}
\end{thm}

\begin{proof}
	We prove the theorem for $\# = \#_1$ and $a\in S^m$, $m \in \mathbb{R}$. The proof for arbitrary quantisations and general $a \in S(M)$ uses similar ideas. 
	
	The first candidate for the inverse symbol of $a$ with respect to the Moyal product is $a^{-1}$. Modulo a symbol $r\in S^{-1}$, the symbol $a^{-1} \in S^{-m}$ is the inverse of $a$ because
	\begin{align}
		a \# a^{-1} \sim 1 + \sum_{|\alpha| \geq 1} \frac{1}{\alpha!} \partial_\xi^\alpha a D_x^\alpha a^{-1} \equiv 1-r.
	\end{align}
	We define the symbol $b_1$ to be asymptotically equivalent to
	\begin{align}
		b_1 \sim a^{-1} \# r + a^{-1} \# r \# r + \cdots.
	\end{align}
	This symbol exists by Theorem \ref{thm:AsymptoticExpansion}, and $b_1$ is a right inverse of $a$ up to asymptotic equivalence:
	\begin{align}
		a \# b_1 \sim (1-r) \# (r+ r \# r \cdots) = 1 + S^{-\infty}.
	\end{align}
	In similar fashion, we construct a left inverse $b_2$ satisfying $b_2 \# a = 1 + S^{-\infty}$. From $b_1 \sim b_2\# a \# b_1 \sim b_2$, it follows that $b_1$ is asymptotically equivalent to $b_2$.
\end{proof}

An immediate consequence of the above theorem is that, for every pseudo-differential operator $A = \Op(a) \in \Op(S^m)$, a pseudo-differential operator $B$ exists such that $AB=BA = 1 + \Op(S^{-\infty})$. The operator $B$ is called an \textbf{approximate inverse} of $A$. Because pseudo-differential operators in $\Op(S^{-\infty})$ are smoothing (see Theorem \ref{thm:Smoothing} below), the approximate inverse $B$ is also called a \textbf{parametrix} of $A$.

\subsection{Functional Calculus}

The symbolic calculus developed so far culminates in the functional calculus of self-adjoint pseudo-differential operators. Although we do not prove the functional calculus, it satisfies all properties we expect intuitively. If $f$ is a sufficiently smooth function, then $f(\Op(a))$ is a pseudo-differential operator and its principal symbol is $f(a)$. 'Sufficiently smooth' in this case means that '$f$ is a symbol of order $p$', that is, an element of
\begin{align}
	S^p(\mathbb{R}) = \{f\in C^{\infty}(\mathbb{R}) \mid |f^{(k)}(t)| \leq C_k\jap{t}^{p-k} \ \forall k \in \mathbb{N} \}.
	\label{eq:Sp}
\end{align}
For a full discussion of the functional calculus we refer to \cite{bony2013}. The functional calculus is usually formulated for pseudo-differential operators in Weyl quantisation because in this quantisation a pseudo-differential operator is symmetric if the underlying symbol is real-valued. 

\begin{thm}[Functional calculus]
	Let $a\in S(M,g)$ be real-valued and elliptic such that its Weyl quantisation $a^{\mathrm{w}}$ is essentially self-adjoint on $\mathcal{S}(\mathbb{R}^d)$. If $f\in S^p(\mathbb{R})$, the operator $f(a^\mathrm{w})$ is a pseudo-differential operator, and the symbol of $f(a^\mathrm{w})$ is an element of $S(M^p,g)$. If $a\in S^m$, then $f(a^\mathrm{w}) - f(a)^\mathrm{w} \in \mathrm{Op}(S^{mp-1})$.
	\label{thm:FunctionalCalculus}
\end{thm}

\section{$L^2$-continuity}
\label{sec:L2Continuity}

For applications in quantum mechanics, it is vital to know how pseudo-differential operators act on $L^2(\mathbb{R}^d) \subset \mathcal{S}'(\mathbb{R}^d)$. It is well known from Fourier analysis that the Fourier transformation on $\mathcal{S}'(\mathbb{R}^d)$ restricts to a bounded operator on $L^2(\mathbb{R}^d)$. The Calderon--Vaillancourt theorem generalises this statement for pseudo-differential operators with bounded symbols. In this section, we present a simplified but instructive proof of the Calderon--Vaillancourt theorem for symbols in $S^0$. In the next section, we present a more general proof for symbols in $S(1)$.

\begin{thm}[Calderon--Vaillancourt]
	For every $a \in S^0$, the pseudo-differential operator $\Op(a)$ defines a bounded operator on $L^2(\mathbb{R}^d)$.
	\label{thm:L2Continuity}
\end{thm}

\begin{proof}
	First, we assume that $a\in S^{-d-1}$. In standard quantisation, the integral kernel of $\mathrm{Op}(a)$ is
	\begin{align}
		K(x,y) = \frac{1}{(2\pi)^d} \int_{\mathbb{R}^d} \e^{\I\scp{x-y}{\xi}} a(x,\xi) \diff \xi,
	\end{align}
	and, due to $a\in S^{-d-1}$,
	\begin{align}
		|(x-y)^\alpha K(x,y)| &= \frac{1}{(2\pi)^d} \left|\int_{\mathbb{R}^d} (\partial_\xi^\alpha \e^{\I\scp{x-y}{\xi}}) a(x,\xi) \diff \xi\right| \notag \\
		&\leq \frac{1}{(2\pi)^d} \int_{\mathbb{R}^d} |\partial_\xi^\alpha a(x,\xi)| \diff \xi
	\end{align}
	is bounded for every multi-index $\alpha \in \mathbb{N}^d$. Hence, $(1+|x-y|^2)^d |K(x,y)|$ is bounded; therefore,
	\begin{align}
		\sup_y \int_{\mathbb{R}^d} |K(x,y)| \diff x \leq \sup_y \int_{\mathbb{R}^d} \frac{C}{(1+|x-y|^2)^d} \diff x  \leq C
	\end{align}
	and
	\begin{align}
		\sup_x \int_{\mathbb{R}^d} |K(x,y)| \diff y \leq \sup_x \int_{\mathbb{R}^d} \frac{C}{(1+|x-y|^2)^d} \diff y \leq C.
	\end{align}
	It follows from the Schur test (see Theorem \ref{thm:SchurTest} with $B=\mathbb{C}$) that the integral operator $\Op(a)$ is bounded on $L^2(\mathbb{R}^d)$.
	
	Next, we prove that $\Op(a)$ is bounded for symbols of negative order. It suffices to prove that $\Op(a)$ is bounded if $a\in S^{-1/2^k}$ for all $k\in\mathbb{\mathbb{Z}}$ because $S^m \subset S^n$ for $m \leq n$. This proof is done by induction. According to the first step, a symbol $a \in S^{-1/2^k}$ defines a bounded operator if $k$ is sufficiently small (e.g. $k\leq -d$). Assume we proved the claim for all $k\leq l$. If $a\in S^{-1/2^{l+1}}$, then $a^* \# a \in S^{-1/2^l}$, and
	\begin{align}
		\norm{\Op(a)u}^2 = \scp{\Op(a^*\# a)u}{u} \leq \norm{\Op(a^*\# a)} \norm{u}^2.
	\end{align}
	By the induction hypothesis, $\norm{\Op(a^*\# a)}<\infty$; hence, $\Op(a)$ is bounded.
	
	It remains to prove the theorem for $a \in S^0$. Define
	\begin{align}
		M^2 = \sup_{(x,\xi)} |a(x,\xi)|^2 < \infty.
	\end{align}
	Then $M^2 - |a(x,\xi)|^2 \in S^0$ is a nonnegative symbol. Furthermore, if $F \in C^\infty(\mathbb{C})$ is any smooth function such that $F(z) = \sqrt{1+z}$ for $z \in \mathbb{R}_+$, then
	\begin{align}
		b(x,\xi) = F(M^2 - |a(x,\xi)|^2)
	\end{align}
	defines a symbol $b \in S^0$. The symbol $b$ is chosen such that $a^* \# a + b^* \# b = 1 + M^2 + c$ for a $c \in S^{-1}$. Thus, it follows that
	\begin{align}
		\norm{\Op(a) u}^2 &\leq \norm{\Op(a) u}^2 + \norm{\Op(b) u}^2 \notag \\
		&= \scp{\Op(a^* \# a + b^* \# b) u}{u} \notag \\ 
		&= \scp{(1 + M^2) u}{u} + \scp{\Op(c) u}{u} \notag \\
		&\leq (1 + M^2 + \norm{\Op(c)}) \norm{u}^2.
	\end{align}
	Hence, $\Op(a)$ is bounded because $\norm{\Op(c)}<\infty$ by the previous step.
\end{proof}

\begin{cor}
	Let $a \in S^m$. Then, for every $s\in\mathbb{R}$, $\Op(a):H^s \to H^{s-m}$ is a bounded operator from the Sobolev space $H^s$ to $H^{s-m}$.
	\label{cor:PseudorsSobolevMapping}
\end{cor}

\begin{proof}
	Set $b=\jap{\cdot}^{s-m}\# a \# \jap{\cdot}^{-s} \in S^0$. Then
	\begin{align}
		\norm{\mathrm{Op}(a) u}_{H^{s-m}} &= \norm{\jap{D}^{s-m} \Op(a)u}_{L^2} 
		= \norm{\mathrm{Op}(b) \jap{D}^{s} u}_{L^2} \notag \\ &\leq C \norm{\jap{D}^{s} u}_{L^2} = C \norm{u}_{H^{s}}
	\end{align}
	proves the claim.
\end{proof}

The Calderon--Vaillancourt theorem has further consequences for mapping properties of pseudo-differential operators. We remind the reader of the following notions. If $u\in \mathcal{S}'(\mathbb{R}^d)$ is a tempered distribution, the \textbf{singular support} of $u$ denoted by $\singsupp(u)$ is the complement of all $x\in \mathbb{R}^d$ that have an open neighbourhood $U$ such that $\varphi u \in C_\mathrm{c}^\infty(U)$ for all $\varphi \in C_\mathrm{c}^\infty(U)$. An operator $A$ on $\mathcal{S}'(\mathbb{R}^d)$ is called \textbf{pseudo-local} if $\singsupp(Au) \subset \singsupp(u)$ for all $u\in\mathcal{S}'(\mathbb{R}^d)$ and \textbf{hypoelliptic} if $\singsupp(Au) = \singsupp(u)$ for all $u\in\mathcal{S}'(\mathbb{R}^d)$. Furthermore, $A$ is called \textbf{smoothing} if it maps tempered distributions to smooth functions.

\begin{thm}
	If $a\in S^{-\infty}$, then $\Op(a)$ is smoothing.
	\label{thm:Smoothing}
\end{thm}

\begin{proof}
	Let $u\in\mathcal{S}'(\mathbb{R}^d)$ be a tempered distribution. To prove that $\Op(a)u$ is smooth, we demonstrate that $D^\alpha \Op(a)u = \Op(\xi^\alpha \# a)u \equiv \Op(b)u$ is a continuous function for every multi-index $\alpha \in \mathbb{N}^d$. As $u$ is a tempered distribution, an $N\in\mathbb{N}$ exists such that $v=(1+|\cdot|^2)^{-N} u \in H^{-N}$ (see \cite[Lem.~1.16]{raymond1991}). A simple computation indicates that
	\begin{align}
		\Op(b) u(x) &= \Op(b) (1+|x|^2)^N v(x) \notag \\
		&= \frac{1}{(2\pi)^d} \int_{\mathbb{R}^d} \e^{\I\scp{x-y}{\xi}} b(x,\xi) \sum_{\beta \in \mathbb{N}^d} \frac{1}{\beta!} \partial_x^\beta(1+|x|^2)^N (y-x)^\beta v(y) \diff y \diff \xi \notag \\
		&= \sum_{\beta \in \mathbb{N}^d} \frac{1}{\beta!} D_x^\beta (1+|x|^2)^N \Op(\partial_\xi^\beta b)v(x).
	\end{align}
	Clearly, $\partial_\xi^\beta b \in S^{-\infty}$ for all $\beta \in \mathbb{N}^d$ because $b\in S^{-\infty}$. It follows that $\Op(\partial_\xi^\beta b)v \in H^d$ by Corollary \ref{cor:PseudorsSobolevMapping}. Because $H^d$ is a subspace of the space of continuous functions, $\Op(\partial_\xi^\beta b)v$ is continuous, and so is $\Op(b)u$.
\end{proof}

\begin{thm}
	If $a\in S^m$, then $\Op(a)$ is pseudo-local. If $a\in S^m$ is elliptic, then $\Op(a)$ is hypoelliptic.
	\label{thm:PseudoLocal}
\end{thm}

\begin{proof}
	Let $u\in \mathcal{S}'(\mathbb{R}^d)$ be a tempered distribution and $U = \mathbb{R}^d \backslash \singsupp(u)$ be the com\-ple\-ment of the singular support of $u$. Then, $\psi u \in C_\mathrm{c}^\infty(U)$ for all $\psi \in C_\mathrm{c}^\infty(U)$. To prove $\singsupp(\Op(a)u) \subset \singsupp(u)$, we show that $\varphi \Op(a)u \in C_\mathrm{c}^\infty(U)$ for all $\varphi \in C_\mathrm{c}^\infty(U)$. If $\varphi \in C_\mathrm{c}^\infty(U)$, a $\psi \in C_\mathrm{c}^\infty(U)$ exists such that $\psi=1$ on $\supp(\varphi)$. We write $\varphi \Op(a)u$ as
	\begin{align}
		\varphi \Op(a)u = \varphi \Op(a)(\psi u) + \varphi \Op(a)((1-\psi)u).
	\end{align}
	The function $\Op(a)(\psi u)$ is a Schwartz function because $\psi u\in C_\mathrm{c}^\infty(U) \subset \mathcal{S}(\mathbb{R}^d)$ and because $\Op(a)$ maps Schwartz functions to Schwartz functions (see Theorem \ref{thm:PseudorsSchwartzSpaceMapping}). Hence, $\varphi \Op(a)(\psi u) \in C_\mathrm{c}^\infty(U)$. The second summand may be written as $\varphi \Op(a)((1-\psi)u) = \Op(b)u$ with $b=\varphi a \# (1-\psi) \in S^{-\infty}$ because the supports of $\varphi a$ and $1-\psi$ are disjoint. From the previous theorem it follows that $\Op(b)u$ is smooth; thus, $\Op(b)u \in C_\mathrm{c}^\infty(U)$.
	
	If $a\in S^m$ is elliptic, then, according to Theorem \ref{thm:MoyalInverse}, a $b \in S^{-m}$ and $r\in S^{-\infty}$ exist such that $b \# a = 1 + r$. The function $\Op(r)u$ is smooth because $\Op(r)$ is smoothing; hence, $\singsupp(u) = \singsupp(\Op(b \# a)u) \subset \singsupp(\Op(a)u)$.
\end{proof}

\section{Operator-valued Symbols}
\label{sec:OperatorValuedSymbols}

In this section, we discuss pseudo-differential operators with operator-valued symbols as presented in \cite[Appendix A]{teufel2003}. This calculus does not provide any more difficulties than what we discussed so far, as most of the proofs can be adapted without modification. In Chapter \ref{ch:IBC}, where we apply the IBC method to the Nelson model with variable coefficients, a calculus with operator-valued symbols simplifies our analysis.

Let $\mathcal{H}_1$, $\mathcal{H}_2$ be two separable Hilbert spaces. We write $\mathfrak{B}(\mathcal{H}_1, \mathcal{H}_2)$ for the space of all bounded operators from $\mathcal{H}_1$ to $\mathcal{H}_2$. The following definition is a direct generalisation of the symbol class $S_k(M)$.

\begin{defn}
	Let $M$ be an order function. Then $S_k(M, \mathfrak{B}(\mathcal{H}_1, \mathcal{H}_2))$ is the space of all $a\in C^\infty(\mathbb{R}^k, \mathfrak{B}(\mathcal{H}_1, \mathcal{H}_2))$ that satisfy, for all $\alpha \in \mathbb{N}^k$, the following bound:
	\begin{align}
		\norm{\partial^\alpha_z a(z)}_{\mathfrak{B}(\mathcal{H}_1, \mathcal{H}_2)} \leq C_{\alpha} M(z).
	\end{align}	
\end{defn}

We recover the symbol class $S_k(M)$ if $\mathcal{H}_1 = \mathcal{H}_2 = \mathbb{C}$. Again, $S_k(M, \mathfrak{B}(\mathcal{H}_1, \mathcal{H}_2))$ is a Fréchet space with semi-norms given by the best constants $C_\alpha$ satisfying the above estimate. The operator-valued pseudo-differential operator $\Op(a)$ for a symbol $a\in S_{3d}(M,\mathfrak{B}(\mathcal{H}_1, \mathcal{H}_2))$ is defined on $C_\mathrm{c}^\infty(\mathbb{R}^d, \mathcal{H}_1)$ in an analogy to Definition \ref{def:Pseudor} by
\begin{align}
	\mathrm{Op}(a) u(x) &= \frac{1}{(2\pi)^d} \int_{\mathbb{R}^d} \int_{\mathbb{R}^d} \e^{\I\scp{x-y}{\xi}} a(x,y,\xi) u(y) \diff y \diff \xi,
\end{align}
where the integral is a \textbf{Bochner integral}. The results of Section \ref{sec:ChangeOfQuantisation} generalise without modification to operator-valued pseudo-differential operators. Thus, it is not a restriction to assume that $a(x,y,\xi) = a(tx+(1-t)y,\xi)$ for some $t\in[0,1]$. 
%For example, an operator-valued pseudo-differential operator in standard quantisation is given by
%\begin{align}
%	\Op_1(a)u(x) = \frac{1}{(2\pi)^d} \int_{\mathbb{R}^d} \e^{\I \scp{x}{\xi}} a(x,\xi) \widehat{u}(\xi) \diff \xi.
%\end{align}

The only theorem we re-prove for operator-valued pseudo-differential operators is the Calderon--Vaillancourt theorem because of its central importance for the present work and because the proof of Theorem \ref{thm:L2Continuity} fails for symbols in $S(1)$. The proof we present follows \cite[Thm. 2.8.1]{martinez2002}, which we translate into the setting of operator-valued symbols.

\begin{thm}[Calderon--Vaillancourt]
	For every $a \in S(1,\mathfrak{B}(\mathcal{H}_1, \mathcal{H}_2))$, the operator-valued pseudo-differential operator $\Op(a)$ extends to a bounded operator from $L^2(\mathbb{R}^d,\mathcal{H}_1)$ to $L^2(\mathbb{R}^d,\mathcal{H}_2)$. Moreover, there are constants $C,N>0$ independent of $a$ such that
	\begin{align}
		\norm{\Op(a)u}_{L^2(\mathbb{R}^d,\mathcal{H}_2)} \leq C \sum_{|\alpha| \leq N} \sup_{z \in \mathbb{R}^{2d}} \norm{\partial_z^\alpha a(z)}_{\mathfrak{B}(\mathcal{H}_1, \mathcal{H}_2)} \norm{u}_{L^2(\mathbb{R}^d,\mathcal{H}_1)}
	\end{align}
	for all $u\in L^2(\mathbb{R}^d,\mathcal{H}_1)$.
	\label{thm:CalderonVaillancourtOV}
\end{thm}

\begin{proof}	
	The key idea of the proof is to apply the Cotlar--Stein Lemma with $H_1 = L^2(\mathbb{R}^d,\mathcal{H}_1)$ and $H_2 = L^2(\mathbb{R}^d,\mathcal{H}_2)$ (see Theorem \ref{thm:CotlarStein}). 
	
	We define a partition of the operator $A=\Op(a)$ by selecting a $\chi \in C_\mathrm{c}^\infty(\mathbb{R}^{2d})$ such that
	\begin{align}
		\sum_{\mu \in \mathbb{Z}^{2d}} \chi_\mu = 1,
	\end{align}
	where $\chi_\mu(z) = \chi(z-\mu)$ is the function $\chi$ shifted by the lattice point $\mu \in \mathbb{Z}^{2d}$.
	We set $a_\mu = a \chi_\mu$ and $A_\mu = \Op(a_\mu)$. The operator $A_\mu:H_1\to H_2$ is bounded for every $\mu \in \mathbb{Z}^{2d}$:
	\begin{align}
		\norm{A_\mu u}_{H_2}^2 &= \int_{\mathbb{R}^d} \norm{A_\mu u(x)}^2_{\mathcal{H}_2} \diff x \notag \\
		&\leq \frac{1}{(2\pi)^{2d}} \int_{\mathbb{R}^d} \left(\int_{\mathbb{R}^d} \norm{a_\mu(x,\xi) \widehat{u}(\xi)}_{\mathcal{H}_2} \diff\xi \right)^2 \diff x \notag \\
		&\leq \frac{1}{(2\pi)^{2d}} \int_{\mathbb{R}^d} \left(\int_{\mathbb{R}^d} \norm{a_\mu(x,\xi)}_{\mathfrak{B}(\mathcal{H}_1,\mathcal{H}_2)}^2 \diff \xi\right) \left( \int_{\mathbb{R}^d} \norm{\widehat{u}(\xi)}_{\mathcal{H}_1}^2 \diff\xi \right) \diff x \notag \\
		&\leq \frac{1}{(2\pi)^{2d}} \sup_{z \in \mathbb{R}^{2d}} \norm{a(z)}_{\mathfrak{B}(\mathcal{H}_1,\mathcal{H}_2)}^2 \int_{\mathbb{R}^d} \int_{\mathbb{R}^d} |\chi_\mu(x,\xi)|^2 \diff x\diff \xi \ \norm{\widehat{u}}_{H_1}^2 \notag \\ 
		&\leq C \sup_{z \in \mathbb{R}^{2d}} \norm{a(z)}_{\mathfrak{B}(\mathcal{H}_1,\mathcal{H}_2)}^2 \norm{u}_{H_1}^2.
	\end{align}
	The constant $C$ does not depend on $\mu$ because the integral over $(x,\xi)$ in the fourth line is translation invariant.
	
	Next, we estimate the operator norms of $A_\mu A_\nu^*$ and $A_\mu^* A_\nu$ for $\mu = (\mu_1,\mu_2), \nu = (\nu_1,\nu_2) \in \mathbb{Z}^{2d}$ and illustrate that
	\begin{align}
		\sup_{\mu \in \mathbb{Z}^{2d}} \sum_{\nu \in \mathbb{Z}^{2d}} \norm{A_\mu A_\nu^*}_{\mathfrak{B}(H_2)}^{\frac{1}{2}} \leq C \sum_{|\alpha| \leq N} \sup_{z \in \mathbb{R}^{2d}} \norm{\partial_z^\alpha a(z)}_{\mathfrak{B}(\mathcal{H}_1, \mathcal{H}_2)} \equiv C(a), \\ 
		\sup_{\mu \in \mathbb{Z}^{2d}} \sum_{\nu \in \mathbb{Z}^{2d}} \norm{A_\mu^* A_\nu}_{\mathfrak{B}(H_1)}^{\frac{1}{2}} \leq C \sum_{|\alpha| \leq N} \sup_{z \in \mathbb{R}^{2d}} \norm{\partial_z^\alpha a(z)}_{\mathfrak{B}(\mathcal{H}_1, \mathcal{H}_2)} \equiv C(a)
	\end{align}
	for some constants $C,N>0$ which do not depend on $a$.
	
	We already observed that $A_\mu$ is bounded by $C(a)$. Hence, for every $\mu \in \mathbb{Z}^{2d}$, it suffices to bound the operator norms of $A_\mu A_\nu^*$ and $A_\mu^* A_\nu$ for almost every $\nu$ and demonstrate that the operator norms are summable. Therefore, we assume $|\mu-\nu|$ is sufficiently large such that the supports of $\chi_\mu$ and $\chi_\nu$ are disjoint. To proceed further, we note that
	\begin{align}
		A_\mu A_\nu^* u(x) = \int_{\mathbb{R}^d} K_{\mu\nu}(x,y) u(y) \diff y
	\end{align}
	with the integral kernel
	\begin{align}
		K_{\mu\nu}(x,y)
		= \frac{1}{(2\pi)^d} \int_{\mathbb{R}^d} \e^{\I\scp{x-y}{\xi}} a_\mu(x,\xi) a_\nu(y,\xi)^* \diff \xi \in \mathfrak{B}(\mathcal{H}_2).
	\end{align}
	Furthermore, we observe that
	\begin{align}
		L \e^{\I\scp{x-y}{\xi}} = \e^{\I\scp{x-y}{\xi}}, \ L = \frac{1+ (x-y)D_\xi}{1+|x-y|^2}.
	\end{align}
	Hence, via integration by parts, it follows that, for any $k \in \mathbb{N}$,
	\begin{align}
		K_{\mu\nu}(x,y) = \frac{1}{(2\pi)^{d}} \int_{\mathbb{R}^d} \e^{\I\scp{x-y}{\xi}} (L^\mathrm{T})^{k}(a_\mu(x,\xi) a_\nu(y,\xi)^*) \diff \xi.
	\end{align}	
	If $(x,y) \in \supp(a_\mu(\cdot,\xi) a_\nu(\cdot,\xi)^*)$, then $x\neq y$ because we assumed that the supports of $\chi_\mu$ and $\chi_\nu$ are disjoint. Thus, a constant $c>0$ independent of $\mu,\nu$, and $\xi$ exists such that, for every $(x,y) \in \supp(a_\mu(\cdot,\xi)\overline{a}_\nu(\cdot,\xi))$,
	\begin{align}
		c^{-1}|\mu-\nu| \leq |x-y| \leq c |\mu-\nu|.
	\end{align}
	With a similar argument as in the first step of the proof of Theorem \ref{thm:L2Continuity}, we obtain the following bound:
	\begin{align}
		\int_{\mathbb{R}^d} \norm{K_{\mu\nu}(x,y)}_{\mathfrak{B}(\mathcal{H}_2)} \diff y \leq C(a) \int_{B_{\mu\nu}(x)} \int_{\substack{|\xi-\mu_2| \leq C_\chi \\ |\xi-\nu_2| \leq C_\chi}} \frac{1}{(1+|x-y|)^k} \diff \xi \diff y,
		\label{eq:CVIntegral1}
	\end{align} 
	where $B_{\mu\nu}(x)= \{y \in \mathbb{R}^d \mid c^{-1} |\mu-\nu| \leq |x-y| \leq c |\mu-\nu|\}$. The constant $C_\chi$ depends only on the support of $\chi$. The integral over $\xi$ in \eqref{eq:CVIntegral1} is bounded and independent of $\mu$ and $\nu$; hence,
	\begin{align}
		\int_{\mathbb{R}^d} \norm{K_{\mu\nu}(x,y)}_{\mathfrak{B}(\mathcal{H}_2)} \diff y
		&\leq C(a) \int_{B_{\mu\nu}(x)} \frac{1}{(1+|x-y|)^k} \diff y \notag \\
		&\leq C(a) \int_{\mathbb{R}^d} \frac{(1+|\mu-\nu|)^{d+1-k}}{(1+|x-y|)^{d+1}} \diff y \notag \\
		&\leq C(a) (1+|\mu-\nu|)^{d+1-k}.
	\end{align}
	Analogously, we confirm that integrating $\norm{K_{\mu\nu}(x,y)}_{\mathfrak{B}(\mathcal{H}_2)}$ over $x$ yields the same bound:
	\begin{align}
		\int_{\mathbb{R}^d} \norm{K_{\mu\nu}(x,y)}_{\mathfrak{B}(\mathcal{H}_2)} \diff x \leq C(a) (1+|\mu-\nu|)^{d+1-k}.
	\end{align}
	This proves that the integral kernel $K_{\mu\nu}(x,y)$ passes the Schur test (see Theorem \ref{thm:SchurTest}); thus,
	\begin{align}
		\norm{A_\mu A_\nu^*}_{\mathfrak{B}(H_2)} \leq C(a) (1+|\mu-\nu|)^{d+1-k}.
	\end{align}	
	Similarly, the operator norm of $A_\mu^* A_\nu$ is bounded:
	\begin{align}
		\norm{A_\mu^* A_\nu}_{\mathfrak{B}(H_1)} \leq C(a) (1+|\mu-\nu|)^{d+1-k}.
	\end{align}
	If $k \in \mathbb{N}$ is chosen sufficiently large, $(1+|\mu-\nu|)^{d+1-k}$ is summable over $\nu \in \mathbb{Z}^{2d}$ uniformly in $\mu \in \mathbb{Z}^{2d}$. From the Cotlar--Stein Lemma it follows that $\sum_\mu A_\mu$ converges in the strong operator topology to a bounded operator $A:L^2(\mathbb{R}^d,\mathcal{H}_1) \to L^2(\mathbb{R}^d,\mathcal{H}_2)$ with $\norm{A}\leq C(a)$. The action of $A$ on $C_\mathrm{c}^\infty(\mathbb{R}^d,\mathcal{H}_1)$ coincides with that of $\Op(a)$. Thus, $\Op(a)$ extends to a bounded operator from $L^2(\mathbb{R}^d,\mathcal{H}_1)$ to $L^2(\mathbb{R}^d,\mathcal{H}_2)$.
\end{proof}

\begin{cor}
	If $a \in S(M,\mathfrak{B}(\mathcal{H}_1,\mathcal{H}_2))$, then $\Op(a)$ is relatively bounded with respect to $M(D)$. Furthermore, if $a$ is elliptic, then $\Op(a^{-1})$ is relatively bounded with respect to $M(D)^{-1}$.
	\label{cor:RelativeBoundednessPseudor}
\end{cor}

\begin{proof}
	It is not a restriction to assume that $\Op(a)$ is given in standard quantisation. A simple computation demonstrates that
	\begin{align}
		\Op(a)u(x) &= \frac{1}{(2\pi)^d} \int_{\mathbb{R}^d} \e^{\I \scp{x}{\xi}} a(x,\xi)M(\xi)^{-1} M(\xi) \widehat{u}(\xi) \diff \xi \notag \\
		&= \frac{1}{(2\pi)^d} \int_{\mathbb{R}^d} \e^{\I \scp{x}{\xi}} a(x,\xi)M(\xi)^{-1} \widehat{M(D) u}(\xi) \diff \xi \notag \\ 
		&= \Op\left(aM^{-1}\right) M(D)u(x);
		\label{eq:RelativeBoundednessPseudor}
	\end{align}
	hence, the first claim is a consequence of Theorem \ref{thm:CalderonVaillancourtOV} because $aM^{-1} \in S(1,\mathfrak{B}(\mathcal{H}_1,\mathcal{H}_2))$.
	
	If $a$ is elliptic, then $a^{-1} \in S(M^{-1},\mathfrak{B}(\mathcal{H}_1,\mathcal{H}_2))$. Because $a^{-1}(x,\xi)M(\xi)$ is bounded by 1, a similar computation as in \eqref{eq:RelativeBoundednessPseudor} verifies that $\Op(a^{-1})$ is relatively bounded with respect to $M(D)^{-1}$.
\end{proof}
\chapter{Removal of the Cut-off}
\label{ch:RemovalCutOff}

In Section \ref{sec:NelsonModelConstantCoefficients}, we argued that the formal Nelson Hamiltonian must be regularised to obtain a well-defined Hamiltonian. We proved in Proposition \ref{prop:CutOffHamiltonianSelfAdjoint} that putting an ultraviolet cut-off $\Lambda$ on the interaction leads to a self-adjoint Hamiltonian $H_\Lambda$. In this chapter, we remove the cut-off using renormalisation to obtain the self-adjoint renormalised Nelson Hamiltonian. %Nelson first demonstrated this renormalisation in the Nelson model with constant coefficients in 1964 (cf. \cite{nelson1964} or Theorem \ref{thm:Nelson}). We prove the analogous theorem for the Nelson model with variable coefficients. 

\begin{ntn}
	We write $f \in C^\infty_{\mathrm{b}}(\mathbb{R}^3)$ if $f$ is a smooth function and if all derivatives of $f$ (including the function $f$ itself) are bounded.
\end{ntn}

\begin{thm}
	Let $H_\Lambda$ be the Nelson Hamiltonian with variable coefficients and cut-off defined in Section \ref{sec:NelonModelVariableCoeff}. Assume that $g^{jk}, \mu^2, W \in C^\infty_{\mathrm{b}}(\mathbb{R}^3)$. Then particle potentials $-E_\Lambda(X)$ exist such that $H_\Lambda + E_\Lambda(X)$ converges in the norm resolvent sense to a self-adjoint bounded from below operator $H$ as $\Lambda \to \infty$.
	\label{thm:Gerard}
\end{thm}
The theorem in the above formulation was first proven by Gérard \textit{et al.}~\cite{gerard2011}, and the proof we present follows their paper. Furthermore, we include ideas from Nelson's original work \cite{nelson1964} and from Griesemer and Wünsch \cite{griesemer2016, griesemer2018}. The latter proved that the convergence holds in the norm resolvent sense in the Nelson model with constant coefficients, and their technique applies to the Nelson model with variable coefficients. The main steps in the proof of Theorem \ref{thm:Gerard} are the following.
\begin{enumerate}
	\item Realise that the relevant operators in the Nelson model with variable coefficients are pseudo-differential operators. 
	\item Define a suitable unitary Gross transformation $U_\Lambda$ and compute the transformed cut-off Hamiltonian $U_\Lambda H_\Lambda U_\Lambda^*$.
	\item Identify a particle potential $-E_\Lambda(X)$ that is divergent as $\Lambda \to \infty$. Remove this divergent potential from the transformed cut-off Hamiltonian.
	\item Show that $H_\Lambda' = U_\Lambda H_\Lambda U_\Lambda^* + E_\Lambda(X)$ defines a family $(q_\Lambda)_{\Lambda\geq 0}$ of quadratic forms. Prove that the limit $q_\infty$ exists and that the corresponding operator $H_\infty'$ defines the required self-adjoint operator $H = U_\infty^* H_\infty' U_\infty$.
\end{enumerate}

\section{Pseudo-differential Operators in Nelson Model}
\label{sec:PseudorsInNelsonModelI}

The particle Hamiltonian $K = K_0 + W(X)$ and the square of the dispersion relation $\omega^2 = h$ can be treated as pseudo-differential operators:
\begin{align}
	K = K_0 + W(X) &= K_0^\mathrm{w}(X, D_X) + W(X) = K^\mathrm{w}(X, D_X), \\
	h = h_0 + \mu(x)^2 &= h_0^\mathrm{w}(x,D_x) + \mu(x)^2 = h^\mathrm{w}(x,D_x).
\end{align}
The corresponding symbols
\begin{align}
	K(X,\Xi) &= K_0(X,\Xi) + W(X) = \Xi \cdot g(X)\Xi + W(X), \\
	h(x,\xi) &= h_0(x,\xi) + \mu(x)^2 = \xi \cdot g(x)\xi + \mu(x)^2
\end{align}
are under the assumptions of Theorem \ref{thm:Gerard} elements of $S^2$. Furthermore, these symbols are elliptic because $g$ is uniformly elliptic.

\subsubsection*{Dispersion relation}

The dispersion relation $\omega$ is the positive square root of $h$. The square root of a pseudo-differential operator is not necessarily a pseudo-differential operator because the square root is not smooth in $0$. However, $h \geq m^2 > 0$ is bounded from below by the square of the (by assumption) positive bosonic field mass $m^2$. If $f\in S^{1/2}(\mathbb{R})$ (see \eqref{eq:Sp} for the definition of $S^{p}(\mathbb{R})$) is such that $f(t) = \sqrt{t}$ for $t \geq m^2/2$, then $f$ and $t\mapsto \sqrt{t}$ agree on the spectrum of $h$; hence,
\begin{align}
	\omega = (h^\mathrm{w})^\frac{1}{2} = f(h^\mathrm{w}).
\end{align}
According to the functional calculus (see Theorem \ref{thm:FunctionalCalculus}), the operator $\omega = f(h^\mathrm{w})$ is a pseudo-differential operator of order 1. Its principal symbol is $f(h(x,\xi)) = h(x,\xi)^{1/2}$, that is, $f(h)^\mathrm{w} - \omega = f(h)^\mathrm{w} - f(h^\mathrm{w}) \in \mathrm{Op}(S^0)$.
If we prefer to use the standard quantisation instead, we can also write $\mathrm{Op}(f(h)) - \omega \in \Op(S^0)$ due to Corollary \ref{cor:ChangingQuantisations}. Similarly, we confirm that $\omega^k$ is a pseudo-differential operator of order $k$ for every $k \in \mathbb{R}$. 

\subsubsection*{Free Nelson Hamiltonian}

The free Nelson Hamiltonian $H_0 = K \otimes \mathbb{1} + \mathbb{1} \otimes \diff\Gamma(\omega) \equiv K + \diff\Gamma(\omega)$ restricted to the $n$-boson subspace $\mathfrak{H}^{(n)}$ is a pseudo-differential operator.\footnote{Hereafter, if it is clear from the context, we omit tensor products with the identity operator.} Its symbol is elliptic and an element of the symbol class $S(M,g)$ with the order function
\begin{align}
	M(\Xi,\xi) = \jap{\Xi}^2 + \sum_{j=1}^n \jap{\xi_j} \equiv \jap{\Xi}^2 + \Omega(\xi),
	\label{eq:OrderFunctionH0}
\end{align}
where $\Omega(\xi) = \sum_{j=1}^n \jap{\xi_j}$, and the metric
\begin{align}
	g_{(X,\Xi,\xi,x)} = \diff X^2 + \jap{\Xi}^{-2} \diff \Xi^2 + \sum_{j=1}^n (\diff x_j^2 + \jap{\xi_j}^{-2} \diff \xi_j^2).
	\label{eq:MetricH0}
\end{align}
A peculiarity of the Nelson model is that the symbol is quadratic in $\Xi$ but linear in $\xi$. In a fully relativistic theory, the symbol is of linear order in $\Xi$ and $\xi$. The additional order in $\Xi$ is vital for our analysis because it makes some otherwise divergent integrals convergent.

\subsubsection*{Form factor}

We denote the regularised form factor of the field operator $\Phi(\omega^{-1/2}\rho_{\Lambda,X}) = a^*(v_{\Lambda,X}) + a(v_{\Lambda,X})$ by $v_{\Lambda,X} = \omega^{-1/2}\rho_{\Lambda,X}/\sqrt{2}$, where the cut-off function $\rho_{\Lambda,X}$ is defined in \eqref{eq:CutOffFunction}. Above, we observed that $\omega^{-1/2}$ is a pseudo-differential operator of order $-1/2$. Let $\sigma(\omega^{-1/2}) \in S^{-1/2}$ be its symbol in standard quantisation. Then the form factor $v_{\Lambda,X}$ can be written as follows:
\begin{align}
	v_{\Lambda,X}(x) = \frac{1}{\sqrt{2}} \int_{\mathbb{R}^3} \frac{\e^{\I\scp{x-X}{\xi}}}{(2\pi)^3} \sigma(\omega^{-\frac{1}{2}})(x,\xi) \widehat{\rho}(\xi/\Lambda) \diff \xi.
\end{align}
We extract the leading contribution of the above integral by expanding the symbol $\sigma(\omega^{-\frac{1}{2}})(x,\xi)$ in $x=X$:
\begin{align}
	\sigma(\omega^{-\frac{1}{2}})(x,\xi) = \sigma(\omega^{-\frac{1}{2}})(X,\xi) - \sqrt{2} (x-X) \cdot r(x,\xi),
\end{align}
where $r\in [S^{-1/2}]^3 = S^{-1/2} \times S^{-1/2} \times S^{-1/2}$. The form factor decomposes: 
\begin{align}
	v_{\Lambda,X} = u_{\Lambda,X} + \tilde{u}_{\Lambda,X}
\end{align}
with
\begin{align}
	u_{\Lambda,X}(x) = \frac{1}{\sqrt{2}} \int_{\mathbb{R}^3} \frac{\e^{\I\scp{x-X}{\xi}}}{(2\pi)^3} \sigma(\omega^{-\frac{1}{2}})(X,\xi) \widehat{\rho}(\xi/\Lambda) \diff \xi
	\label{eq:FormFactorLeadingOrder}
\end{align}
and
\begin{align}
	\tilde{u}_{\Lambda,X}(x) &= -\int_{\mathbb{R}^3} \frac{\e^{\I\scp{x-X}{\xi}}}{(2\pi)^3} (x-X) \cdot r(x,\xi) \widehat{\rho}(\xi/\Lambda) \diff \xi \notag \\
	&= \int_{\mathbb{R}^3} \frac{\e^{\I\scp{x-X}{\xi}}}{(2\pi)^3} (D_\xi \cdot r)(x,\xi) \widehat{\rho}(\xi/\Lambda) \diff \xi
	+ \frac{1}{\Lambda} \int_{\mathbb{R}^3}  \frac{\e^{\I\scp{x-X}{\xi}}}{(2\pi)^3} r(x,\xi) \cdot (D_\xi \widehat{\rho})(\xi/\Lambda) \diff \xi \notag \\
	&= \Op(D_\xi \cdot r)\rho_{\Lambda,X}(x) + \frac{1}{\Lambda}\Op(r) \tau_{\Lambda,X}(x),
\end{align}
where $\tau(x) = x\rho(x)$ is a vector of Schwartz functions.

\section{Gross Transformation}

We proceed with the second step of the proof and find a unitary transformation of the cut-off Hamiltonian that reveals a divergent potential in the Nelson Hamiltonian. 

Assume for a moment that the nonrelativistic particle has infinite mass (i.e. its position in space is fixed, and its kinetic energy is negligible). The cut-off Hamiltonian, in this case, reduces to
\begin{align}
	H^\circ_\Lambda = \diff\Gamma(\omega) + \Phi(\omega^{-\frac{1}{2}}\rho_{\Lambda,X}).
\end{align}
Recall the definition and properties of the unitary Weyl operators $V(f) = \e^{\I\Pi(f)}$ from Subsection \ref{ssec:FieldWeylOperators}. We compute the \textbf{Gross transformed} Hamiltonian $V(f)H_\Lambda^\circ V(f)^*$. From Proposition \ref{prop:PropertiesWeylOperator} it follows that
\begin{align}
	\e^{\I\Pi(f)} H^\circ_\Lambda \e^{-\I\Pi(f)} &= \diff \Gamma(\omega) + \Phi(\omega f + \omega^{-\frac{1}{2}}\rho_{\Lambda,X}) + \frac{1}{2} \norm{\omega^{\frac{1}{2}}f}^2_{\mathfrak{h}} + \Re \scp{f}{\omega^{-\frac{1}{2}}\rho_{\Lambda,X}}_\mathfrak{h}.
\end{align}
If $f=-\omega^{-\frac{3}{2}}\rho_{\Lambda,X}$, then
\begin{align}
	\e^{\I\Pi(\omega^{-\frac{3}{2}}\rho_{\Lambda,X})} H^\circ_\Lambda \e^{-\I\Pi(\omega^{-\frac{3}{2}}\rho_{\Lambda,X})} = \diff\Gamma(\omega) - \frac{1}{2}\norm{\omega^{-1}\rho_{\Lambda,X}}^2_\mathfrak{h}.
	\label{eq:GrossTransformation}
\end{align}
Thus, the Gross transformation reveals the term $-\norm{\omega^{-1}\rho_{\Lambda,X}}^2_{\mathfrak{h}}/2$, which diverges as $\Lambda \to \infty$. This term must be removed to renormalise the Nelson model with a particle of infinite mass.

We perform a similar Gross transformation for the full cut-off Hamiltonian $H_\Lambda$. The correct generalisation of the function $f=-\omega^{-\frac{3}{2}}\rho_{\Lambda,X}$ turns out to be
\begin{align}
	B_{\Lambda, X}^\sigma(x) = -(K+\omega)^{-1}\omega^{-\frac{1}{2}}\rho_{\Lambda, X}^\sigma(x).
	\label{eq:B}
\end{align}
We added a superscript $\sigma$ in the cut-off function $\rho_{\Lambda}$. This is an infrared cut-off and is defined via
\begin{align}
	\widehat{\rho_\Lambda^\sigma}(\xi) = \chi(\xi) \widehat{\rho}(\xi/\Lambda),
	\label{eq:InfraredCutOff}
\end{align}
where $\chi \in C^{\infty}(\mathbb{R}^3)$, $0\leq\chi\leq 1$, is any smooth symmetric function equal to $0$ if $|\xi| \leq \sigma$ and equal to $1$ if $|\xi| \geq 2\sigma$. It becomes clear why this cut-off is needed at the end of the proof. 

We collect analytical properties of the function $B_{\Lambda, X}^\sigma$. First, for every $X \in \mathbb{R}^3$, $B_{\Lambda, X}^\sigma$ is a Schwartz function because pseudo-differential operators map Schwartz functions to Schwartz functions (see Theorem \ref{thm:PseudorsSchwartzSpaceMapping}). Moreover, $B_{\Lambda, X}^\sigma$ is real-valued because $\rho_{\Lambda,X}^\sigma$ is real-valued and because the operators $K$ and $\omega$ are real operators.

\begin{lem}
	For every $\alpha < 1$, $\omega^\alpha B_{\Lambda, X}^\sigma$, $\omega^{\alpha -1} \partial_X B_{\Lambda, X}^\sigma \in \mathfrak{h} = L^2(\mathbb{R}^3,\diff x)$, and an $s<-3/2$ exists such that
	\begin{align}
		\norm{\omega^\alpha B_{\Lambda, X}^\sigma}_{\mathfrak{h}} &\leq C \norm{\rho_\Lambda^\sigma}_{H^{s}}, \label{eq:B1} \\
		\norm{\omega^{\alpha -1} \partial_X B_{\Lambda, X}^\sigma}_{\mathfrak{h}} &\leq C \norm{\rho_\Lambda^\sigma}_{H^{s}} \label{eq:B2} 
	\end{align}
	for a constant $C$ independent of $\Lambda$, $\sigma$, and $X$. Here $H^s$ is the Sobolev space $H^s(\mathbb{R}^3,\diff x)$.
	\label{lem:B}
\end{lem}

\begin{proof}
	Let $t^{-1} \in S(M^{-1},g)$ be the symbol of the pseudo-differential operator $(K+\omega)^{-1}$ in standard quantisation with order function $M$ and metric $g$ defined in \eqref{eq:OrderFunctionH0} and \eqref{eq:MetricH0} with $n=1$. Then,
	\begin{align}
		(K+\omega)^{-1}\rho_{\Lambda,X}^\sigma(x) &= \frac{1}{(2\pi)^3} \int_{\mathbb{R}^3} \int_{\mathbb{R}^3} \e^{\I \scp{X}{\Xi} + \I \scp{x}{\xi}} t^{-1}(X,\Xi,x,\xi) \delta(\Xi+\xi) \widehat{\rho^\sigma}(\xi/\Lambda) \diff\Xi \diff\xi \notag \\
		&= \frac{1}{(2\pi)^3} \int_{\mathbb{R}^3} \e^{\I \scp{x-X}{\xi}} t^{-1}(X,-\xi,x,\xi) \widehat{\rho^\sigma}(\xi/\Lambda) \diff\xi.
	\end{align}	
	For $X \in \mathbb{R}^3$ fixed, $t_X^{-1}(x,\xi) = t^{-1}(X,-\xi,x,\xi)$ defines a symbol $t_X^{-1} \in S^{-2}$. The form of the symbol $t_X^{-1}$ is irrelevant for this proof, but for later applications, we note that $t_X^{-1}$ is given by
	\begin{align}
		t^{-1}_X(x,\xi) &= (K(X,\xi) + \sigma(\omega)(x,\xi))^{-1} + S^{-3} \notag \\
		&= (K_0(X,\xi) + 1)^{-1} + S^{-3}.
		\label{eq:SymbolTX}
	\end{align} 
	Because $\omega^{-1/2} \in \Op(S^{-1/2})$, a symbol $b_X \in S^{-5/2}$ exists such that
	\begin{align}
		B_{\Lambda,X}^\sigma = (K+\omega)^{-1} \omega^{-\frac{1}{2}} \rho_{\Lambda,X}^\sigma = \omega^{-\frac{1}{2}}\Op(t_X^{-1})\rho_{\Lambda,X}^\sigma = \Op(b_X)\rho_{\Lambda,X}^\sigma.
	\end{align}	
	Due to Corollary \ref{cor:PseudorsSobolevMapping}, we obtain the following inequality for every $s \in \mathbb{R}$:
	\begin{align}
		\norm{\omega^\alpha \Op(b_X)\rho_{\Lambda,X}}_{H^{s-m}} \leq C \norm{\rho_{\Lambda,X}^\sigma}_{H^{s}} = C \norm{\rho_{\Lambda}^\sigma}_{H^{s}},
	\end{align}
	where $m = -5/2+\alpha < -3/2$. The inequality \eqref{eq:B1} is proven if we choose $s=m$. 
	
	To prove \eqref{eq:B2}, we observe that $\partial_X \rho_{\Lambda,X}^\sigma(x) = -\partial_x \rho_{\Lambda,X}^\sigma(x)$ and
	\begin{align}
		\partial_X B_{\Lambda,X}^\sigma = \Op(\partial_X b_X)\rho_{\Lambda,X}^\sigma - \Op(b_X)\partial_x \rho_{\Lambda,X}^\sigma
	\end{align}
	with $\partial_X b_X \in S^{-5/2}$ and $\norm{\partial_x \rho_{\Lambda,X}^\sigma}_{H^{s}} \leq \norm{\rho_{\Lambda,X}^\sigma}_{H^{s+1}}$. The rest of the proof is similar to the proof of the first inequality.
\end{proof}

We return to the transformation of the cut-off Hamil\-tonian $H_\Lambda$. The Gross trans\-for\-mation of $H_\Lambda$ is realised using the unitary operator
\begin{align}
	U_\Lambda = V(B_{\Lambda,X}^\sigma) = \e^{\I \Pi(B_{\Lambda,X}^\sigma)}.
\end{align}
Computing the transformed operator $U_\Lambda H_\Lambda U_\Lambda^*$, however, is more tedious than for $H^\circ_\Lambda$. At the end of a careful computation, we arrive at
\begin{align}
	U_\Lambda H_\Lambda U_\Lambda^* = H_0 + R_\Lambda(X) + V_\Lambda(X)
\end{align}
with $R_\Lambda(X)$ and $V_\Lambda(X)$ to be specified below.

\begin{ntn}
	For notational convenience, we drop the sub\-scripts in $\rho = \rho_{\Lambda,X}$, $\rho^\sigma = \rho_{\Lambda,X}^\sigma$, and $B^\sigma = B_{\Lambda,X}^\sigma$ for the rest of this section. 
\end{ntn}

The transformation of the free field Hamiltonian and the interacting field are known from the beginning of this section:
\begin{align}
	U_\Lambda \diff\Gamma(\omega) U_\Lambda^* &= \diff\Gamma(\omega) + \Phi(\omega B^\sigma) + \frac{1}{2} \norm{\omega^{\frac{1}{2}}B^\sigma}^2_\mathfrak{h}, \\
	U_\Lambda \Phi(\omega^{-\frac{1}{2}}\rho) U_\Lambda^* &= \Phi(\omega^{-\frac{1}{2}}\rho) + \scp{B^\sigma}{\omega^{-\frac{1}{2}} \rho}_\mathfrak{h}.
\end{align}
The particle Hamiltonian is defined as $K=K_0+W(X)$. Because the potential $W(X)$ does not depend on the boson variables, it commutes with $U_\Lambda$ and remains invariant under the unitary transformation. To transform the free particle Hamiltonian $K_0$, observe that
\begin{align}
	U_\Lambda \partial_{X} U_\Lambda^* = \partial_{X} - \I \Pi(\partial_{X}B^\sigma). \label{eq:TransformationPartialDerivative}
\end{align}
We used the fact that $\Pi(B^\sigma)$ and $\Pi(\partial_{X}B^\sigma)$ commute according to the commutation relation \eqref{eq:CommutationRelationConjugateMomentum} because $B^\sigma$ and $\partial_{X}B^\sigma$ are real-valued functions. Thus,
\begin{align}
	U_\Lambda K_0 U_\Lambda^* &= -[\partial_{X} - \I\Pi(\partial_{X}B^\sigma)] \cdot g(X) [\partial_{X} - \I \Pi(\partial_{X}B^\sigma)] \notag \\
	&= K_0 - \I \Pi(K_0 B^\sigma) + \Pi(\partial_{X}B^\sigma) \cdot g(X) \Pi(\partial_X B^\sigma) \notag \\
	&\ \ \ +2 \I\Pi(\partial_X B^\sigma) \cdot g(X) \partial_X. \label{eq:ParticleHamiltonianTransformed}
\end{align}
We expand and simplify the right-hand side of the above equality. First, the conjugate momentum operator $\Pi$ and field operator $\Phi$ are related by
\begin{align}
	-\I \Pi(K_0B^\sigma) = \Phi(K_0B^\sigma) - \sqrt{2} a(K_0B^\sigma).
\end{align}
Moreover, expanding the third summand in \eqref{eq:ParticleHamiltonianTransformed} leads to 
\begin{align}
	&\Pi(\partial_{X}B^\sigma) \cdot g(X) \Pi(\partial_X B^\sigma) \notag \\
	&\ \ \ = -\frac{1}{2} a^*(\partial_X B^\sigma) \cdot g(X) a^*(\partial_X B^\sigma) -\frac{1}{2} a(\partial_X B^\sigma) \cdot g(X) a(\partial_X B^\sigma) \notag \\
	&\ \ \ \ \ \ + a^*(\partial_X B^\sigma) \cdot g(X) a(\partial_X B^\sigma) + \frac{1}{2} \scp{\partial_X B^\sigma}{g(X)\partial_X B^\sigma}_\mathfrak{h}. \label{eq:ThirdTermTransformedParticleHamiltonian}
\end{align}

\begin{rem}
	In \eqref{eq:ThirdTermTransformedParticleHamiltonian}, we use the commutation relation \eqref{eq:CommutatorAnniCrea} to obtain \textbf{normal ordering}, that is, creation operators are to the left of annihilation operators. This ordering is practical because it avoids the appearance of creation operators if \eqref{eq:ThirdTermTransformedParticleHamiltonian} is interpreted in the sense of quadratic forms. The additional term $\scp{\partial_X B^\sigma}{g(X)\partial_X B^\sigma}_\mathfrak{h}$, however, is divergent as $\Lambda \to \infty$.
\end{rem}

Combining the fourth term in \eqref{eq:ParticleHamiltonianTransformed} with $-\sqrt{2} a(K_0B^\sigma) = \sqrt{2} a(\partial_X \cdot g(X)\partial_X B^\sigma)$ from the second term yields
\begin{align}
	&2 \I\Pi(\partial_X B^\sigma) \cdot g(X) \partial_X -\sqrt{2} a(K_0B^\sigma) \notag \\ 
	&\ \ \ = -\sqrt{2} a^*(\partial_X B^\sigma) \cdot g(X) \partial_X + \sqrt{2} a(\partial_X B^\sigma) \cdot g(X) \partial_X -\sqrt{2} a(K_0B^\sigma) \notag \\
	&\ \ \ = -\sqrt{2} a^*(\partial_X B^\sigma) \cdot g(X) \partial_X + \sqrt{2}\partial_X \cdot g(X) a(\partial_X B^\sigma).
\end{align}
Altogether, we obtain the following for the transformed cut-off Hamiltonian:
\begin{align}
	U_\Lambda H_\Lambda U_\Lambda^* &= K + \diff\Gamma(\omega) + \Phi(\omega^{-\frac{1}{2}}\rho + (K_0 + \omega)B^\sigma) \notag \\
	&\ \ \ - \sqrt{2} a^*(\partial_X B^\sigma) \cdot g(X)\partial_X + \sqrt{2} \partial_X \cdot g(X) a(\partial_X B^\sigma) \notag \\ 
	&\ \ \ - \frac{1}{2} a^*(\partial_X B^\sigma) \cdot g(X) a^*(\partial_X B^\sigma) - \frac{1}{2} a(\partial_XB^\sigma)\cdot g(X) a(\partial_X B^\sigma) \notag \\
	&\ \ \ + a^*(\partial_X B^\sigma) \cdot g(X) a(\partial_X B^\sigma) \notag \\
	&\ \ \ + \frac{1}{2}\norm{\omega^{\frac{1}{2}} B^\sigma}^2_\mathfrak{h} + \scp{B^\sigma}{\omega^{-\frac{1}{2}}\rho}_\mathfrak{h} + \frac{1}{2} \scp{\partial_X B^\sigma}{g(X)\partial_X B^\sigma}_\mathfrak{h}. 
	\label{eq:HamiltonianTransformed}
\end{align}
The form factor of the field operator $\Phi$ is $\omega^{-\frac{1}{2}}(\rho - \rho^\sigma)$ according to the definition of $B^\sigma$ in \eqref{eq:B}. We combine this field operator with the second to fourth line to $R_\Lambda(X)$ (reinstating subscripts $\Lambda$ and $X$):
\begin{align}
	R_\Lambda(X) = &- \sqrt{2} a^*(\partial_X B_{\Lambda,X}^\sigma) \cdot g(X)\partial_X + \sqrt{2} \partial_X \cdot g(X) a(\partial_X B_{\Lambda,X}^\sigma) \notag \\ 
	&- \frac{1}{2} a^*(\partial_X B_{\Lambda,X}^\sigma) \cdot g(X) a^*(\partial_X B_{\Lambda,X}^\sigma) \notag \\
	&- \frac{1}{2} a(\partial_XB_{\Lambda,X}^\sigma)\cdot g(X) a(\partial_X B_{\Lambda,X}^\sigma) \notag \\
	&+ a^*(\partial_X B_{\Lambda,X}^\sigma) \cdot g(X) a(\partial_X B_{\Lambda,X}^\sigma) 
	+ \Phi(\omega^{-\frac{1}{2}}(\rho_{\Lambda,X} - \rho_{\Lambda,X}^\sigma)).
	\label{eq:RLambda}
\end{align}
In Section \ref{sec:ProofOfNelsonsTheorem}, we show that $R_\Lambda(X)$ defines a quadratic form on $D(H_0^{1/2})$. We abbreviate the fifth line in \eqref{eq:HamiltonianTransformed} by $V_\Lambda(X)$:
\begin{align}
	V_{\Lambda}(X) &= \frac{1}{2}\norm{(K+\omega)^{-1}\rho^\sigma_{\Lambda,X}}^2_{\mathfrak{h}} -\scp{\rho^\sigma_{\Lambda,X}}{(K+\omega)^{-1}\omega^{-1}\rho_{\Lambda,X}}_{\mathfrak{h}} \notag \\
	&\ \ \ + \frac{1}{2} \scp{\omega^{-1} \partial_{X}(K+\omega)^{-1}\rho^\sigma_{\Lambda,X}}{g(X)\partial_{X}(K+\omega)^{-1}\rho^\sigma_{\Lambda,X}}_{\mathfrak{h}}.
	\label{eq:VLambda}
\end{align}
The potential $V_\Lambda(X)$ is chosen such that it contains all divergent terms. In the next section, we identify these divergent terms.

\section{Renormalisation}
\label{sec:Renormalisation}

The discussions in this section are not strictly necessary for the proof of Theorem \ref{thm:Gerard} if we choose $-E_\Lambda(X) = V_\Lambda(X)$. Nevertheless, this section provides a physical interpretation of the term $V_\Lambda(X)$ and some useful mathematical techniques. 

In Section \ref{sec:NelsonModelConstantCoefficients}, we argued that it makes sense to remove the vacuum energy of the formal Nelson Hamiltonian. The minimum of the spectrum of the free Nelson Hamiltonian $H_0$ is 0, and $0\in \sigma_\mathrm{ess}(H_0)$ lies in the essential spectrum of $H_0$. If $(\psi_k)_{k\in \mathbb{N}} \subset L^2(\mathbb{R}^3)$ is a singular Weyl sequence for $(-\Delta_X,0)$, that is, $(\psi_k)_{k\in \mathbb{N}}$ is an orthonormal family and $\norm{-\Delta_X \psi_k}_{L^2} \to 0$ as $k\to\infty$, then $(\psi_k \otimes \Omega)_{k\in \mathbb{N}}$ is a singular Weyl sequence for $(H_0,0)$, where $\Omega$ is the Fock space vacuum. We interpret $\psi_k \otimes \Omega$ as an approximate ground state of the free Nelson Hamiltonian for a sufficiently large $k$. Perturbation theory in leading order suggests that the perturbation $\Phi(\omega^{-1/2}\rho_{\Lambda,X})$ of the cut-off Hamiltonian $H_\Lambda = H_0 + \Phi(\omega^{-1/2}\rho_{\Lambda,X})$ causes the following energy shift:
\begin{align}
	&\scp{\psi_k \otimes \Omega}{\Phi(\omega^{-\frac{1}{2}}\rho_{\Lambda,X}) \psi_k \otimes \Omega}_\mathfrak{H} - \scp{\psi_k \otimes \Omega}{\Phi(\omega^{-\frac{1}{2}}\rho_{\Lambda,X}) H_0^{-1} \Phi(\omega^{-\frac{1}{2}}\rho_{\Lambda,X}) \psi_k \otimes \Omega}_\mathfrak{H} \notag \\
	&\ \ \ = -\frac{1}{2} \int_{\mathbb{R}^3} \scp{\Omega}{a(\omega^{-\frac{1}{2}}\rho_{\Lambda,X}) H_0^{-1} a^*(\omega^{-\frac{1}{2}}\rho_{\Lambda,X})\Omega}_\mathfrak{h} |\psi_k(X)|^2 \diff X \notag \\
	&\ \ \ = -\frac{1}{2} \int_{\mathbb{R}^3} \scp{\omega^{-\frac{1}{2}} \rho_{\Lambda,X}}{(K+\omega)^{-1}\omega^{-\frac{1}{2}} \rho_{\Lambda,X}}_\mathfrak{h} |\psi_k(X)|^2 \diff X.
	\label{eq:PerturbationTheoreticExpansion}
\end{align}
We call $\scp{\omega^{-\frac{1}{2}} \rho_{\Lambda,X}}{(K+\omega)^{-1} \omega^{-\frac{1}{2}} \rho_{\Lambda,X}}_\mathfrak{h}/2$ the \textbf{(boson) vacuum energy}. Using only the leading order part $u_{\Lambda,X}$ in \eqref{eq:FormFactorLeadingOrder} of the form factor $v_{\Lambda,X} = \omega^{-\frac{1}{2}}\rho_{\Lambda,X}/\sqrt{2}$ and the symbol $t_X^{-1}$ in \eqref{eq:SymbolTX}, the boson vacuum energy in leading order is equal to
\begin{align}
	-E_\Lambda(X) &= -\frac{1}{2(2\pi)^3} \int_{\mathbb{R}^3} \frac{|\sigma(\omega^{-\frac{1}{2}})(X,\xi)|^2}{K_0(X,\xi)+1} \left|\widehat{\rho}\left(\frac{\xi}{\Lambda}\right)\right|^2 \diff\xi \notag \\
	&= -\frac{1}{2(2\pi)^3} \int_{\mathbb{R}^3} \frac{(h_0(X,\xi)+1)^{-\frac{1}{2}}}{K_0(X,\xi)+1} \left|\widehat{\rho}\left(\frac{\xi}{\Lambda}\right)\right|^2 \diff\xi + F_\Lambda(X), 
	\label{eq:VacuumEnergy}
\end{align}
where $(F_\Lambda)_{\Lambda\geq 0} \subset C_\mathrm{b}^\infty(\mathbb{R}^3)$ is a family of functions converging uniformly to some $F \in C_\mathrm{b}(\mathbb{R}^3)$ as $\Lambda \to \infty$. In the following, any family of functions (not necessarily the same on different lines) that has this property is denoted by $(F_\Lambda)_{\Lambda\geq 0}$. %In the following theorem, we prove that $-E_\Lambda(X)$ is the divergent part of $V_\Lambda(X)$. In more singular models, however, we must also consider higher order terms of the vacuum energy.

\begin{thm}
	A bounded continuous potential $V_{\mathrm{ren}} \in C_\mathrm{b}(\mathbb{R}^3)$ exists such that $V_\Lambda(X) + E_\Lambda(X)$ converges uniformly in $X$ to $V_{\mathrm{ren}}(X)$ as $\Lambda \to \infty$.
	\label{thm:RenormalisedPotential}
\end{thm}

\begin{proof}
	%In \eqref{eq:ConvergenceCutOffToDelta}, we proved the convergence of $\rho_{\Lambda,X}$ to $\delta_X$ in $H^{s}(\mathbb{R}^3)$ for any $s<-3/2$. Similarly, $\partial_x \rho_{\Lambda,X} \to \partial_x\delta_X$ in $H^s(\mathbb{R}^3)$ for any $s<-5/2$. Both convergences are uniform in $X$. 
	From the proof of Lemma \ref{lem:B} it follows that a symbol $t^{-1}_X \in S^{-2}$ exists such that
	\begin{align}
		(K+\omega)^{-1} \rho_{\Lambda,X}^\sigma &= \Op(t_X^{-1}) \rho_{\Lambda,X}^\sigma, \\
		\partial_X (K+\omega)^{-1} \rho_{\Lambda,X}^\sigma &= \Op(\partial_X t^{-1}_X)\rho_{\Lambda,X}^\sigma - \Op(t^{-1}_X)\partial_x\rho_{\Lambda,X}^\sigma.
	\end{align}
	We insert these identities in \eqref{eq:VLambda}:
	\begin{align}
		V_\Lambda(X) &= \frac{1}{2} \norm{\Op(t^{-1}_X)\rho_{\Lambda,X}^\sigma}^2_\mathfrak{h} 
		-\scp{\rho_{\Lambda,X}^\sigma}{\Op(t^{-1}_X)\omega^{-1}\rho_{\Lambda,X}}_\mathfrak{h}  \notag \\
		&\ \ \ + \frac{1}{2} \scp{\omega^{-1}\Op(\partial_X t^{-1}_X)\rho_{\Lambda,X}^\sigma}{g(X)\Op(\partial_X t^{-1}_X)\rho_{\Lambda,X}^\sigma}_\mathfrak{h}  \notag \\
		&\ \ \ -\scp{\omega^{-1}\Op(\partial_X t^{-1}_X) \rho_{\Lambda,X}^\sigma}{g(X) \Op(t^{-1}_X) \partial_x \rho_{\Lambda,X}^\sigma}_\mathfrak{h}  \notag \\
		&\ \ \ + \frac{1}{2} \scp{\omega^{-1}\Op(t^{-1}_X) \partial_{x} \rho_{\Lambda,X}^\sigma}{g(X) \Op(t^{-1}_X) \partial_{x} \rho_{\Lambda,X}^\sigma}_\mathfrak{h}.
		\label{eq:VLambdaRen}
	\end{align}
	Some of these terms converge uniformly to a bounded continuous function as $\Lambda \to \infty$. For the first summand, this is a consequence of the inverse triangle inequality and Lemma \ref{lem:B} with $\alpha = 1/2$:
	\begin{align}
		|\norm{\Op(t^{-1}_X) \rho_{\Lambda,X}^\sigma}_{\mathfrak{h}} - \norm{\Op(t^{-1}_X) \delta_{X}^\sigma}_{\mathfrak{h}}| &\leq \norm{\Op(t^{-1}_X) (\rho_{\Lambda,X}^\sigma - \delta_{X}^\sigma)}_{\mathfrak{h}} \notag \\
		&\leq C \norm{\rho_{\Lambda}^\sigma - \delta^\sigma}_{H^{s}}
	\end{align}
	for an $s<-3/2$. Because $\norm{\rho_{\Lambda}^\sigma - \delta^\sigma}_{H^{s}} \to 0$ as $\Lambda \to \infty$ for every $s<-3/2$, $\norm{\Op(t^{-1}_X) \rho_{\Lambda,X}^\sigma}_{\mathfrak{h}}$ converges to $\norm{\Op(t^{-1}_X) \delta_{X}^\sigma}_{\mathfrak{h}}$. The expression $\norm{\Op(t^{-1}_X) \delta_{X}^\sigma}_{\mathfrak{h}}$ depends continuously on $X$ and is bounded by Lemma \ref{lem:B}; thus, it defines an element of $C_\mathrm{b}(\mathbb{R}^3)$.
	
	In similar fashion, we demonstrate the convergence of the second line in \eqref{eq:VLambdaRen}:
	\begin{align}
		&|\scp{\omega^{-1}\Op(\partial_X t^{-1}_X)\rho_{\Lambda,X}^\sigma}{g(X)\Op(\partial_X t^{-1}_X)\rho_{\Lambda,X}^\sigma}_\mathfrak{h} - \notag \\
		& \ \ \ \ \ \ \ \scp{\omega^{-1}\Op(\partial_X t^{-1}_X)\delta_{X}^\sigma}{g(X)\Op(\partial_X t^{-1}_X)\delta_{X}^\sigma}_\mathfrak{h}| \notag \\
		&\ \ \ = |\scp{\omega^{-1}\Op(\partial_X t^{-1}_X)(\rho_{\Lambda,X}^\sigma-\delta_X^\sigma)}{g(X)\Op(\partial_X t^{-1}_X)(\rho_{\Lambda,X}^\sigma+\delta_X^\sigma)}_\mathfrak{h}| \notag \\
		&\ \ \ \leq C\norm{\omega^{-\frac{1}{2}} \Op(\partial_X t^{-1}_X)(\rho_{\Lambda,X}^\sigma-\delta_X^\sigma)}_\mathfrak{h} \norm{\omega^{-\frac{1}{2}} \Op(\partial_X t^{-1}_X)(\rho_{\Lambda,X}^\sigma+\delta_X^\sigma)}_\mathfrak{h} \notag \\
		&\ \ \ \leq C \norm{\rho_{\Lambda}^\sigma-\delta^\sigma}_{H^s} \sup_{\Lambda}\norm{\rho_{\Lambda}^\sigma+\delta^\sigma}_{H^s}
	\end{align}
	for an $s<-3/2$. Analogously, we confirm that the third line in \eqref{eq:VLambdaRen} converges. However, for the remaining two terms in \eqref{eq:VLambdaRen},
	\begin{align}
		&-\scp{\rho_{\Lambda,X}^\sigma}{\Op(t^{-1}_X)\omega^{-1}\rho_{\Lambda,X}}_\mathfrak{h} \notag \\
		&+ \frac{1}{2} \scp{\omega^{-1}\Op(t^{-1}_X) \partial_{x} \rho_{\Lambda,X}^\sigma}{g(X) \Op(t^{-1}_X) \partial_{x} \rho_{\Lambda,X}^\sigma}_\mathfrak{h},
		\label{eq:RemainingTerms}
	\end{align}	
	the above line of reasoning fails because the order of the involved pseudo-differential operators is not sufficiently negative. These are the divergent terms in $V_\Lambda(X)$.
	
	We use the pseudo-differential calculus to extract the vacuum energy $-E_\Lambda(X)$ from \eqref{eq:RemainingTerms}. From the proof of Lemma \ref{lem:B}, the symbol of $t_X^{-1}$ is given by
	\begin{align}
		t^{-1}_X(x,\xi) = (K_0(X,\xi) + 1)^{-1} + S^{-3}.
	\end{align}
	The symbol $\sigma(\omega^{-1})$ of $\omega^{-1}$ in standard quantisation is
	\begin{align}
		\sigma(\omega^{-1})(x,\xi) &= (h_0(x,\xi) + \mu(x)^2)^{-\frac{1}{2}} + S^{-2} \notag \\
		&= (h_0(x,\xi) + 1)^{-\frac{1}{2}} + S^{-2}.
	\end{align}
	Relevant operators for the remaining terms in \eqref{eq:RemainingTerms} are compositions of $\Op(t^{-1}_X)$ and $\omega^{-1}$, which can be written as
	\begin{align}
		\Op(t^{-1}_X)\omega^{-1} &= \Op(c_X) + \mathrm{Op}(\tilde{c}_X) \label{eq:cX} \\
		\Op(t^{-1}_X)^*\omega^{-1}\Op(t^{-1}_X) &= \Op(d_X) + \mathrm{Op}(\tilde{d}_X) \label{eq:dX}
	\end{align}
	with $\tilde{c}_X \in S^{-4}$, $\tilde{d}_X \in S^{-6}$ uniformly in $X$, and
	\begin{align}
		c_X(x,\xi) &= (h_0(x,\xi) + 1)^{-\frac{1}{2}}(K_0(X,\xi) + 1)^{-1} \in S^{-3}, \\
		d_X(x,\xi) &= (h_0(x,\xi) + 1)^{-\frac{1}{2}}(K_0(X,\xi) + 1)^{-2} \in S^{-5}.
	\end{align}
	We can neglect the pseudo-differential operators of order $-4$ and $-6$ in \eqref{eq:cX} and \eqref{eq:dX}, respectively, because
	\begin{align}
		\scp{\rho_{\Lambda,X}^\sigma}{\Op(\tilde{c}_X)\rho_{\Lambda,X}}_\mathfrak{h} + \frac{1}{2} \scp{\partial_{x} \rho_{\Lambda,X}^\sigma}{g(X) \Op(\tilde{d}_X) \partial_{x} \rho_{\Lambda,X}^\sigma}_\mathfrak{h}
	\end{align}	
	converges uniformly to a bounded continuous function as $\Lambda \to \infty$ by the same arguments as above. It remains to analyse the terms with $\Op(c_X)$ and $\Op(d_X)$. We begin with
	\begin{align}
		-\scp{\rho_{\Lambda,X}^\sigma}{\Op(c_X) \rho_{\Lambda,X}}_\mathfrak{h} = -\frac{1}{(2\pi)^3} \int_{\mathbb{R}^3} \int_{\mathbb{R}^3} \e^{\I\scp{x-X}{\xi}} c_X(x,\xi) \rho_{\Lambda,X}^\sigma(x) \widehat{\rho}(\xi/\Lambda) \diff \xi \diff x
		\label{eq:cXexpectation}
	\end{align}
	and expand $c_X(x,\xi)$ in $x=X$:
	\begin{align}
		c_X(x,\xi) = c_X(X,\xi) + (x-X) \cdot r_X(x,\xi)
	\end{align}
	with $r_X \in [S^{-3}]^3$. The remainder $r_X$ in \eqref{eq:cXexpectation} can be written as follows:
	\begin{align}
		&-\frac{1}{(2\pi)^3} \int_{\mathbb{R}^3} \int_{\mathbb{R}^3} D_\xi(\e^{\I\scp{x-X}{\xi}}) \cdot r_X(x,\xi) \rho_{\Lambda,X}^\sigma(x) \widehat{\rho}(\xi/\Lambda) \diff \xi \diff x \notag \\
		&= \frac{1}{(2\pi)^3} \int_{\mathbb{R}^3} \int_{\mathbb{R}^3} \e^{\I\scp{x-X}{\xi}} D_\xi \cdot r_X(x,\xi) \rho_{\Lambda,X}^\sigma(x) \widehat{\rho}(\xi/\Lambda) \diff \xi \diff x \notag \\
		&\ \ \  + \frac{1}{(2\pi)^3} \frac{1}{\Lambda} \int_{\mathbb{R}^3} \int_{\mathbb{R}^3} \e^{\I\scp{x-X}{\xi}} r_X(x,\xi) \cdot \rho_{\Lambda,X}^\sigma(x) (D_\xi \widehat{\rho})(\xi/\Lambda) \diff \xi \diff x.
	\end{align}
	The integral in the second line is equal to $\scp{\rho_{\Lambda,X}^\sigma}{\Op(D_\xi \cdot r)\rho_{\Lambda,X}}_\mathfrak{h}$. This has a bounded continuous limit as $\Lambda \to \infty$ because $D_\xi \cdot r_X \in S^{-4}$. The integral in the third line is asymptotically equivalent to $\log(\Lambda)/\Lambda$ for large $\Lambda$ because the integrand is order $-3$ in $\xi$. This summand vanishes as $\Lambda \to \infty$. Thus,
	\begin{align}
		-\scp{\rho_{\Lambda,X}^\sigma}{\Op(c_X)\rho_{\Lambda,X}}_\mathfrak{h} = -\frac{1}{(2\pi)^3} \int_{\mathbb{R}^3} c_X(X,\xi)  |\widehat{\rho}(\xi/\Lambda)|^2 \diff \xi + F_\Lambda(X).
	\end{align}	
	With similar techniques, we approach the term with $\Op(d_X)$:
	\begin{align}
		&\frac{1}{2}\scp{\partial_{x}\rho_{\Lambda,X}^\sigma}{g(X)\Op(d_X)\partial_{x}\rho_{\Lambda,X}^\sigma}_\mathfrak{h} \notag \\ 
		&\ \ \ = \frac{\I}{2(2\pi)^3} \int_{\mathbb{R}^3} \int_{\mathbb{R}^3} \e^{\I\scp{x-X}{\xi}} d_X(x,\xi) \scp{\partial_{x} \rho_{\Lambda,X}^\sigma(x)}{g(X)\xi} \widehat{\rho^\sigma}(\xi/\Lambda) \diff \xi \diff x \notag \\
		&\ \ \ = \frac{1}{2(2\pi)^3} \int_{\mathbb{R}^3} \int_{\mathbb{R}^3} \e^{\I\scp{x-X}{\xi}} d_X(x,\xi) \scp{\xi}{g(X)\xi} \rho_{\Lambda,X}^\sigma(x) \widehat{\rho^\sigma}(\xi/\Lambda) \diff \xi \diff x \notag \\
		&\ \ \ \ \ \  -\frac{\I}{2(2\pi)^3} \int_{\mathbb{R}^3} \int_{\mathbb{R}^3} \e^{\I\scp{x-X}{\xi}} \scp{\partial_{x} d_X(x,\xi)}{g(X)\xi} \rho_{\Lambda,X}^\sigma(x) \widehat{\rho^\sigma}(\xi/\Lambda) \diff \xi \diff x.
	\end{align}
	Because $\scp{\partial_{x} d_X(x,\xi)}{g(X)\xi} \in S^{-4}$, the integral in the last line has a bounded continuous limit as $\Lambda \to \infty$. In the first integral after the second equality sign, we expand $d_X(x,\xi)$ in $x=X$:
	\begin{align}
		&\frac{1}{2}\scp{\partial_{x}\rho_{\Lambda,X}^\sigma}{g(X)\Op(d_X)\partial_{x}\rho_{\Lambda,X}^\sigma}_\mathfrak{h} \notag \\
		&\ \ \ = \frac{1}{2(2\pi)^3} \int_{\mathbb{R}^3} d_X(X,\xi) \scp{\xi}{g(X)\xi} |\widehat{\rho}(\xi/\Lambda)|^2 \diff \xi + F_\Lambda(X).
	\end{align}
	Note that we also replaced $\rho^\sigma$ with $\rho$. 
	
	We combine the contributions from the terms with $\Op(c_X)$ and $\Op(d_X)$. From the identity
	\begin{align}
		&-c_X(X,\xi) + \frac{1}{2} d_X(X,\xi) \scp{\xi}{g(X)\xi} = -\frac{(h_0(X,\xi)+1)^{-\frac{1}{2}} (K_0(X,\xi)+2)}{2(K_0(X,\xi)+1)^2},
	\end{align}
	it follows that
	\begin{align}
		&-\scp{\rho_{\Lambda,X}^\sigma}{\Op(c_X)\rho_{\Lambda,X}} + \frac{1}{2} \scp{\partial_{x}\rho_{\Lambda,X}}{g(X)\Op(d_X)\partial_{x}\rho_{\Lambda,X}} \notag \\
		&\ \ \  = -\frac{1}{2(2\pi)^3} \int_{\mathbb{R}^3} \frac{(h_0(X,\xi)+1)^{-\frac{1}{2}} (K_0(X,\xi)+2)}{(K_0(X,\xi)+1)^2} \left|\widehat{\rho}\left(\frac{\xi}{\Lambda}\right)\right|^2 \diff \xi + F_\Lambda(X) \notag \\
		&\ \ \  = -\frac{1}{2(2\pi)^3} \int_{\mathbb{R}^3} \frac{(h_0(X,\xi)+1)^{-\frac{1}{2}}}{K_0(X,\xi)+1} \left|\widehat{\rho}\left(\frac{\xi}{\Lambda}\right)\right|^2 \diff \xi + F_\Lambda(X).
	\end{align}
	This expression contains the vacuum energy $-E_\Lambda(X)$ from \eqref{eq:VacuumEnergy}. Thus, $V_\Lambda(X) + E_\Lambda(X) = F_\Lambda(X)$. This concludes the proof because $F_\Lambda$ converges uniformly to a bounded continuous function. 
\end{proof}

\section{Proof of Nelson's Theorem}
\label{sec:ProofOfNelsonsTheorem}

The last step of the proof consists of showing that $H_\Lambda + E_\Lambda(X)$ has a limit in the norm resolvent sense, where $E_\Lambda(X)$ is defined in \eqref{eq:VacuumEnergy}. First, we prove that the transformed Hamiltonian $H_\Lambda' = U_\Lambda H_\Lambda U_\Lambda^* + E_\Lambda(X) = H_0 + R_\Lambda(X) + F_\Lambda(X)$ interpreted in quadratic forms converges to some self-adjoint operator $H_\infty'$. The operator $R_\Lambda(X)$ is defined in \eqref{eq:RLambda}, and $(F_\Lambda)_{\Lambda \geq 0} \subset C_\mathrm{b}^\infty(\mathbb{R}^3)$ is a bounded family of smooth functions that converges uniformly to some $F\in C_\mathrm{b}(\mathbb{R}^3)$ as $\Lambda \to \infty$. Then, we obtain the required renormalised Nelson Hamiltonian $H$ by a retransformation: $H=U_\infty^* H_\infty' U_\infty$.

The following useful proposition is a consequence of the well-known KLMN theorem.

\begin{prop}
	Let $A \geq 0$ be a self-adjoint operator and $(q_\Lambda)_{\Lambda \geq 0}$ a family of quadratic forms on $D(A^{1/2})$ such that
	\begin{align}
		|q_\Lambda(u)| \leq a \norm{A^{\frac{1}{2}}u}^2 + b_\Lambda \norm{u}^2
		\label{eq:FormEstimateKLMN}
	\end{align}
	for some $a<1$ and $b_\Lambda \in \mathbb{R}$. Moreover, assume that
	\begin{align}
		|q_{\Lambda_1}(u) - q_{\Lambda_2}(u)| \leq C_{\Lambda_1 \Lambda_2} \norm{(A+1)^{\frac{1}{2}}u}^2, 
		\label{eq:FormDifferenceEstimateKLMN}
	\end{align}
	where $C_{\Lambda_1 \Lambda_2} \to 0$ as $\Lambda_1, \Lambda_2 \to \infty$. Then \eqref{eq:FormEstimateKLMN} extends to $\Lambda = \infty$ with some $b_\infty < \infty$, and for each $\Lambda \leq \infty$, a unique self-adjoint bounded from below operator $A_\Lambda$ exists with $D(A_\Lambda) \subset D(A^{1/2})$ such that
	\begin{align}
		\scp{u}{A_\Lambda v} = \scp{A^{\frac{1}{2}} u}{A^{\frac{1}{2}} v} + q_\Lambda(u,v)
	\end{align}
	for all $u \in D(A^{\frac{1}{2}})$, $v \in D(A_\Lambda)$. Furthermore, $A_\Lambda$ converges in the norm resolvent sense to $A_\infty$.
	\label{prop:KLMN}
\end{prop}

\begin{proof}
	We take a $\Lambda_0 \geq 0$ that is sufficiently large such that $C_{\Lambda_1\Lambda_2} \leq (1-a)/2$ for all $\Lambda_1,\Lambda_2 \geq \Lambda_0$. Then, for $\Lambda \geq \Lambda_0$, the following inequality holds:
	\begin{align}
		|q_\Lambda(u)| &\leq |q_\Lambda(u)-q_{\Lambda_0}(u)| + |q_{\Lambda_0}(u)| \notag \\
		&\leq (a+C_{\Lambda\Lambda_0}) \norm{A^{\frac{1}{2}}u}^2 + (b_{\Lambda_0} + C_{\Lambda\Lambda_0}) \norm{u}^2 \notag \\
		&\leq \frac{1+a}{2} \norm{A^{\frac{1}{2}}u}^2 + (b_{\Lambda_0} + C_{\Lambda\Lambda_0}) \norm{u}^2.
		\label{eq:KLMN1}
	\end{align}
	Hence, inequality \eqref{eq:FormEstimateKLMN} extends to $\Lambda = \infty$ with $b_\infty = b_{\Lambda_0} + C_{\infty\Lambda_0} < \infty$, where $C_{\infty\Lambda_0} = \limsup_{\Lambda \to \infty} C_{\Lambda \Lambda_0} < \infty$. 
	
	From the KLMN theorem (see \cite[Thm. X.17]{reed2}) it follows that, for each $\Lambda \leq \infty$, a unique self-adjoint operator $A_\Lambda$ exists with $D(A_\Lambda) \subset D(A^{\frac{1}{2}})$ such that
	\begin{align}
		\scp{u}{A_\Lambda v} = \scp{A^{\frac{1}{2}} u}{A^{\frac{1}{2}} v} + q_\Lambda(u,v)
		\label{eq:KLMN2}
	\end{align}
	for all $u\in D(A^\frac{1}{2})$, $v \in D(A_\Lambda)$. Moreover, $A_\Lambda$ is bounded from below by $-b_\Lambda$. 
	
	It remains to prove the convergence of $A_\Lambda$ to $A_\infty$ in the norm resolvent sense. The resolvent of $A_\Lambda$ is denoted by $\mathcal{R}_\Lambda(z) = (A_\Lambda-z)^{-1}$. From assumption \eqref{eq:FormDifferenceEstimateKLMN} and the Cauchy-Schwarz inequality for sesquilinear forms, we obtain, for every $u,v \in D(A^{1/2})$, the following estimate:
	\begin{align}
		&|\scp{u}{(\mathcal{R}_\Lambda(z) - \mathcal{R}_\infty(z))v}| \notag \\ 
		&\ \ \ = |\scp{\mathcal{R}_\Lambda(\overline{z})u}{(A_\infty-z)\mathcal{R}_\infty(z) v} - \scp{(A_\Lambda - \overline{z}) \mathcal{R}_\Lambda(\overline{z})u}{\mathcal{R}_\infty(z)v}| \notag \\
		&\ \ \ = |q_\infty(\mathcal{R}_\Lambda(\overline{z})u, \mathcal{R}_\infty(z)v) - q_\Lambda(\mathcal{R}_\Lambda(\overline{z})u, \mathcal{R}_\infty(z)v)| \notag \\
		&\ \ \ \leq C_{\Lambda\infty} \norm{(A+1)^{\frac{1}{2}}\mathcal{R}_\Lambda(\overline{z})u} \norm{(A+1)^{\frac{1}{2}}\mathcal{R}_\infty(z)v}.
	\end{align}
	The factor $\norm{A^{\frac{1}{2}}\mathcal{R}_\Lambda(\overline{z})u} \leq C_z \norm{u}$ is uniformly bounded by $\Lambda$ because, by \eqref{eq:KLMN1} and \eqref{eq:KLMN2},
	\begin{align}
		A \leq \frac{2}{1-a} (A_\Lambda + b_\infty)
	\end{align}
	for $\Lambda \geq \Lambda_0$. This confirms that, for every $z \in \mathbb{C}\backslash \mathbb{R}$, $\mathcal{R}_\Lambda(z)$ converges to $\mathcal{R}_\infty(z)$ in norm because $C_{\Lambda\infty} \to 0$ by assumption.
\end{proof}

\begin{thm}
	A self-adjoint bounded from below operator $H_\infty'$ exists such that $H_\Lambda' = U_\Lambda H_\Lambda U_\Lambda^* + E_\Lambda(X)$ converges in the norm resolvent sense to $H_\infty'$ as $\Lambda \to \infty$. Moreover, $D(H_\infty') \subset D(H_0^{1/2})$.
	\label{thm:ConvergenceTransformedHamiltonian}
\end{thm}

\begin{proof}	
	We apply Proposition \ref{prop:KLMN} with $A= H_0$ and $q_\Lambda(\Psi) = \scp{\Psi}{(R_\Lambda(X)+ F_\Lambda(X))\Psi}_{\mathfrak{H}}$. First, we verify the assumption \eqref{eq:FormEstimateKLMN}. We estimate the quadratic form $q_\Lambda$ term by term and repeatedly use Lemmata \ref{lem:EstimatesAnnihilationCreation} and \ref{lem:B}. The first two terms in $R_\Lambda(X)$ contain one annihilation or one creation operator and are bounded in quadratic forms by\footnote{Remember that we assume that $C_0\leq g(X) \leq C_1$ is uniformly elliptic.}
	\begin{align}
		&|\scp{a(\partial_X B_{\Lambda,X}^\sigma)\Psi}{g(X) \partial_X\Psi}_{\mathfrak{H}}|
		\leq C \norm{a(\partial_X B_{\Lambda,X}^\sigma)\Psi}_{\mathfrak{H}} \norm{\partial_X\Psi}_{\mathfrak{H}} \notag \\
		&\ \ \ \leq C \norm{\omega^{-\frac{1}{2}} \partial_{X} B_{\Lambda,X}^\sigma}_\mathfrak{h} \norm{\diff\Gamma(\omega)^{\frac{1}{2}}\Psi}_{\mathfrak{H}} \norm{H_0^{\frac{1}{2}}\Psi}_{\mathfrak{H}} \notag \\
		&\ \ \ \leq C \norm{\rho_\Lambda^\sigma}_{H^s} \norm{H_0^{\frac12}\Psi}^2_{\mathfrak{H}}.
	\end{align}
	Then, there are terms in $R_\Lambda(X)$ with two annihilation operators or two creation operators:
	\begin{align}
		&|\scp{\Psi}{a(\partial_X B_{\Lambda,X}^\sigma)\cdot g(X) a(\partial_X B_{\Lambda,X}^\sigma)\Psi}_{\mathfrak{H}}| \notag \\
		&\ \ \ \leq C \norm{(N+1)^{\frac{1}{2}}\Psi}_{\mathfrak{H}} \norm{(N+1)^{-\frac{1}{2}}a(\partial_X B_{\Lambda,X}^\sigma) \cdot g(X)a(\partial_X B_{\Lambda,X}^\sigma)\Psi}_{\mathfrak{H}} \notag \\
		&\ \ \ \leq C \norm{(H_0+1)^{\frac{1}{2}}\Psi}_{\mathfrak{H}} \norm{\omega^{-1/4}\partial_X B_{\Lambda,X}^\sigma}^2_\mathfrak{h} \norm{\diff\Gamma(\omega)^{\frac{1}{2}}\Psi}_{\mathfrak{H}} \notag \\
		&\ \ \ \leq C \norm{\rho_\Lambda^\sigma}_{H^s}^2 (\norm{H_0^{\frac{1}{2}}\Psi}^2_{\mathfrak{H}} + \norm{\Psi}^2_{\mathfrak{H}}).
	\end{align}	
	Next, there is a term with one annihilation and one creation operator in $R_\Lambda(X)$:
	\begin{align}
		&|\scp{a(\partial_X B_{\Lambda,X}^\sigma)\Psi}{g(X)a(\partial_X B_{\Lambda,X}^\sigma)\Psi}_{\mathfrak{H}}|
		\leq C \norm{a(\partial_X B_{\Lambda,X}^\sigma)\Psi}^2_{\mathfrak{H}} \notag \\
		&\ \ \ \leq C \norm{\omega^{-1/2}\partial_X B_{\Lambda,X}^\sigma}^2_\mathfrak{h} \norm{\diff\Gamma(\omega)^{\frac{1}{2}}\Psi}^2_{\mathfrak{H}} \notag \\
		&\ \ \ \leq C \norm{\rho_\Lambda^\sigma}_{H^s}^2 \norm{H_0^{\frac{1}{2}}\Psi}^2_{\mathfrak{H}}.
	\end{align}
	The last term we included in $R_\Lambda(X)$ is the field operator $\Phi(\omega^{-\frac{1}{2}}(\rho_{\Lambda,X} - \rho_{\Lambda,X}^\sigma))$:
	\begin{align}
		&|\scp{\Psi}{\Phi(\omega^{-\frac{1}{2}}(\rho_{\Lambda,X} - \rho_{\Lambda,X}^\sigma))\Psi}_{\mathfrak{H}}|
		\leq C\norm{\omega^{-\frac{1}{2}}(\rho_{\Lambda} - \rho_{\Lambda}^\sigma)}_\mathfrak{h} \norm{(N+1)^{\frac{1}{4}}\Psi}^2_{\mathfrak{H}} \notag \\
		&\ \ \ \leq C\epsilon \norm{\omega^{-\frac{1}{2}}(\rho_{\Lambda} - \rho_{\Lambda}^\sigma)}_\mathfrak{h} \norm{H_0^{\frac{1}{2}}\Psi}^2_{\mathfrak{H}} + C_\epsilon \norm{\Psi}^2_{\mathfrak{H}}
	\end{align}
	for any $\epsilon > 0$.	In the first inequality, we used that $\Phi(f) \leq \sqrt{2} \norm{f}_\mathfrak{h} (N+1)^{1/2}$. 
	
	Moreover, it is clear that
	\begin{align}
		|\scp{\Psi}{F_\Lambda(X)\Psi}_\mathfrak{H}| \leq C \norm{\Psi}^2_\mathfrak{H}
	\end{align}
	because $(F_\Lambda)_{\Lambda \geq 0}$ is a bounded family of bounded functions. 
	
	Collecting all these bounds, we proved that an $s<-3/2$ and a $b_\Lambda \in \mathbb{R}$ exist such that
	\begin{align}
		|q_\Lambda(\Psi)| \leq C (\norm{\rho_{\Lambda}^\sigma}_{H^s} + \norm{\rho_{\Lambda}^\sigma}^2_{H^s} + \epsilon \norm{\omega^{-\frac{1}{2}}(\rho_{\Lambda} - \rho_{\Lambda}^\sigma)}_\mathfrak{h}) \norm{H_0^{\frac{1}{2}}\Psi}^2_\mathfrak{H} + b_\Lambda \norm{\Psi}^2_\mathfrak{H}
	\end{align}
	We must ensure that the prefactor in front of $\norm{H_0^{1/2}\Psi}^2$ is strictly smaller than $1$. This is possible if we first choose $\sigma$ introduced in \eqref{eq:InfraredCutOff} to be sufficiently large and then $\epsilon$ to be sufficiently small such that, for a $\Lambda_0 \geq 0$, the prefactor is strictly smaller than $1$ for $\Lambda_0 \leq \Lambda < \infty$. This verifies the assumption \eqref{eq:FormEstimateKLMN}.
	
	The second assumption of Proposition \ref{prop:KLMN} that we verify is \eqref{eq:FormDifferenceEstimateKLMN}. This is achieved through estimates similar to those above: For an $s<-3/2$,
	\begin{align}
		|\scp{[a(\partial_X B_{\Lambda_1,X}^\sigma) - a(\partial_X B_{\Lambda_2,X}^\sigma)]\Psi}{g(X) \partial_X\Psi}_\mathfrak{H} |
		\leq C \norm{\rho_{\Lambda_1}^\sigma - \rho_{\Lambda_2}^\sigma}_{H^s} \norm{H_0^{\frac{1}{2}}\Psi}^2_\mathfrak{H},
	\end{align}
	and
	\begin{align}
		&|\scp{\Psi}{a(\partial_X B_{\Lambda_1,X}^\sigma)\cdot g(X) a(\partial_X B_{\Lambda_1,X}^\sigma)\Psi}_\mathfrak{H}
		- \scp{\Psi}{a(\partial_X B_{\Lambda_2,X}^\sigma)\cdot g(X) a(\partial_X B_{\Lambda_2,X}^\sigma)\Psi}_\mathfrak{H}| \notag \\
		&\ \ \ =|\scp{\Psi}{a(\partial_X B_{\Lambda_1,X}^\sigma - \partial_X B_{\Lambda_2,X}^\sigma)\cdot g(X) a(\partial_X B_{\Lambda_1,X}^\sigma + \partial_X B_{\Lambda_2,X}^\sigma)\Psi}_\mathfrak{H}| \notag \\ 
		&\ \ \ \leq C \norm{\rho_{\Lambda_1}^\sigma - \rho_{\Lambda_2}^\sigma}_{H^s} \norm{\rho_{\Lambda_1}^\sigma + \rho_{\Lambda_2}^\sigma}_{H^s} \norm{(H_0+1)^{\frac{1}{2}}\Psi}^2_\mathfrak{H} \notag \\
		&\ \ \ \leq C \norm{\rho_{\Lambda_1}^\sigma - \rho_{\Lambda_2}^\sigma}_{H^s} \sup_\Lambda \norm{\rho_{\Lambda}^\sigma}_{H^s} \norm{(H_0+1)^{\frac{1}{2}}\Psi}^2_\mathfrak{H},
	\end{align}	
	as well as
	\begin{align}
		&|\scp{a(\partial_X B_{\Lambda_1,X}^\sigma)\Psi}{g(X)a(\partial_X B_{\Lambda_1,X}^\sigma)\Psi}_\mathfrak{H}
		-\scp{a(\partial_X B_{\Lambda_2,X}^\sigma)\Psi}{g(X)a(\partial_X B_{\Lambda_2,X}^\sigma)\Psi}_\mathfrak{H}| \notag \\ 
		&\ \ \ = |\scp{a(\partial_X B_{\Lambda_1,X}^\sigma - \partial_X B_{\Lambda_2,X}^\sigma)\Psi}{g(X)a(\partial_X B_{\Lambda_1,X}^\sigma + \partial_X B_{\Lambda_2,X}^\sigma)\Psi}_\mathfrak{H}| \notag \\ 
		&\ \ \ \leq C \norm{a(\partial_X B_{\Lambda_1,X}^\sigma-\partial_X B_{\Lambda_2,X}^\sigma)\Psi}_\mathfrak{H} \norm{a(\partial_X B_{\Lambda_1,X}^\sigma + \partial_X B_{\Lambda_2,X}^\sigma)\Psi}_\mathfrak{H} \notag \\
		&\ \ \ \leq C \norm{\rho_{\Lambda_1}^\sigma - \rho_{\Lambda_2}^\sigma}_{H^s} \sup_{\Lambda_1,\Lambda_2} \norm{\rho_{\Lambda_1}^\sigma + \rho_{\Lambda_2}^\sigma}_{H^s} \norm{H_0^{\frac{1}{2}}\Psi}^2_\mathfrak{H}.
	\end{align}
	The difference in the field operators is estimated by
	\begin{align}
		&|\scp{\Psi}{[\Phi(\omega^{-\frac{1}{2}}(\rho_{\Lambda_1,X} - \rho_{\Lambda_1,X}^\sigma))- \Phi(\omega^{-\frac{1}{2}}(\rho_{\Lambda_2,X} - \rho_{\Lambda_2,X}^\sigma))]\Psi}_\mathfrak{H}| \notag \\
		&\ \ \ \leq C \norm{\omega^{-\frac{1}{2}}(\rho_{\Lambda_1} - \rho_{\Lambda_1}^\sigma - \rho_{\Lambda_2} + \rho_{\Lambda_2}^\sigma)}_\mathfrak{h} \norm{(N+1)^{\frac{1}{4}}\Psi}^2_\mathfrak{H}.
	\end{align}
	All factors converge to $0$ as $\Lambda_1, \Lambda_2 \to \infty$. 
	
	Further, because $F_\Lambda$ converges uniformly,
	\begin{align}
		|\scp{\Psi}{(F_{\Lambda_1}(X) -F_{\Lambda_2}(X)) \Psi}_\mathfrak{H}| \leq \norm{F_{\Lambda_1}-F_{\Lambda_2}}_\infty \norm{\Psi}_{\mathfrak{H}}^2 \stackrel{\Lambda_1,\Lambda_2\to \infty}{\longrightarrow} 0.
	\end{align}	
	Hence, all assumptions of Proposition \ref{prop:KLMN} are satisfied. Therefore, for all $\Lambda \leq \infty$, a unique self-adjoint bounded from below operator $H_\Lambda'$ with $D(H_\Lambda') \subset D(H_0^{1/2})$ exists such that
	\begin{align}
		\scp{\Psi_1}{H_\Lambda' \Psi_2}_\mathfrak{H} = \scp{H_0^{\frac{1}{2}}\Psi_1}{H_0^{\frac{1}{2}}\Psi_2}_\mathfrak{H} +  \scp{\Psi_1}{(R_\Lambda(X) + F_\Lambda(X)) \Psi_2}_\mathfrak{H}
	\end{align}
	for all $\Psi_1 \in D(H_0^{\frac{1}{2}})$, $\Psi_2 \in D(H_\Lambda')$. Furthermore, $H_\Lambda'$ converges in the norm resolvent sense to $H_\infty'$. %The equality $H_\Lambda' = H_0 + R_\Lambda(X)$ should be understood in quadratic forms.
\end{proof}

A simple retransformation of the limit operator $H_\infty'$ suffices to finish the proof of Theorem \ref{thm:Gerard}.

\begin{proof}[Proof of Theorem \ref{thm:Gerard}]
	We choose $E_\Lambda(X)$ as in \eqref{eq:VacuumEnergy} and set $H =  U_\infty^* H_\infty' U_\infty$, where $H_\infty'$ is the operator from Theorem \ref{thm:ConvergenceTransformedHamiltonian}. We demonstrate that $H$ is the desired renormalised Nelson Hamiltonian. First, we observe that $(H_0+1)^{-\frac{1}{2}}U_\Lambda$ converges in norm to $(H_0+1)^{-\frac{1}{2}}U_\infty$ by Proposition \ref{prop:PropertiesWeylOperator} item \ref{it:PropWeylOperator3} and Lemma \ref{lem:EstimatesAnnihilationCreation}:
	\begin{align}
		&\sup_{\norm{\Psi}=1} \norm{(U_\Lambda^* - U_\infty^*) (H_0+1)^{-\frac{1}{2}} \Psi}_\mathfrak{H} \leq \norm{\Pi(B_{\Lambda,X}^\sigma - B_{\infty,X}^\sigma) (H_0+1)^{-\frac{1}{2}} \Psi}_\mathfrak{H}\notag \\
		&\ \ \ \leq \sup_{\norm{\Psi}=1} \norm{\omega^{-\frac{1}{2}}(B_{\Lambda,X}^\sigma - B_{\infty,X}^\sigma)}_\mathfrak{h} \norm{\diff\Gamma(\omega)^{\frac{1}{2}} (H_0+1)^{-\frac{1}{2}} \Psi}_\mathfrak{H} \notag \\
		&\ \ \ \leq C \norm{\rho_\Lambda^\sigma - \delta^\sigma}_{H^{-2}} \stackrel{\Lambda\to\infty}{\longrightarrow} 0.
	\end{align}	
	We use this to prove that $H_\Lambda + E_\Lambda(X) = U_\Lambda^*H_\Lambda'U_\Lambda$ converges in the norm resolvent sense to $H$. We write $\mathcal{R}_\Lambda(z) = (H_\Lambda' - z)^{-1}$ for the resolvent of $H_\Lambda'$. From the proof of Proposition \ref{prop:KLMN} it follows that $H_0 \leq C(H_\Lambda' + 1)$ for some constant $C>0$ and all $\Lambda \geq \Lambda_0$. Thus, $(H_0+1)^{\frac{1}{2}} \mathcal{R}_\Lambda(z)$ is bounded uniformly in $\Lambda$ for all $\Lambda \geq \Lambda_0$. A simple computation illustrates that for all $z \in \mathbb{C}\backslash \mathbb{R}$,
	\begin{align}
		&(U_\Lambda^* H_\Lambda' U_\Lambda - z)^{-1} - (U_\infty^*H_\infty' U_\infty - z)^{-1} \notag \\
		&\ \ \ = U_\Lambda^* \mathcal{R}_\Lambda(z) U_\Lambda - U_\infty^* \mathcal{R}_\infty(z) U_\infty \notag \\
		&\ \ \ = U_\Lambda^* \mathcal{R}_\Lambda(z) (U_\Lambda - U_\infty) + U_\Lambda^*(\mathcal{R}_\Lambda(z) - \mathcal{R}_\infty(z))U_\infty + (U_\Lambda^* - U_\infty^*) \mathcal{R}_\infty(z) U_\infty \notag \\
		&\ \ \ = U_\Lambda^* \mathcal{R}_\Lambda(z) (H_0+1)^{\frac{1}{2}}  (H_0+1)^{-\frac{1}{2}} (U_\Lambda - U_\infty) + U_\Lambda^*(\mathcal{R}_\Lambda(z) - \mathcal{R}_\infty(z))U_\infty \notag \\ 
		&\ \ \ \ \ \ + (U_\Lambda^* - U_\infty^*) (H_0+1)^{-\frac{1}{2}}  (H_0+1)^{\frac{1}{2}} \mathcal{R}_\infty(z) U_\infty.
	\end{align}
	All three terms in the last line vanish in the limit of $\Lambda \to \infty$. Hence, $H_\Lambda + E_\Lambda(X)$ converges in the norm resolvent sense to $H$.
\end{proof}

From the proof presented in this chapter, we gain the additional information that the domain $D(H)$ of the renormalised Hamiltonian is a subset of $U_\infty^* D(H_0^{1/2})$. This information was sufficient for Griesemer and Wünsch \cite{griesemer2018} to demonstrate that the domains of the renormalised Hamiltonian and the square root of the free Nelson Hamiltonian intersect trivially: $D(H) \cap D(H_0^{1/2}) = \{0\}$. However, no explicit description of $D(H)$ or $U_\infty^* D(H_0^{1/2})$ is available from this proof. Moreover, the renormalised Hamiltonian $H$ was constructed as a limit of a sequence of operators. This makes it difficult to analyse the renormalised Hamiltonian because we expect that many properties of $H_\Lambda(X) + E_\Lambda(X)$ do not survive the limit $\Lambda \to \infty$. In the next chapter, we obtain an explicit description of the renormalised Hamiltonian through the novel IBC method.
\chapter{IBC Method}
\label{ch:IBC}

The cut-off regularisation presented in the previous chapter is physically equivalent to suppressing high energies or treating particles as smeared objects rather than points. Another well-known re\-gu\-la\-ri\-sation procedure, not discussed here, is to discretise space. Both these regularisations are universally accepted in quantum field theories because it is expected that quantum field theories are merely effective theories and are not valid up to arbitrarily large energy scales or arbitrarily small length scales.

This chapter presents a novel approach to defining Hamiltonians in quantum field theories, which takes particles as points and space as a continuum. This new method, which is from Teufel and Tumulka \cite{teufel2016, teufel2020}, is known as the method of \textbf{interior boundary conditions (IBC)}. The IBC Hamiltonian we construct in this chapter is well defined and self-adjoint and is equivalent to the renormalised Nelson Hamiltonian from the previous chapter. However, this time, we obtain an explicit description of the renormalised Nelson Hamiltonian and its domain, which was not available before. 

Lampart and Schmidt \cite{lampart2021, lampart2019, schmidt2019} already demonstrated that the IBC method applies to Nelson-type Hamiltonians on Euclidean spacetime. The original part of the present work generalises their technique to the Nelson model with variable coefficients.

\section{IBC Hamiltonian}
\label{sec:IBCHamiltonian}

In Section \ref{sec:NelsonModelConstantCoefficients}, we observed that the creation operator $a^*(v_X)$ with form factor $v_X = \omega^{-\frac12}\delta_X/\sqrt{2}$ is not a well-defined operator on the Hilbert space $\mathfrak{H}$ because $v_X \notin L^2(\mathbb{R}^3,\diff x)$. The new approach in this chapter is to interpret $a^*(v_X) \Psi$ as a distribution for suitable $\Psi \in \mathfrak{H}$. In return, we allow $H_0\Psi$ to be a distribution as well and choose $\Psi \in \mathfrak{H}$ in such a way that the singular parts of $a^*(v_X) \Psi$ and $H_0 \Psi$ exactly cancel each other. Then, $L=H_0 + a^*(v_X)$ is a well-defined operator. To realise this idea, we observe that
\begin{align}
	L \Psi = H_0(1+H_0^{-1}a^*(v_X))\Psi \equiv H_0(1-G)\Psi,
\end{align}
where $G = -H_0^{-1}a^*(v_X)$ is a bounded operator on $\mathfrak{H}$ by Proposition \ref{prop:H0GBounded} below. If $\Psi \in \mathfrak{H}$ is chosen such that the interior boundary condition $(1-G)\Psi \in D(H_0)$ is satisfied, then $L\Psi$ is indeed well defined (i.e. the singular parts of $H_0 \Psi$ and $a^*(v_X) \Psi$ exactly cancel each other).

\begin{rem}
	The condition $(1-G)\Psi \in D(H_0)$ is called an interior boundary condition. This is motivated by the following reason. If $\Psi = (\Psi^{(0)}, \Psi^{(1)}, \dots) \in \mathfrak{H}$, the condition $(1-G)\Psi \in D(H_0)$ can (locally) be written as
	\begin{align}
		H_0 \Psi^{(n)}(X,x) &= -\frac{1}{\sqrt{n}} \sum_{j=1}^n v_X(x_j) \Psi^{(n-1)}(X,\hat{x}_j) + F(X,x), 
		\label{eq:IBC}
	\end{align}
	where $F \in \mathfrak{H}$, and $x = (x_1, \dots, x_n) \in \mathbb{R}^{3n}$. Because the pseudo-differential operator $\omega^{-1/2}$ is hypoelliptic (i.e. $\omega^{-1/2}$ leaves the singular support of $\delta_X$ invariant; see Theorem \ref{thm:PseudoLocal}), the form factor $v_X(x_j) = \omega^{-1/2}\delta_X(x_j)$ is singular at the \textbf{collision point} $x_j = X$. These collision points lie in the boundary of the \textbf{configuration space}
	\begin{align}
		\mathcal{Q} = \bigsqcup_{n\in \mathbb{N}} (\mathbb{R}^3 \times \mathbb{R}^{3n}) \backslash \mathcal{C}^{(n)},
	\end{align}
	where $\mathcal{C}^{(n)} \subset \mathbb{R}^3 \times \mathbb{R}^{3n}$ is the set of all $(X,x) \in \mathbb{R}^3 \times \mathbb{R}^{3n}$ such that $x_j = X$ for at least one $1\leq j\leq n$. The boundary $\partial\mathcal{Q}$ of the configuration space is the union of all $\mathcal{C}^{(n)}$. We assume that $(X,x) \in \mathcal{C}^{(n)}$ is such that $x_i = X$ for exactly one $i$. From \eqref{eq:IBC} it follows that
	\begin{align}
		\Psi^{(n-1)}(X,\hat{x}_i) = -\sqrt{n} \lim_{|x_i-X| \to 0} \frac{H_0\Psi^{(n)}(X,x)}{v_X(x_i)}.
	\end{align}
	Thus, the value (or limit) of $H_0\Psi^{(n)}$ at $(X,x) \in \partial \mathcal{Q}$ in the boundary of the configuration space is related to the value of $\Psi^{(n-1)}$ at $(X,\hat{x}_i) \in \mathcal{Q}$ in the interior of the configuration space in a sector with fewer bosons.
\end{rem}

Thus far, we have defined $L = H_0 + a^*(v_X)$ on $(1-G)^{-1}D(H_0)$. The only part of the Nelson Hamiltonian we have not yet discussed is the annihilation operator $a(v_X)$. The annihilation operator must be extended to an operator $A$ on $(1-G)^{-1}D(H_0)$. According to Proposition \ref{prop:H0GBounded} below, the annihilation operator $a(v_X)$ is well defined on $D(H_0)$. Therefore, if $A$ is an extension of $a(v_X)$, it can be written as
\begin{align}
	A\Psi = A(1-G)\Psi + AG\Psi = a(v_X) (1-G)\Psi + T\Psi,
\end{align}
where the action of $A$ on certain elements in $\mathrm{ran}(G)$ is encoded in the operator $T$. In principle, $T$ could be any (symmetric) operator such that the \textbf{IBC Hamiltonian} 
\begin{align}
	H = L + A = (1-G)^* H_0 (1-G) + T \equiv H_0' + T
\end{align}
is self-adjoint on $D(H) = (1-G)^{-1}D(H_0)$. However, whether all choices of $T$ lead to physically interesting Hamiltonians is unclear. To ensure that the IBC Hamiltonian is the Hamiltonian of the Nelson model, we choose $T$ in such a way that the IBC Hamiltonian is equivalent to the renormalised Hamiltonian from Chapter \ref{ch:RemovalCutOff}. In this sense, equivalent means that the Hamiltonians agree up to an exterior particle potential, which is bounded and continuous. For this purpose, we repeat the above construction of the IBC Hamiltonian for the cut-off Hamiltonian $H_\Lambda$. If $v_{\Lambda,X} = \omega^{-1/2}\rho_{\Lambda,X}/\sqrt{2}$ is the regularised form factor with the cut-off function $\rho_{\Lambda,X}$ defined in \eqref{eq:CutOffFunction}, then
\begin{align}
	H_\Lambda &= H_0 + a^*(v_{\Lambda,X}) + a(v_{\Lambda,X}) \notag \\
	&= H_0(1-G_\Lambda) + a(v_{\Lambda,X})(1-G_\Lambda) + a(v_{\Lambda,X})G_\Lambda \notag \\
	&= (1-G_\Lambda)^* H_0 (1-G_\Lambda) + T_\Lambda
\end{align}
with $G_\Lambda = -H_0^{-1}a^*(v_{\Lambda,X})$ and $T_\Lambda = a(v_{\Lambda,X})G_\Lambda = -a(v_{\Lambda,X}) H_0^{-1} a^*(v_{\Lambda,X})$. Due to the cut-off, the operator $T_\Lambda$ is well defined for all $\Lambda < \infty$. If $T$ is taken to be the limit of $T_\Lambda + E_\Lambda(X)$ (in a suitable sense), then the IBC Hamiltonian $H$ is the (norm resolvent) limit of $H_\Lambda + E_\Lambda(X)$. Having fixed this choice of $T$, the main result of this chapter is the following theorem.
\begin{thm}
	The IBC Hamiltonian $H = L + A = H_0'+T$ is self-adjoint on 
	\begin{align}
		D(H) =\{\Psi \in \mathfrak{H} \mid (1-G)\Psi \in D(H_0) \},
	\end{align}
	bounded from below, and its action on $\Psi \in D(H)$ is given by
	\begin{align}
		H \Psi = H_0 \Psi + a^*(v_X)\Psi + A\Psi,
	\end{align}
	where the sum is understood to be taken in $D(H_0^{-1})$. Furthermore, $H$ is equivalent to the renormalised Hamiltonian from Chapter \ref{ch:RemovalCutOff}.
	\label{thm:IBCHamiltonian}
\end{thm}

\begin{proof}
	To begin, $D(H) = (1-G)^{-1}D(H_0)$ is dense in $\mathfrak{H}$ because $D(H_0)$ is dense in $\mathfrak{H}$ and because the inverse operator $(1-G)^{-1}$ is bounded according to Corollary \ref{cor:BoundedInverse1-G}. 
	
	Next, we prove that the free IBC Hamiltonian $H_0'$ is self-adjoint on $D(H_0') = D(H)$. Clearly, $H_0' = (1-G)^* H_0 (1-G)$ is symmetric on $D(H)$; thus, $D(H) \subset D((H_0')^*)$. To demonstrate that $D((H_0')^*) \subset D(H)$, take a $\Psi_1 \in D((H_0')^*)$ and observe that, for all $\Psi_2 \in D(H_0') = (1-G)^{-1}D(H_0)$,
	\begin{align}
		\scp{\Psi_1}{H_0' \Psi_2} = \scp{(1-G)\Psi_1}{H_0(1-G)\Psi_2};
	\end{align}
	hence, $(1-G)\Psi_1 \in D(H_0^*) = D(H_0)$. 
	
	The IBC Hamiltonian $H = H_0' + T$ is self-adjoint on $D(H_0') = D(H)$ because $T$ is an infinitesimal perturbation of $H_0'$ by Proposition \ref{prop:InfinitesimalBoundForT}.	It is also bounded from below because $H_0'\geq 0$ is bounded from below.	That the IBC Hamiltonian is equivalent to the renormalised Hamiltonian is proven in Section \ref{sec:ProofIBCHamiltonian}.
\end{proof}

\begin{rem}
	The operator $H_0' = (1-G)^* H_0 (1-G)$ is similar to the transformed Hamiltonian $U_\infty^* H_0 U_\infty$ from the previous chapter, where $1-G$ takes the role of the Gross transformation $U_\infty$. However, notable differences exist. First, $U_\infty$ is unitary whereas $1-G$ is merely invertible. Furthermore, the action of $1-G$ is easier to evaluate than that of $U_\infty$ because $((1-G)\Psi)^{(n)}$ only depends on $\Psi^{(n)}$ and $\Psi^{(n-1)}$, but $(U_\infty\Psi)^{(n)}$ depends on all boson sectors of $\Psi$.
\end{rem}

In the proof of Theorem \ref{thm:IBCHamiltonian}, we used several results we prove in the sections below. We analyse the operator $G$ in Section \ref{sec:MappingPropertiesOfG}, and define the extension $A$ of the annihilation operator in Section \ref{sec:ExtensionAnnihilationOperator}. In Section \ref{sec:RegularityDomainVectors}, we examine the regularity of the vectors in the domain $D(H)$ of the IBC Hamiltonian. In the last section of this chapter, we prove that the IBC Hamiltonian is equivalent to the renormalised Hamiltonian.

For the rest of this chapter, we fix the following notation.

%\section{Pseudo-differential Operators in the Nelson Model II}
%\label{sec:PseudorsInNelsonModelII}
\begin{ntn}
	In Definition \ref{def:Pseudor}, we introduced the pseudo-differential operator $\Op(a)$ for a symbol $a\in S(M)$ by
	\begin{align}
		\Op(a)\Psi(x) = \frac{1}{(2\pi)^d} \int_{\mathbb{R}^d} \int_{\mathbb{R}^d} \e^{\I\scp{x-y}{\xi}} a(x,\xi) \Psi(y) \diff y \diff \xi,
	\end{align}	
	where $\diff y$ is the Lebesgue measure on $\mathbb{R}^d$ and $\diff\xi$ the dual measure on $(\mathbb{R}^{d})' \equiv \mathbb{R}^d$. To simplify the notation, we absorb the factor $(2\pi)^{-d}$ into the dual measure, that is, we replace $(2\pi)^{-d}\diff \xi$ with $\diff\xi$.
\end{ntn}

\begin{ntn}
	We denote by $f(\Lambda)$ an arbitrary continuous function on $[0,\infty)$ that converges to 0 as $\Lambda \to \infty$.
\end{ntn}

\section{Mapping Properties of $G$}
\label{sec:MappingPropertiesOfG}

In this section, we analyse mapping properties of the operator $G = -H_0^{-1}a^*(v_X)$ and its regularised version $G_\Lambda = -H_0^{-1}a^*(v_{\Lambda, X})$. We prove that $G$ is continuous as an operator from $D(N^{1/2})$ to $D(H_0^p)$ for every $p<1/2$ (Proposition \ref{prop:H0GRelativelyBounded}) and that $G$ is continuous from $D(N^{1/2})$ to $D(\diff\Gamma(\omega)^{1/2})$ (Proposition \ref{prop:GmapsDomainToItself}). For $p < 1/4$, we moreover show that $H_0^pG$ is bounded (Proposition \ref{prop:H0GBounded}). From the boundedness of $H_0^pG$ for some $p>0$ it follows that $1-G$ has a bounded inverse (Corollary \ref{cor:BoundedInverse1-G}).

%There is a relatively simple argument for the boundedness of $H_0^pG$ if $p<0$. Namely, from Lemma \ref{lem:EstimatesAnnihilationCreation} it follows that
%\begin{align}
%	\norm{(H_0^pG)^* \Psi}_{\mathfrak{H}} &= \norm{a(v_X)H_0^{p-1}\Psi}_{\mathfrak{H}} \notag \\
%	&\leq \sup_X \norm{\omega^{-1-\epsilon}v_X}_\mathfrak{h} \norm{\diff\Gamma(\omega)^{1+\epsilon}H_0^{p-1} \Psi}_{\mathfrak{H}}
%\end{align}
%for every $\epsilon > 0$. The mapping properties of pseudo-differential operators between Sobolev spaces (cf. Corollary \ref{cor:PseudorsSobolevMapping}) imply that the factor
%\begin{align}
%	\sup_X \norm{\omega^{-1-\epsilon}v_X}_\mathfrak{h} = \sup_X \norm{\omega^{-3/2-\epsilon}\delta_X}_\mathfrak{h} \leq C \norm{\delta}_{H^{-3/2-\epsilon}}
%\end{align}
%is finite. Moreover, $\diff\Gamma(\omega)^{1+\epsilon} H_0^{p-1} \leq H_0^{p+\epsilon}$ is bounded if $\epsilon \leq -p$. Hence, $(H_0^pG)^*$ is bounded and so is its adjoint. To prove that $H_0^pG$ is bounded for $0\leq p\leq 1/4$ and relatively bounded with respect to $N^{1/2}$ for $p<1/2$ requires a more careful analysis.

\begin{prop}
	For every $p<1/2$, the operator $G$ is continuous as an operator from $D(N^{1/2})$ to $D(H_0^p)$. Moreover, $G_\Lambda$ converges in norm to $G$ as an operator from $D(N^{1/2})$ to $D(H_0^p)$.
	\label{prop:H0GRelativelyBounded}
\end{prop}

\begin{proof}
	We prove that $\norm{H_0^p(G_\Lambda - G)\Psi}_\mathfrak{H} \leq f(\Lambda)\norm{(N+1)^{1/2}\Psi}_\mathfrak{H}$ for every $\Psi \in D(N^{1/2})$. The relative boundedness of $H_0^pG$ with respect to $N^{1/2}$ follows from setting $\Lambda = 0$ and the convergence from $f(\Lambda) \to 0$ as $\Lambda\to \infty$.
	
	We split the operator $H_0^p(G_\Lambda-G)$ into two summands according to the decomposition $v_{\Lambda,X} = u_{\Lambda,X} + \tilde{u}_{\Lambda,X}$ of the form factor from Section \ref{sec:PseudorsInNelsonModelI}:
	\begin{align}
		H_0^p(G_\Lambda-G) = H_0^{p-1}a^*(u_X - u_{\Lambda,X}) + H_0^{p-1}a^*(\tilde{u}_X - \tilde{u}_{\Lambda,X}).
		\label{eq:GLambda-G}
	\end{align}		
	The adjoint operator of the second summand satisfies, by Lemma \ref{lem:EstimatesAnnihilationCreation}, the bound
	\begin{align}
		&\norm{a(\tilde{u}_{\Lambda,X} - \tilde{u}_X)H_0^{p-1}\Psi}_{\mathfrak{H}} \notag \\ 
		&\ \ \leq \sup_X \left( \norm{\omega^{-\frac{1}{2}} \Op(D_\xi \cdot r) (\rho_{\Lambda,X} - \delta_X)}_\mathfrak{h} + \frac{\norm{\omega^{-\frac{1}{2}} \Op(r) \tau_{\Lambda,X}}_\mathfrak{h}}{\Lambda} \right) \norm{\diff\Gamma(\omega)^{\frac{1}{2}} H_0^{p-1}\Psi}_{\mathfrak{H}} \notag \\ 
		&\ \ \leq C \left(\norm{\rho_{\Lambda}-\delta}_{H^{-2}} + \frac{\norm{\tau_{\Lambda}}_{H^{-1}}}{\Lambda}\right) \norm{\diff\Gamma(\omega)^{p-\frac{1}{2}}\Psi}_\mathfrak{H}.
		\label{eq:H0GFormFactorRemainder}
	\end{align}
	The factor $\norm{\rho_{\Lambda}-\delta}_{H^{-2}}$ converges to $0$ due to \eqref{eq:ConvergenceCutOffToDelta} and the factor $\norm{\tau_{\Lambda}}_{H^{-1}}/\Lambda$ is, for any $\epsilon >0$ and large $\Lambda$, bounded by $\Lambda^{-1/2 + \epsilon}$ times a constant; hence, $\norm{\tau_{\Lambda}}_{H^{-1}}/\Lambda$ converges to $0$ as $\Lambda \to \infty$. Also, $\norm{\rho_{\Lambda}-\delta}_{H^{-2}} + \norm{\tau_{\Lambda}}_{H^{-1}}/\Lambda$ is continuous in $\Lambda >0$ and attains a finite value as $\Lambda \to 0$.
	
	Next, we consider the summand $H_0^{p-1}a^*(u_X - u_{\Lambda,X})$. According to the functional calculus (see Theorem \ref{thm:FunctionalCalculus}), $(H_0\vert_{\mathfrak{H}^{(n)}})^{p-1}$ is a pseudo-differential operator. Its symbol is an element of $S(M^{p-1},g)$ with the order function $M$ from \eqref{eq:OrderFunctionH0} and the metric $g$ from \eqref{eq:MetricH0}. Let $\sigma(H_0^{p-1}) \in S(M^{p-1}, g)$ be the symbol of $(H_0\vert_{\mathfrak{H}^{(n)}})^{p-1}$ in the right quantisation in the particle variables and in the standard quantisation in the boson variables. Then, for any $\Psi^{(n-1)} \equiv \Psi \in \mathfrak{H}^{(n-1)} \cap D(N^{1/2})$,
	\begin{align}
		& H_0^{p-1}a^*(u_X-u_{\Lambda,X})\Psi(X,x) \label{eq:H0GPseudor} \\
		&= \frac{1}{\sqrt{n}} \sum_{j=1}^n \int \e^{\I\scp{X-Y}{\Xi}} \e^{\I\scp{x-y}{\xi}} \frac{(u_Y-u_{\Lambda,Y})(y_j) \Psi(Y,\hat{y}_j)}{\sigma(H_0^{p-1})(Y,\Xi,x,\xi)^{-1}} \diff Y\diff\Xi \diff y \diff\xi \notag \\
		&= \frac{1}{\sqrt{n}} \sum_{j=1}^n \int \e^{\I\scp{X-Y}{\Xi}} \e^{\I\scp{\hat{x}_j-\hat{y}_j}{\hat{\xi}_j}} \frac{\e^{\I \scp{x_j-Y}{\xi_j}} \sigma(\omega^{-\frac{1}{2}})(Y,\xi_j) \zeta_\Lambda(\xi_j) \Psi(Y,\hat{y}_j)}{\sqrt{2} \sigma(H_0^{p-1})(Y,\Xi,x,\xi)^{-1} } \diff Y\diff\Xi \diff \hat{y}_j \diff\xi, \notag
	\end{align}		
	where $\zeta_\Lambda(\xi_j) = 1-\widehat{\rho}(\xi_j/\Lambda)$.	We expand the symbol $\sigma(H_0^{p-1})(Y,\Xi,x,\xi)$ at the collision point $x_j = Y$:
	\begin{align}
		\sigma(H_0^{p-1})(Y,\Xi,x,\xi) = \sigma(H_0^{p-1})(Y,\Xi,\hat{x}_j,Y,\xi) - (x_j-Y) \cdot R(Y,\Xi,x,\xi),
	\end{align}
	where $R \in [S(M^{p-1},g)]^3$ and $(\hat{x}_j,Y)$ is the vector $x\in\mathbb{R}^{3n}$ with the $j$-th entry replaced by $Y$. We insert this expansion in \eqref{eq:H0GPseudor}. The summand containing $\sigma(H_0^{p-1})(Y,\Xi,\hat{x}_j,Y,\xi)$ is called the leading order, and the summand containing $R$ the remainder.
	
	%\begin{enumerate}[wide, labelwidth=!, labelindent=0pt]
		\item \textbf{Leading order}: The leading order term describes a sum of operator-valued pseudo-differential operators acting on $\Psi \in \mathfrak{H}^{(n-1)} \cap D(N^{1/2})$:
		\begin{align}
			&\frac{1}{\sqrt{n}} \sum_{j=1}^n \int \e^{\I\scp{X-Y}{\Xi}} \e^{\I\scp{\hat{x}_j-\hat{y}_j}{\hat{\xi}_j}} \frac{\e^{\I \scp{x_j-Y}{\xi_j}} \sigma(\omega^{-\frac{1}{2}})(Y,\xi_j) \zeta_\Lambda(\xi_j) \Psi(Y,\hat{y}_j)}{\sqrt{2} \sigma(H_0^{p-1})(Y,\Xi,\hat{x}_j,Y,\xi)^{-1} } \diff Y\diff\Xi \diff \hat{y}_j \diff\xi \notag \\
			&= \frac{1}{\sqrt{n}} \sum_{j=1}^n \int \e^{\I\scp{X-Y}{\Xi}} \e^{\I\scp{\hat{x}_j-\hat{y}_j}{\hat{\xi}_j}} \frac{\e^{\I \scp{x_j-X}{\xi_j}} \sigma(\omega^{-\frac{1}{2}})(Y,\xi_j) \zeta_\Lambda(\xi_j) \Psi(Y,\hat{y}_j)}{\sqrt{2}\sigma(H_0^{p-1})(Y,\Xi-\xi_j,\hat{x}_j,Y,\xi)^{-1}} \diff Y\diff\Xi \diff \hat{y}_j \diff\xi \notag \\
			&= \frac{1}{\sqrt{n}}\sum_{j=1}^n \int \e^{\I\scp{X-Y}{\Xi}} \e^{\I\scp{\hat{x}_j-\hat{y}_j}{\hat{\xi}_j}} [m_{j,\Lambda}(Y,\Xi,\hat{x}_j,\hat{\xi}_j) \Psi(Y,\hat{y}_j)](x_j-X) \diff Y \diff\Xi \diff \hat{y}_j \diff\hat{\xi}_j \notag \\
			&\equiv \frac{1}{\sqrt{n}}\sum_{j=1}^n \Op(m_{j,\Lambda})\Psi(X,\hat{x}_j,x_j-X).
			\label{eq:H0GLeadingOrder}
		\end{align}
		The operator-valued symbol $m_{j,\Lambda}(Y,\Xi,\hat{x}_j,\hat{\xi}_j):\mathbb{C}\to L^2(\mathbb{R}^3)$ is defined by
		\begin{align}
			m_{j,\Lambda}(Y,\Xi,\hat{x}_j,\hat{\xi}_j) = \frac{1}{\sqrt{2}} \int_{\mathbb{R}^3} \frac{\e^{\I\scp{\cdot}{\xi_j}}\sigma(\omega^{-\frac{1}{2}})(Y,\xi_j)\zeta_\Lambda(\xi_j)}{\sigma(H_0^{p-1})(Y,\Xi-\xi_j,\hat{x}_j,Y,\xi)^{-1}} \diff \xi_j.
			\label{eq:mjSymbol}
		\end{align}	
		We apply the Calderon--Vaillancourt Theorem to prove that $\Op(m_{j,\Lambda})$ is bounded. The operator norm of $m_{j,\Lambda}(Y,\Xi,\hat{x}_j,\hat{\xi}_j)$ is bounded because, for any $\epsilon > 0$ sufficiently small,
		\begin{align}
			\norm{m_{j,\Lambda}(Y,\Xi,\hat{x}_j,\hat{\xi}_j)}_{L^2}^2 
			&= \frac{1}{2}\int_{\mathbb{R}^3} \bigg|\frac{\sigma(\omega^{-\frac{1}{2}})(Y,\xi) \zeta_\Lambda(\xi_j)}{\sigma(H_0^{p-1})(Y,\Xi-\xi_j,\hat{x}_j,Y,\xi)^{-1}}\bigg|^2 \diff \xi_j \notag \\
			&\leq C \int_{\mathbb{R}^3} \frac{\jap{\xi_j}^{-1} |\zeta_\Lambda(\xi_j)|^2}{(\jap{\Xi-\xi_j}^2 + \Omega(\xi))^{2-2p}} \diff \xi_j \notag \\
			&\leq C \int_{\mathbb{R}^3} \frac{|\xi_j|^{-1} |\zeta_\Lambda(\xi_j)|}{(|\Xi-\xi_j|^2 + |\xi_j| + \Omega(\hat{\xi_j}))^{2-2p}} \diff \xi_j \notag \\
			&\leq f(\Lambda) \Omega(\hat{\xi}_j)^{-1+2p+\epsilon}.
			\label{eq:mjOperatorNorm}
		\end{align}
		In the second inequality, we used $ |\zeta_\Lambda(\xi_j)| \leq 2$ and, in the third, Corollary \ref{cor:IntegralEstimate}. We take $0<\epsilon\leq 1-2p$ so that $\Omega(\hat{\xi}_j)^{-1+2p+\epsilon} \leq 1$.
		
		The operator norms of the derivatives of $m_{j,\Lambda}$ are bounded by similar estimates; hence, $m_{j,\Lambda} \in S(f(\Lambda),\mathfrak{B}(\mathbb{C}, L^2))$. It follows that $\Op(m_{j,\Lambda})$ is a bounded operator, and its operator norm is bounded by $f(\Lambda)$. The bound can be chosen independent of $j$ because there are only finitely many $j$. Thus, we obtain the following estimate:
		\begin{align}
			&\int \bigg|\frac{1}{\sqrt{n}}\sum_{j=1}^n \Op(m_{j,\Lambda})\Psi(X,\hat{x}_j,x_j-X)\bigg|^2 \diff X \diff x \leq \sum_{j=1}^n \int |\Op(m_{j,\Lambda})\Psi(X,x)|^2 \diff X \diff x \notag \\		
			&\ \ \ \leq f(\Lambda) \sum_{j=1}^{n} \int |\Psi(X,\hat{x}_j)|^2 \diff X \diff \hat{x}_j = f(\Lambda) \norm{(N+1)^{\frac{1}{2}}\Psi}^2_{\mathfrak{H}^{(n-1)}}.
		\end{align} 
	
		\item \textbf{Remainder}: We demonstrate that the integral containing $R \in [S(M^{p-1},g)]^3$,
		\begin{align}
			-\frac{1}{\sqrt{2n}} \sum_{j=1}^n \int \e^{\I\scp{X-Y}{\Xi}} \e^{\I\scp{\hat{x}_j-\hat{y}_j}{\hat{\xi}_j}} (x_j-Y)\e^{\I \scp{x_j-Y}{\xi_j}} \sigma(\omega^{-\frac{1}{2}})(Y,\xi_j) \zeta_\Lambda(\xi_j) \notag \\ 
			\cdot R(Y,\Xi,x,\xi) \Psi(Y,\hat{y}_j)  \diff Y\diff\Xi \diff \hat{y}_j \diff\xi,
		\end{align}	
		defines a bounded operator acting on $\Psi \in \mathfrak{H}^{(n-1)} \cap D(N^{1/2})$. We replace the factor $x_j-Y$ with a derivative $D_{\xi_j}$ and integrate by parts:
		\begin{align}
			&\frac{1}{\sqrt{2n}} \sum_{j=1}^n \int \e^{\I\scp{X-Y}{\Xi}} \e^{\I\scp{\hat{x}_j-\hat{y}_j}{\hat{\xi}_j}} \e^{\I \scp{x_j-Y}{\xi_j}} \big\{[D_{\xi_j}\sigma(\omega^{-\frac{1}{2}})(Y,\xi_j) \zeta_\Lambda(\xi_j) \notag \\  
			&\ \ \ + \frac{1}{\Lambda}\sigma(\omega^{-\frac{1}{2}})(Y,\xi_j) (D_{\xi_j}\widehat{\rho})(\xi_j/\Lambda)] \cdot R(Y,\Xi,x,\xi) \notag \\ 
			&\ \ \ + \sigma(\omega^{-\frac{1}{2}})(Y,\xi_j) \zeta_\Lambda(\xi_j) (D_{\xi_j} \cdot R)(Y,\Xi,x,\xi) \big\} \Psi(Y,\hat{y}_j) \diff Y\diff\Xi \diff \hat{y}_j \diff\xi.
			\label{eq:GRemainder}
		\end{align}	
		The first summand in \eqref{eq:GRemainder} is equal to $\Op(R)a^*(w_X - w_{\Lambda,X})\Psi(X,x)$ with
		\begin{align}
			\widehat{w_{\Lambda,X}}(\xi_j) = \frac{1}{\sqrt{2}} \e^{-\I \scp{X}{\xi_j}} D_{\xi_j} \sigma(\omega^{-\frac{1}{2}})(X,\xi_j) \widehat{\rho}(\xi_j/\Lambda).
		\end{align} 
%because
%		\begin{align}
%			\norm{a(w_X)\Op(R)^*\Psi}_\mathfrak{H} \leq \sup_X \norm{\omega^{-\frac{1}{2}} w_X}_\mathfrak{h} \norm{\diff\Gamma(\omega)^{\frac{1}{2}}\Op(R)^*\Psi}_\mathfrak{H} \leq C \norm{\diff\Gamma(\omega)^{p-\frac{1}{2}}\Psi}_{\mathfrak{H}}.
%		\end{align}	
		The third summand has the same structure as the first summand. To see this, we write
		\begin{align}
			\sigma(\omega^{-\frac{1}{2}})(Y,\xi_j)D_{\xi_j}\cdot R(Y,\Xi,x,\xi) = \sigma(\omega^{-\frac{1}{2}})(Y,\xi)\jap{\xi_j}^{-1} \jap{\xi_j} D_{\xi_j}\cdot R(Y,\Xi,x,\xi).
		\end{align}
		Then, $\sigma(\omega^{-\frac{1}{2}})(Y,\xi) \jap{\xi_j}^{-1}$ as well as $D_{\xi_j} \sigma(\omega^{-\frac{1}{2}})(X,\xi_j)$ define symbols in $S^{-3/2}$, and $\jap{\xi_j}D_{\xi_j}\cdot R(Y,\Xi,x,\xi)$ as well as each vector component of $R(Y,\Xi,x,\xi)$ define symbols in $S(M^{p-1},g)$. The second summand is equal to $\Op(R)a^*(\omega^{-1/2}\tau_{\Lambda,X})/\Lambda\Psi$. 
		
		All three summands are bounded in norm by $f(\Lambda) \norm{\Psi}_{\mathfrak{H}^{(n-1)}}$ according to a similar argument as in \eqref{eq:H0GFormFactorRemainder}.
	%\end{enumerate}	
	\item Combining the estimates of the leading order and the remainder, we conclude that, for every $\Psi^{(n-1)} \in \mathfrak{H}^{(n-1)} \cap D(N^{1/2})$,
	\begin{align}
		\norm{H_0^p(G_\Lambda-G)\Psi^{(n-1)}}_{\mathfrak{H}^{(n)}} \leq f(\Lambda) \norm{(N+1)^{\frac{1}{2}} \Psi^{(n-1)}}_{\mathfrak{H}^{(n-1)}}.
	\end{align}
	Summing over $n$ yields the desired inequality.
\end{proof}

We work with the decomposition of $H_0^pG$ developed in the above proof repeatedly in this chapter. For reference, we formulate the decomposition in the following lemma.

\begin{lem}
	For every $p\leq 1/2$ and $n\in \mathbb{N}$, a bounded operator $B$ exists such that $\diff\Gamma(\omega)^{1/2-p}B$ is bounded and such that
	\begin{align}
		H_0^pG\Psi^{(n-1)}(X,x) = - \frac{1}{\sqrt{n}} \sum_{j=1}^n \Op(m_j)\Psi^{(n-1)}(X,\hat{x}_j,x_j-X) + B\Psi^{(n-1)}(X,x),
	\end{align}
	where $m_j \equiv m_{j,0} \in S(1,\mathfrak{B}(\mathbb{C}, L^2))$ is defined in \eqref{eq:mjSymbol}.
	\label{lem:DecompositionG}
\end{lem}

\begin{rem}
	Both $B$ and $m_j$ depend on the parameter $p$ and the boson number $n$.
\end{rem}

We consider two modifications of Proposition \ref{prop:H0GRelativelyBounded}. First, we extend the statement to $p=1/2$ by replacing $H_0$ with $\diff\Gamma(\omega)$. Then, we improve the bound on $H_0^pG$ for $p<1/4$.

\begin{prop}
	The operator $G$ is continuous as an operator from $D(N^{1/2})$ to $D(\diff\Gamma(\omega)^{1/2})$. Moreover, for every $\epsilon >0$, $G_\Lambda$ converges in norm to $G$ as an operator from $D(N^{1/2}\diff\Gamma(\omega)^{\epsilon})$ to $D(\diff\Gamma(\omega)^{1/2})$.
	\label{prop:GmapsDomainToItself}
\end{prop}

\begin{proof}
	We prove that $\norm{\diff\Gamma(\omega)^{1/2}(G_\Lambda - G)\Psi}_\mathfrak{H} \leq f(\Lambda)\norm{(N+1)^{1/2}\diff\Gamma(\omega)^{\epsilon(\Lambda)}\Psi}_\mathfrak{H}$ for every $\Psi \in D(N^{1/2}\diff\Gamma(\omega)^{\epsilon(\Lambda)})$, where $\epsilon(\Lambda)$ is equal to 0 if $\Lambda = 0$ and equal to $\epsilon$ if $\Lambda > 0$. 
	
	To prove this inequality, we copy the proof of Proposition \ref{prop:H0GRelativelyBounded} but replace $H_0^{p-1}$ with $\diff \Gamma(\omega)^{1/2}H_0^{-1}$. The critical difference arises in the bound of the leading order. The leading order term of $\diff \Gamma(\omega)^{1/2}H_0^{-1}a^*(v_X-v_{\Lambda,X})$ is a sum of operator-valued pseudo-differential operators with symbols $k_{j,\Lambda}$, where
	\begin{align}
		k_{j,\Lambda}(Y,\Xi,\hat{x}_j,\hat{\xi}_j) = \frac{1}{\sqrt{2}}\int_{\mathbb{R}^3} \frac{\e^{\I \scp{\cdot}{\xi_j}}\sigma(\omega^{-\frac{1}{2}})(Y,\xi_j) \zeta_\Lambda(\xi_j)}{\sigma(\diff \Gamma(\omega)^{1/2}H_0^{-1})(Y,\Xi-\xi_j,\hat{x}_j,Y,\xi)^{-1}} \diff \xi_j.
	\end{align}	
	The symbol $\sigma(\diff \Gamma(\omega)^{1/2}H_0^{-1})$ is an element of $S(M)$ with the order function $M(\Xi,\xi) = \Omega(\xi)^{1/2}/(\jap{\Xi}^2 + \Omega(\xi))$. Thus, the operator norm of $k_{j,\Lambda}(Y,\Xi,\hat{x}_j,\hat{\xi}_j)$ is bounded by
	\begin{align}
		\norm{k_{j,\Lambda}(Y,\Xi,\hat{x}_j,\hat{\xi}_j)}^2_{L^2} &= \frac{1}{2}\int_{\mathbb{R}^3} \bigg| \frac{\sigma(\omega^{-\frac{1}{2}})(Y,\xi_j) \zeta_\Lambda(\xi_j)}{\sigma(\diff \Gamma(\omega)^{1/2}H_0^{-1})(Y,\Xi-\xi_j,\hat{x}_j,Y,\xi)^{-1}} \bigg|^2 \diff \xi_j \notag \\
		&\leq C \int_{\mathbb{R}^3} \frac{(\Omega(\hat{\xi}_j) + \jap{\xi_j}) \jap{\xi_j}^{-1} |\zeta_\Lambda(\xi_j)|}{(\jap{\Xi-\xi_j}^2 + \Omega(\xi))^2} \diff \xi_j \notag \\ 
		&\leq C \int_{\mathbb{R}^3} \frac{(\Omega(\hat{\xi}_j)|\xi_j|^{-1} +1)|\zeta_\Lambda(\xi_j)|}{(|\Xi - \xi_j|^2 + |\xi_j| + \Omega(\hat{\xi}_j))^2} \diff \xi_j \leq f(\Lambda)\Omega(\hat{\xi}_j)^{\epsilon(\Lambda)}.
		\label{eq:kjNorm}
	\end{align}
	The operator norm of the derivatives of $k_{j,\Lambda}$ are bounded by similar estimates; hence, $k_{j,\Lambda} \in S(f(\Lambda)\Omega(\hat{\xi}_j)^{\epsilon(\Lambda)}, \mathfrak{B}(\mathbb{C},L^2))$. The desired inequality is then a consequence of the Calderon--Vaillancourt Theorem.
\end{proof}

%In Proposition \ref{prop:H0GRelativelyBounded}, we proved that $H_0^pG$ is relatively bounded with respect to $N^{1/2}$. We improve the bound for $p< 1/4$.

\begin{prop}
	For every $p < 1/4$, the operator $H_0^pG$ is bounded. Moreover, $G_\Lambda$ converges to $G$ as an operator from $\mathfrak{H}$ to $D(H_0^p)$.
	\label{prop:H0GBounded}
\end{prop}

\begin{proof}
	We prove that $\norm{H_0^p(G_\Lambda - G)\Psi}_\mathfrak{H} \leq f(\Lambda)\norm{\Psi}_\mathfrak{H}$ for every $\Psi \in \mathfrak{H}$. 
	
	In the proof of Proposition \ref{prop:H0GRelativelyBounded}, we already demonstrated that the remainder terms of $H_0^p(G_\Lambda - G)$ satisfy this bound. It remains to prove the bound for the leading order.
	
	It follows from the Cauchy--Schwarz inequality that, for every $\Psi \in \mathfrak{H}^{(n-1)}$,
	\begin{align}
		\norm{\frac{1}{\sqrt{n}}\sum_{j=1}^n \Op(m_{j,\Lambda})\Psi}^2_{\mathfrak{H}^{(n)}} = \frac{1}{n} \int \Bigg|\sum_{j=1}^n \frac{\mathcal{F}[\Op(m_{j,\Lambda})\Psi](\Xi,\xi)\jap{\xi_j}^{\frac{1}{2}}}{\jap{\xi_j}^{\frac{1}{2}}} \Bigg|^2 \diff \Xi \diff \xi \notag \\
		\leq \frac{1}{n} \sum_{i,j=1}^{n} \int \frac{|\mathcal{F}[\Op(m_{j,\Lambda})\Psi](\Xi,\xi)|^2 \jap{\xi_i}}{\jap{\xi_j}} \diff \Xi \diff \xi.
	\end{align}
	We distinguish two cases. First, if $i = j$, we use that $\Op(m_{j,\Lambda})$ is bounded by $f(\Lambda)$:
	\begin{align*}
		\frac{1}{n}\sum_{j=1}^n \int |\mathcal{F}[\Op(m_{j,\Lambda})\Psi](\Xi,\xi)|^2 \diff \Xi \diff \xi = \frac{1}{n}\sum_{j=1}^n \norm{\Op(m_{j,\Lambda})\Psi}^2_{\mathfrak{H}^{(n)}}
		\leq f(\Lambda) \norm{\Psi}^2_{\mathfrak{H}^{(n-1)}}.
	\end{align*}	
	For the terms with $i \neq j$, we recall equation \eqref{eq:H0GLeadingOrder} and observe that
	\begin{align}
		&\frac{1}{n} \sum_{j=1}^n \sum_{\substack{i=1 \\ i\neq j}}^n  \int \frac{|\mathcal{F}[\Op(m_{j,\Lambda})\Psi](\Xi,\xi)|^2 \jap{\xi_i}}{\jap{\xi_j}} \diff \Xi \diff \xi \notag \\
		&\ \ \ = \frac{1}{n} \sum_{j=1}^n \int \Omega(\hat{\xi}_j) \frac{|\mathcal{F}[\Op(m_{j,\Lambda})\Psi](\Xi,\xi)|^2}{\jap{\xi_j}} \diff \Xi \diff \xi \notag \\
		%&\ \ \ = \frac{1}{n} \sum_{j=1}^n \int \Omega(\hat{\xi}_j) |\mathcal{F}[\Op(\mathcal{F}^{-1}(\mathcal{F}(m_j)/\jap{\cdot}^{\frac{1}{2}}))\Psi](\Xi,\xi)|^2 \diff \Xi \diff \xi \notag \\
		&\ \ \ = \frac{1}{n} \sum_{j=1}^n \norm{\Op(\Omega(\hat{\xi}_j)^{\frac{1}{2}})\Op(\mathcal{F}^{-1}(\mathcal{F}(m_{j,\Lambda})/\jap{\cdot}^{\frac{1}{2}})) \Psi}^2_{\mathfrak{H}^{(n)}}.
		\label{eq:ineqj}
	\end{align}
	The operator norm of $\mathcal{F}^{-1}(\mathcal{F}(m_{j,\Lambda})/\jap{\cdot}^{\frac{1}{2}})$ is, due to Corollary \ref{cor:IntegralEstimate} with $0 < \epsilon \leq 1-2p$, bounded by
	\begin{align}
		\norm{\mathcal{F}^{-1}(\mathcal{F}(m_{j,\Lambda})/\jap{\cdot}^{\frac{1}{2}})}_{L^2}^2 &\leq C \int_{\mathbb{R}^3} \frac{\jap{\xi_j}^{-2} |\zeta_\Lambda(\xi_j)|^2}{(\jap{\Xi-\xi_j}^2 + \Omega(\xi))^{2-2p}} \diff \xi_j \notag \\
		&\leq f(\Lambda) \Omega(\hat{\xi}_j)^{-\frac{1}{2}}
	\end{align}
	The operator norms of the derivatives of $\mathcal{F}^{-1}(\mathcal{F}(m_j)/\jap{\cdot}^{\frac{1}{2}})$ are bounded by similar estimates; hence, $\mathcal{F}^{-1}(\mathcal{F}(m_{j,\Lambda})/\jap{\cdot}^{\frac{1}{2}}) \in S(f(\Lambda)\Omega(\hat{\xi}_j)^{-1/2},\mathfrak{B}(\mathbb{C},L^2))$. It follows that $\Op(\Omega(\hat{\xi}_j)^{\frac{1}{2}})\Op(\mathcal{F}^{-1}(\mathcal{F}(m_j)/\jap{\cdot}^{\frac{1}{2}}))$ is an operator-valued pseudo-differential operator with underlying symbol in $S(f(\Lambda),\mathfrak{B}(\mathbb{C},L^2))$. Thus, \eqref{eq:ineqj} is bounded by $f(\Lambda)\norm{\Psi}_{\mathfrak{H}^{(n-1)}}$.
\end{proof}

%From the boundedness of $H_0^pG$ for some $p>0$, we deduce the following corollary which guarantees that the domain $D(H) = (1-G)^{-1}D(H_0)$ of the IBC Hamiltonian is dense in $\mathfrak{H}$.

\begin{cor}
	For every $\Lambda \leq \infty$, $1-G_\Lambda$ is invertible and its inverse is a bounded operator uniformly in $\Lambda$. Moreover, a constant $C>0$ independent of $\Lambda$ exists such that
		\begin{align}
			\norm{N\Psi}_\mathfrak{H} \leq C(\norm{N(1-G_\Lambda)\Psi}_\mathfrak{H} + \norm{\Psi}_\mathfrak{H})
			\label{eq:1-GBoundedInverse}
		\end{align}
		for all $\Psi \in (1-G_\Lambda)^{-1}D(N)$.
	\label{cor:BoundedInverse1-G}
\end{cor}

\begin{proof}	
	According to Proposition \ref{prop:H0GBounded}, there is a $p>0$ such that, for all $\Psi \in \mathfrak{H}$,
	\begin{align}
		\norm{G^*_{\Lambda}\Psi}_\mathfrak{H} \leq C\norm{H_0^{-p} \Psi}_\mathfrak{H} \leq C\norm{N^{-p}\Psi}_\mathfrak{H};
	\end{align}
	thus, $\norm{G_{\Lambda}}_{\mathfrak{H}^{(n-1)}\to \mathfrak{H}^{(n)}} = \norm{G^*_{\Lambda}}_{\mathfrak{H}^{(n)} \to \mathfrak{H}^{(n-1)}} \leq C n^{-p}$. It follows that
	\begin{align}
		\norm{G^k_{\Lambda} \Psi}^2_{\mathfrak{H}} &= \sum_{n\in\mathbb{N}} \norm{(G^k_{\Lambda}\Psi)^{(n)}}^2_{\mathfrak{H}^{(n)}} 
		\leq \sum_{n\geq k} \prod_{l=1}^{k} \norm{G_{\Lambda}}_{\mathfrak{H}^{(n-l)} \to \mathfrak{H}^{(n-l+1)}}^2 \norm{\Psi^{(n-k)}}^2_{\mathfrak{H}^{(n-k)}} \notag \\
		&\leq C (k!)^{-2p} \norm{\Psi}^2_\mathfrak{H}.
	\end{align}
	Hence, the von Neumann series $\sum_{k=0}^m G^k_{\Lambda}$ converges in the strong operator topology as $m \to \infty$ because $(k!)^{-p}$ is summable. Thus, the inverse of $1-G_\Lambda$ exists and it is bounded.
	
	To prove the second assertion, we set $G_{\Lambda;\eta} = -(H_0+\eta^2)^{-1}a^*(v_{\Lambda,X})$. For every $\eta \in \mathbb{R}$, $G_{\Lambda;\eta}$ is a bounded operator from $\mathfrak{H}$ to $\mathfrak{H}$ and from $D(N)$ to $D(N)$ because $G_{\Lambda;\eta}$ maps $\mathfrak{H}^{(n-1)}$ into $\mathfrak{H}^{(n)}$. Moreover, the operator norm $\norm{G_{\Lambda;\eta}}_{D(N) \to D(N)} \leq c_{\eta}$ is bounded by a constant $c_{\eta}$ (independent of $\Lambda$) that decreases to $0$ as $\eta \to \infty$. Thus, $c_\eta < 1$ for $\eta$ sufficiently large. For these $\eta$, the operator $1-G_{\Lambda;\eta}$ is invertible as an operator from $D(N)$ to $D(N)$ and its inverse is a bounded operator with norm at most $(1-c_\eta)^{-1}$. It follows that
	\begin{align}
		\norm{N\Psi}_\mathfrak{H} &= \norm{N(1-G_{\Lambda;\eta})^{-1}(1-G_{\Lambda;\eta})\Psi}_\mathfrak{H} \\	
		&\leq (1-c_\eta)^{-1} (\norm{N(1-G_{\Lambda;\eta})\Psi}_\mathfrak{H} + \norm{(1-G_{\Lambda;\eta})\Psi}_\mathfrak{H}) \notag \\
		&\leq (1-c_\eta)^{-1} (\norm{N(1-G_\Lambda)\Psi}_\mathfrak{H} + \eta^2\norm{N(H_0+\eta^2)^{-1} G_\Lambda \Psi}_\mathfrak{H} \notag + \norm{(1-G_{\Lambda;\eta})\Psi}_\mathfrak{H}).
	\end{align}
	In the second inequality, we used $G_\Lambda-G_{\Lambda;\eta} = \eta^2 (H_0 + \eta^2)^{-1}H_0^{-1}a^*(v_{\Lambda,X})$ by the resolvent formula. The proof of \eqref{eq:1-GBoundedInverse} is complete because $N(H_0+\eta^2)^{-1} G_\Lambda$ and $1-G_{\Lambda;\eta}$ are bounded operators uniformly in $\Lambda$.
\end{proof}

\section{Extension of the Annihilation Operator}
\label{sec:ExtensionAnnihilationOperator}

We construct a physically reasonable extension $A = a(v_X) (1-G) + T$ of the annihilation operator $a(v_X)$ to the domain $D(H) = (1-G)^{-1}D(H_0)$ of the IBC Hamiltonian $H$. In the first part of this section, we construct the operator $T$. In the second part, we prove that $T$ is infinitesimally bounded with respect to the free IBC Hamiltonian $H_0'$.

\subsection{Construction of $T$}

We require that the operator $T$ is chosen such that the IBC Hamiltonian is equivalent to the renormalised Nelson Hamiltonian from Chapter \ref{ch:RemovalCutOff}. To ensure this, we take $T$ to be the limit of $T_{\Lambda} + E_\Lambda(X)$ where $T_{\Lambda} = a(v_{\Lambda,X})G_\Lambda$. In which sense this limit should be understood becomes clear later.

Let $\sigma(H_0^{-1}) \in S(M^{-1},g)$ with the order function $M$ from \eqref{eq:OrderFunctionH0} and the metric $g$ from \eqref{eq:MetricH0} be the symbol of $H_0^{-1}\vert_{\mathfrak{H}^{(n)}}$ in the right quantisation in the $n$-th boson variable and in the standard quantisation in all other variables. Then, for suitable $\Psi^{(n-1)} \equiv \Psi \in \mathfrak{H}^{(n-1)}$, $T_\Lambda\Psi$ in leading order is equal to
\begin{align}
	&-a(u_{\Lambda,X}) H_0^{-1} a^*(u_{\Lambda,X}) \Psi(X,\hat{x}_n) 	 \notag \\
	&\ \ \ =-\sum_{j=1}^{n} \int \e^{\I \scp{X-Y}{\Xi}} \e^{\I\scp{x-y}{\xi}} \frac{\overline{u_{\Lambda,X}(x_{n})} u_{\Lambda,Y}(y_j) \Psi(Y,\hat{y}_j) }{\sigma(H_0^{-1})(X,\Xi,\hat{x}_{n},y_{n},\xi)^{-1}} \diff Y \diff \Xi \diff y \diff \xi \diff x_{n}. \label{eq:ConstructionOfT} 
\end{align}
Beyond leading order are terms which include $\tilde{u}_{\Lambda,X}$ of the form factor $v_{\Lambda,X} = u_{\Lambda, X} + \tilde{u}_{\Lambda,X}$ (see Section \ref{sec:PseudorsInNelsonModelI}). We collect these and all other remainders arising in the following at the end of this subsection. 

In \eqref{eq:ConstructionOfT}, we perform the $x_n$-integral because it is a simple Fourier transformation:
\begin{align}
	-\sum_{j=1}^{n} \int \e^{\I \scp{X-Y}{\Xi}} \e^{\I\scp{\hat{x}_{n}-\hat{y}_{n}}{\hat{\xi}_{n}}} \frac{\e^{\I\scp{X-y_n}{\xi_n}} \overline{\sigma}(\omega^{-\frac{1}{2}})(X,\xi_{n}) \widehat{\rho}(\xi_n/\Lambda)}{\sqrt{2}\sigma(H_0^{-1})(X,\Xi,\hat{x}_{n},y_{n},\xi)^{-1}} \notag \\
	u_{\Lambda,Y}(y_j) \Psi(Y,\hat{y}_j) \diff Y \diff \Xi \diff y \diff \xi,
	\label{eq:ConstructionOfT2} 
\end{align}
where $\overline{\sigma}(\omega^{-1/2})$ denotes the complex conjugate of $\sigma(\omega^{-1/2})$. We make similar mani\-pu\-lations as in the proof of Proposition \ref{prop:H0GRelativelyBounded} to simplify the above integral. First, we expand $\sigma(H_0^{-1})(X,\Xi,\hat{x}_{n},y_{n},\xi)$ at the collision point $y_n=X$:
\begin{align}
	\sigma(H_0^{-1})(X,\Xi,\hat{x}_{n},y_{n},\xi) = \sigma(H_0^{-1})(X,\Xi,\hat{x}_{n},X,\xi) + (X-y_n) R_1(X,\Xi,\hat{x}_{n},y_{n},\xi)
	\label{eq:ExpansionH0ForT}
\end{align}
with $R_1\in [S(M^{-1},g)]^3$. In \eqref{eq:ConstructionOfT2}, we separate in the leading order term the diagonal part, which is the term with $j=n$, from the off-diagonal part, which are all other summands with $1 \leq j \leq n-1$.

\subsubsection*{Off-Diagonal Part}

We perform the $y_j$-integral in the off-diagonal part:
\begin{align}
	-\sum_{j=1}^{n-1} \int \e^{\I \scp{X-Y}{\Xi}} \e^{\I\scp{\hat{x}_{j,n}-\hat{y}_{j,n}}{\hat{\xi}_{j,n}}} \frac{\overline{\sigma}(\omega^{-\frac{1}{2}})(X,\xi_{n}) \sigma(\omega^{-\frac{1}{2}})(Y,\xi_j)\widehat{\rho}(\xi_n/\Lambda)\widehat{\rho}(\xi_j/\Lambda)}{2\sigma(H_0^{-1})(X,\Xi,\hat{x}_{n},X,\xi)^{-1}}  \notag \\
	\e^{\I\scp{X-y_n}{\xi_n}} \e^{\I\scp{x_j-Y}{\xi_j}} \Psi(Y,\hat{y}_j) \diff Y \diff \Xi \diff \hat{y}_j \diff \xi.
	\label{eq:OffDiagonal}
\end{align}
Next, we expand the symbol $\sigma(H_0^{-1})(X,\Xi,\hat{x}_{n},X,\xi)$ at the collision point $x_j = Y$:
\begin{align}
	&\sigma(H_0^{-1})(X,\Xi,\hat{x}_{n},X,\xi) = \sigma(H_0^{-1})(X,\Xi,\hat{x}_{j,n},Y,X,\xi) \notag \\
	&\ \ \ + (x_j-Y) \cdot R_{\mathrm{od},1}(X,\Xi,\hat{x}_{j,n},Y,X,\xi) \notag \\ 
	&\ \ \ + (x_j-Y)^T R_{\mathrm{od},2}(X,\Xi,\hat{x}_{n},X,\xi) (x_j-Y)
	\label{eq:OffDiagonalRest}
\end{align}
with $R_{\mathrm{od},1} \in [S(M^{-1},g)]^3$ and $R_{\mathrm{od},2} \in [S(M^{-1},g)]^{3\times 3}$. It becomes clear later why we expand this symbol up to second order. We define the leading order term in \eqref{eq:OffDiagonal} as the off-diagonal part $T_{\mathrm{od},\Lambda}$ of the operator $T_\Lambda$:
\begin{align}
	T_{\mathrm{od},\Lambda} \Psi(X,\hat{x}_{n})
	&= \sum_{j=1}^{n-1} \int \e^{\I \scp{X-Y}{\Xi}} \e^{\I\scp{\hat{x}_{j,n}-\hat{y}_{j,n}}{\hat{\xi}_{j,n}}} \notag \\  
	&\ [o_{j,\Lambda}(X,Y,\Xi,\hat{x}_{j,n},\hat{\xi}_{j,n})\Psi(Y,\hat{y}_{j,n},\cdot+Y)](x_j-X) \diff Y \diff \Xi \diff \hat{y}_{j,n} \diff \hat{\xi}_{j,n}\notag \\
	&= \sum_{j=1}^{n-1} [\Op(o_{j,\Lambda}) \Psi_X(X,\hat{x}_{j,n},\cdot)](x_j-X),
	\label{eq:TodLambda}
\end{align}
where $\Psi_X$ denotes the translated function $\Psi_X(X,\hat{x}_j) = \Psi(X,\hat{x}_{j,n},x_n+X)$. The operator-valued symbol $o_{j,\Lambda}(X,Y,\Xi,\hat{x}_{j,n},\hat{\xi}_{j,n})$ acts on suitable $\psi \in L^2(\mathbb{R}^3)$ via
\begin{align}
	&o_{j,\Lambda}(X,Y,\Xi,\hat{x}_{j,n},\hat{\xi}_{j,n}) \psi(x_j) \label{eq:OffDiagonalSymbol} \notag \\
	&\ \ \ = -\int_{\mathbb{R}^6} \e^{\I\scp{x_j}{\xi_j}} \frac{\overline{\sigma}(\omega^{-\frac{1}{2}})(X,\xi_{n}) \sigma(\omega^{-\frac{1}{2}})(Y,\xi_j) \widehat{\rho}(\xi_n/\Lambda)\widehat{\rho}(\xi_j/\Lambda)}{2\sigma(H_0^{-1})(X,\Xi-\xi_j-\xi_n,\hat{x}_{j,n},Y,X,\xi)^{-1}} \widehat{\psi}(\xi_n) \diff \xi_j \diff \xi_n. 
\end{align}

\subsubsection*{Diagonal Part}

We perform the $y_n$-integral in the diagonal part:
\begin{align}
	-\int \e^{\I \scp{X-Y}{\Xi}} \e^{\I\scp{\hat{x}_{n}-\hat{y}_{n}}{\hat{\xi}_{n}}} \e^{\I\scp{X-Y}{\xi_n}} \frac{\overline{\sigma}(\omega^{-\frac{1}{2}})(X,\xi_{n}) \sigma(\omega^{-\frac{1}{2}})(Y,\xi_n) |\widehat{\rho}(\xi_n/\Lambda)|^2 }{2\sigma(H_0^{-1})(X,\Xi,\hat{x}_{n},X,\xi)^{-1}} \notag \\
	\Psi(Y,\hat{y}_n) \diff Y \diff \Xi \diff \hat{y}_n \diff \xi.
\end{align}
Next, we expand the symbol $\sigma(\omega^{-1/2})(Y,\xi_n)$ at the point $Y=X$:
\begin{align}
	\sigma(\omega^{-\frac{1}{2}})(Y,\xi_n) = \sigma(\omega^{-\frac{1}{2}})(X,\xi_n) + (X-Y)\cdot R_\mathrm{d}(Y,\xi_n) 
	\label{eq:DiagonalRest}
\end{align}
with $R_\mathrm{d} \in [S^{-1/2}]^3$. A closer inspection of the leading order term
\begin{align}
	-\int \e^{\I \scp{X-Y}{\Xi}} \e^{\I\scp{\hat{x}_{n}-\hat{y}_{n}}{\hat{\xi}_{n}}} \frac{|\sigma(\omega^{-\frac{1}{2}})(X,\xi_n)|^2 |\widehat{\rho}(\xi_n/\Lambda)|^2 \Psi(Y,\hat{y}_n) }{2\sigma(H_0^{-1})(X,\Xi-\xi_n,\hat{x}_{n},X,\xi)^{-1}}
	\diff Y \diff \Xi \diff \hat{y}_n \diff \xi
	\label{eq:DiagonalPart}
\end{align}
reveals that the $\xi_n$-integral is divergent in the limit $\Lambda \to \infty$.\footnote{This becomes clear in the Nelson model with constant coefficients in which the symbols are independent of the position variables $(x,X)$. In the Nelson model with constant coefficients, \eqref{eq:DiagonalPart} for $\Lambda = \infty$ (remember that $\widehat{\rho}(0) = 1$) simplifies to the inverse Fourier transform of
\begin{align*}
	-\int_{\mathbb{R}^3} \frac{|\sigma(\omega^{-\frac{1}{2}})(\xi_{n})|^2 \widehat{\Psi}(\Xi,\hat{\xi}_n) }{2\sigma(H_0^{-1})(\Xi-\xi_n,\xi)^{-1}} \diff \xi_n.
\end{align*}
The function $\widehat{\Psi}(\Xi,\hat{\xi}_n)$ does not depend on $\xi_n$, and the integrand is asymptotically equivalent to $\jap{\xi_n}^{-3}$ for large $\xi_n$; thus, the $\xi_n$-integral is divergent.} To obtain a well-defined operator in the limit $\Lambda \to \infty$, we must renormalise the $\xi_n$-integral. This is achieved by subtracting the vacuum energy
\begin{align}
	-E_\Lambda(X) = -\frac{1}{2} \int_{\mathbb{R}^3} \frac{|\sigma(\omega^{-\frac{1}{2}})(X,\xi_n)|^2}{\sigma(K)(X,\xi_n) + \sigma(\omega)(X,\xi_n)}|\widehat{\rho}(\xi_n/\Lambda)|^2 \diff \xi_n,
\end{align} 
where $\sigma(K)$ and $\sigma(\omega)$ are the symbols of $K$ and $\omega$ in the standard quantisation. We define the renormalised integral as
\begin{align}
	&d_\Lambda(X,\Xi,\hat{x}_n,\hat{\xi}_n) = -\frac{1}{2} \int_{\mathbb{R}^3} |\sigma(\omega^{-\frac{1}{2}})(X,\xi_n)|^2 |\widehat{\rho}(\xi_n/\Lambda)|^2 \notag \\ 
	&\ \ \ \left[\frac{1}{\sigma(H_0^{-1})(X,\Xi-\xi_n,\hat{x}_{n},X,\xi)^{-1}} - \frac{1}{\sigma(K)(X,\xi_n) + \sigma(\omega)(X,\xi_n)}\right] \diff \xi_n,
	\label{eq:diagonalSymbol}
\end{align}
and the diagonal part $T_{\mathrm{d},\Lambda}$ of $T_\Lambda$ by
\begin{align}
	T_{\mathrm{d},\Lambda} \Psi(X,\hat{x}_{n}) &= \int \e^{\I \scp{X-Y}{\Xi}} \e^{\I\scp{\hat{x}_n-\hat{y}_n}{\hat{\xi}_n}} d_\Lambda(X,\Xi,\hat{x}_n,\hat{\xi}_n) \Psi(Y,\hat{y}_n) \diff Y \diff \Xi \diff \hat{y}_n \diff \hat{\xi}_n \notag \\
	&= \Op(d_\Lambda)\Psi(X,\hat{x}_n).
	\label{eq:Td}
\end{align}

\subsubsection*{Remainders}

We collect the several remainders we encountered above. First, we neglected in \eqref{eq:ConstructionOfT} the terms
\begin{align}
	T_{\tilde{R},\Lambda} = a(\tilde{u}_{\Lambda,X})H_0^{-1}a^*(u_{\Lambda,X}+\tilde{u}_{\Lambda,X}) + a(u_{\Lambda,X})H_0^{-1}a^*(\tilde{u}_{\Lambda,X})
	\label{eq:TRestTilde}
\end{align}
when considering only the leading order part of the form factor $v_{\Lambda,X} = u_{\Lambda,X}+\tilde{u}_{\Lambda,X}$. 

Next, we obtained from the expansion in \eqref{eq:ExpansionH0ForT} the remainder
\begin{align}
	-\frac{1}{\sqrt{2}}\sum_{j=1}^{n} \int \e^{\I \scp{X-Y}{\Xi}} \e^{\I\scp{\hat{x}_{n}-\hat{y}_{n}}{\hat{\xi}_{n}}} (X-y_n) \e^{\I\scp{X-y_n}{\xi_n}}\overline{\sigma}(\omega^{-\frac{1}{2}})(X,\xi_{n}) \widehat{\rho}(\xi_n/\Lambda) \notag \\ 
	\cdot u_{\Lambda,Y}(y_j) R_1(X,\Xi,\hat{x}_{n},y_{n},\xi) \Psi(Y,\hat{y}_j) \diff Y \diff \Xi \diff y \diff \xi.
	\label{eq:TRest1}
\end{align}
Integration by parts leads to the following expression:
\begin{align}
	&\frac{1}{\sqrt{2}}\sum_{j=1}^{n} \int \e^{\I \scp{X-Y}{\Xi}} \e^{\I\scp{\hat{x}_{n}-\hat{y}_{n}}{\hat{\xi}_{n}}} \e^{\I\scp{X-y_n}{\xi_n}} \big\{[D_{\xi_n}\overline{\sigma}(\omega^{-\frac{1}{2}})(X,\xi_{n}) \widehat{\rho}(\xi_n/\Lambda) \notag \\ 
	&+ \frac{1}{\Lambda} \overline{\sigma}(\omega^{-\frac{1}{2}})(X,\xi_{n}) (D_{\xi_n} \widehat{\rho})(\xi_n/\Lambda)] \cdot R_1(X,\Xi,\hat{x}_{n},y_{n},\xi) \notag \\
	&+ \overline{\sigma}(\omega^{-\frac{1}{2}})(X,\xi_{n}) \widehat{\rho}(\xi_n/\Lambda) D_{\xi_n} \cdot R_1(X,\Xi,\hat{x}_{n},y_{n},\xi) \big\}  u_{\Lambda,Y}(y_j) \Psi(Y,\hat{y}_j) \diff Y \diff \Xi \diff y \diff \xi.
	\label{eq:Rest1}	 
\end{align}
The first summand in \eqref{eq:Rest1} is equal to $a(w_{\Lambda,X})\Op(R_1)a^*(u_{\Lambda,X})\Psi(X,x)$ with
\begin{align}
	\widehat{w_{\Lambda,X}}(\xi_n) = \frac{1}{\sqrt{2}} \e^{-\I\scp{X}{\xi_n}}D_{\xi_n}\overline{\sigma}(\omega^{-\frac{1}{2}})(X,\xi_n)\widehat{\rho}(\xi_n/\Lambda).
\end{align}
The third summand in \eqref{eq:Rest1} has the same structure if we write
\begin{align}
	\overline{\sigma}(\omega^{-\frac{1}{2}})(X,\xi_{n})\jap{\xi_n}^{-1} \jap{\xi_n} D_{\xi_n} \cdot R_1(X,\Xi,\hat{x}_{n},y_{n},\xi)
\end{align}
because $\overline{\sigma}(\omega^{-\frac{1}{2}})(X,\xi_{n})\jap{\xi_n}^{-1}$ as well as $D_{\xi_n}\overline{\sigma}(\omega^{-\frac{1}{2}})(X,\xi_{n})$ define symbols in $S^{-3/2}$ and $\jap{\xi_n} D_{\xi_n} \cdot R_1(X,\Xi,\hat{x}_{n},y_{n},\xi)$ as well as all vector components of $R_1(X,\Xi,\hat{x}_{n},y_{n},\xi)$ define symbols in $S(M^{-1},g)$. 

The second summand in \eqref{eq:Rest1} is equal to $a(\omega^{-1/2}\tau_{\Lambda,X}) \Op(R_1) a^*(u_{\Lambda,X})/\Lambda\Psi$.

Thereafter, we got two remainders in the off-diagonal part from the expansion in \eqref{eq:OffDiagonalRest}. The first remainder containing $R_{\mathrm{od},1}$ is
\begin{align}
	\frac{1}{2}\sum_{j=1}^{n-1} \sum_{k=1}^3 \int \e^{\I \scp{X-Y}{\Xi}} \e^{\I\scp{\hat{x}_{j,n}-\hat{y}_{j,n}}{\hat{\xi}_{j,n}}} D_{\xi_j}^{(k)}\left[ \frac{\sigma(\omega^{-\frac{1}{2}})(Y,\xi_j) \widehat{\rho}(\xi_j/\Lambda) }{ R_{\mathrm{od},1}^k(X,\Xi,\hat{x}_{j,n},Y,X,\xi)^{-1} } \right] \notag \\ 
	\overline{\sigma}(\omega^{-\frac{1}{2}})(X,\xi_{n}) \widehat{\rho}(\xi_n/\Lambda) \e^{\I\scp{X-y_n}{\xi_n}} \e^{\I\scp{x_j-Y}{\xi_j}} \Psi(Y,\hat{y}_j) \diff Y \diff \Xi \diff \hat{y}_j \diff \xi,
	\label{eq:Rod1st}
\end{align}
where $D_{\xi_j}^{(k)}$ denotes $-\I$ times the derivative of the $k$-th component of $\xi_j$ and $R_{\mathrm{od},1}^k$ the $k$-th component of the vector $R_{\mathrm{od},1}$. Similarly, we write the second remainder of the off-diagonal part as
\begin{align}
	\frac{1}{2}\sum_{j=1}^{n-1} \sum_{k,l=1}^3 \int \e^{\I \scp{X-Y}{\Xi}} \e^{\I\scp{\hat{x}_{j,n}-\hat{y}_{j,n}}{\hat{\xi}_{j,n}}} D_{\xi_j}^{(k)} D_{\xi_j}^{(l)} \left[\frac{\sigma(\omega^{-\frac{1}{2}})(Y,\xi_j) \widehat{\rho}(\xi_j/\Lambda)} {R_{\mathrm{od},2}^{kl}(X,\Xi,\hat{x}_{n},X,\xi)^{-1}}\right] \notag \\
	\overline{\sigma}(\omega^{-\frac{1}{2}})(X,\xi_{n}) \widehat{\rho}(\xi_n/\Lambda) \e^{\I\scp{X-y_n}{\xi_n}} \e^{\I\scp{x_j-Y}{\xi_j}} \Psi(Y,\hat{y}_j) \diff Y \diff \Xi \diff \hat{y}_j \diff \xi.
	\label{eq:Rod2nd}
\end{align}
Lastly, we obtained a remainder in the diagonal part from the expansion in \eqref{eq:DiagonalRest}:
\begin{align}
	\frac{1}{2}\int \e^{\I \scp{X-Y}{\Xi}} \e^{\I\scp{\hat{x}_{n}-\hat{y}_{n}}{\hat{\xi}_{n}}} \e^{\I\scp{X-Y}{\xi_n}} D_{\Xi} \sigma(H_0^{-1})(X,\Xi,\hat{x}_{n},X,\xi) \notag \\
	\cdot \overline{\sigma}(\omega^{-\frac{1}{2}})(X,\xi_{n}) R_\mathrm{d}(Y,\xi_n)|\widehat{\rho}(\xi_n/\Lambda)|^2 \Psi(Y,\hat{y}_n) \diff Y \diff \Xi \diff \hat{y}_n \diff \xi.
	\label{eq:Rd}
\end{align}

\begin{rem}
	In the Nelson model with constant coefficients all remainders are absent.
\end{rem}

\subsubsection*{Definition of $T$}

To summarise, we decomposed the operator $T_\Lambda + E_\Lambda(X)$ into $T_{\mathrm{od},\Lambda} + T_{\mathrm{d},\Lambda} + T_{R,\Lambda}$, where $T_{\mathrm{od},\Lambda}$ is defined in \eqref{eq:TodLambda}, $T_{\mathrm{d},\Lambda}$ in \eqref{eq:Td}, and $T_{R,\Lambda}$ contains the various remainders we collected in the previous paragraph, that is, $T_{R,\Lambda} = T_{\tilde{R},\Lambda} + T_{R_1,\Lambda} + T_{R_{\mathrm{od}},\Lambda} + T_{R_{\mathrm{d}},\Lambda}$ with $T_{\tilde{R},\Lambda}$ defined in \eqref{eq:TRestTilde}, $T_{R_1,\Lambda}$ in \eqref{eq:TRest1}, $T_{R_{\mathrm{od}},\Lambda}$ in \eqref{eq:Rod1st} and \eqref{eq:Rod2nd}, and $T_{R_{\mathrm{d}},\Lambda}$ in \eqref{eq:Rd}.

We define the operator $T$ by formally setting $\Lambda = \infty$, what amounts for replacing $\widehat{\rho}(\xi/\Lambda)$ with $\widehat{\rho}(0) = 1$. We demonstrate in the next subsection that the operator $T$ is well defined.

\subsection{Mapping Properties of $T$}
\label{ssec:MappingPropertiesOfT}

We establish different operator bounds on the off-diagonal, diagonal, and remainder part of $T$, and prove convergence of $T_\Lambda + E_\Lambda(X)$ to $T$. At the end of this subsection, we use these operator bounds to prove that $T$ is infinitesimally bounded with respect to the free IBC Hamiltonian $H_0'$.

\begin{prop}
	For every $\epsilon > 0$, $T_\mathrm{od}$ is continuous as an operator from $N^{\frac{1}{4}+\epsilon}\diff\Gamma(\omega)^{\frac{1}{2}}$ to $\mathfrak{H}$. Moreover, $T_{\mathrm{od},\Lambda}$ converges in norm to $\mathrm{T}_\mathrm{od}$ as an operator from $D(N^{\frac{1}{4}+\epsilon}\diff\Gamma(\omega)^{\frac{1}{2}})$ to $\mathfrak{H}$.
	\label{prop:Tod}
\end{prop}

\begin{proof} 	
	We prove that $\norm{(T_{\mathrm{od},\Lambda}-T_{\mathrm{od}})\Psi}_\mathfrak{H} \leq f(\Lambda) \norm{N^{\frac{1}{4}+\epsilon}\diff\Gamma(\omega)^{\frac{1}{2}}\Psi}_\mathfrak{H}$. 
	
	The action of $T_{\mathrm{od},\Lambda}-T_{\mathrm{od}}$ on the subspace $\mathfrak{H}^{(n-1)}$ is described by a finite sum of operator-valued pseudo-differential operators with symbols $o_{j,\Lambda} - o_j$. From the Cauchy--Schwarz inequality it follows that, for every $\psi \in D(\omega^{1/2})$,
	\begin{align}
		&\norm{(o_{j,\Lambda}-o_j)(X,Y,\Xi,\hat{x}_{j,n},\hat{\xi}_{j,n}) \psi}_{L^2}^2 \\
		&= \int_{\mathbb{R}^3} \bigg| \int_{\mathbb{R}^3}	\frac{\overline{\sigma}(\omega^{-\frac{1}{2}})(X,\xi_{n})\sigma(\omega^{-\frac{1}{2}})(Y,\xi_j) (1 - \widehat{\rho}(\xi_n/\Lambda) \widehat{\rho}(\xi_j/\Lambda)) \widehat{\psi}(\xi_n)}{ 2\sigma(H_0^{-1})(X,\Xi-\xi_j-\xi_n,\hat{x}_{j,n},Y,X,\xi)^{-1} } \diff \xi_n \bigg|^2 \diff \xi_j \notag \\
		&\leq C \int_{\mathbb{R}^3} \left( \int_{\mathbb{R}^3} \frac{\jap{\xi_j}^{-\frac{1}{2}}\jap{\xi_n}^{\frac{1}{2}} |\widehat{\psi}(\xi_n)|}{(\jap{\Xi-\xi_j-\xi_n}^2 + \Omega(\xi))^\frac{1}{2}} \frac{\jap{\xi_n}^{-1} |1-\widehat{\rho}(\xi_n/\Lambda) \widehat{\rho}(\xi_j/\Lambda)|}{(\jap{\Xi-\xi_j-\xi_n}^2 + \Omega(\xi))^\frac{1}{2}} \diff\xi_n \right)^2 \diff \xi_j \notag \\
		&\leq C \int_{\mathbb{R}^3} \left(\int_{\mathbb{R}^3} \frac{\jap{\xi_j}^{-1}\jap{\xi_n} |\widehat{\psi}(\xi_n)|^2}{\jap{\Xi-\xi_j-\xi_n}^2 + \Omega(\xi)} \diff \xi_n \right) \left( \int_{\mathbb{R}^3} \frac{\jap{\xi_n}^{-2} |1-\widehat{\rho}(\xi_n/\Lambda) \widehat{\rho}(\xi_j/\Lambda)|^2}{\jap{\Xi-\xi_j-\xi_n}^2 + \Omega(\xi)} \diff\xi_n \right) \diff \xi_j. \notag
	\end{align}
	The factor $|1-\widehat{\rho}(\xi_n/\Lambda) \widehat{\rho}(\xi_j/\Lambda)|$ is bounded by
	\begin{align}
		|1- \widehat{\rho}(\xi_n/\Lambda) \widehat{\rho}(\xi_j/\Lambda)| &= |1-\widehat{\rho}(\xi_n/\Lambda) + \widehat{\rho}(\xi_n/\Lambda)(1-\widehat{\rho}(\xi_j/\Lambda))| \notag \\
		&\leq 2(|\zeta_\Lambda(\xi_n)| + |\zeta_\Lambda(\xi_j)|),
		\label{eq:InequalityRhoSquared}
	\end{align}
	where $\zeta_{\Lambda}(\xi_j) = 1-\widehat{\rho}(\xi_j/\Lambda)$. We apply Corollary \ref{cor:IntegralEstimate} to estimate the second $\xi_n$-integral:
	\begin{align}
		\int_{\mathbb{R}^3} \frac{\jap{\xi_n}^{-2} (|\zeta_\Lambda(\xi_n)| + |\zeta_\Lambda(\xi_j)|)}{\jap{\Xi-\xi_j-\xi_n}^2 + \Omega(\xi)} \diff\xi_n &\leq 
		C[f(\Lambda) + |\zeta_\Lambda(\xi_j)|]\Omega(\hat{\xi}_n)^{-\frac{1}{2}+\epsilon} \notag \\
		&\leq C[f(\Lambda) + |\zeta_\Lambda(\xi_j)|] (n-1)^{-\frac{1}{2}+2\epsilon} \jap{\xi_j}^{-\epsilon},
	\end{align}
	for every $\epsilon > 0$ sufficiently small; thus,
	\begin{align}
		&\norm{(o_{j,\Lambda}-o_j)(X,Y,\Xi,\hat{x}_{j,n},\hat{\xi}_{j,n}) \psi}_{L^2}^2 \notag \\
		&\ \ \ \leq C (n-1)^{-\frac{1}{2}+2\epsilon} \int_{\mathbb{R}^3} \left( \int_{\mathbb{R}^3} \frac{\jap{\xi_j}^{-1-\epsilon} [f(\Lambda) + |\zeta_\Lambda(\xi_j)|]}{\jap{\Xi-\xi_j-\xi_n}^2 + \Omega(\xi)} \diff \xi_j \right) \jap{\xi_n} |\widehat{\psi}(\xi_n)|^2 \diff \xi_n \notag \\
		&\ \ \ \leq f(\Lambda) (n-1)^{-\frac{1}{2}+2\epsilon} \int_{\mathbb{R}^3} \jap{\xi_n} |\widehat{\psi}(\xi_n)|^2 \diff \xi_n 
		\leq f(\Lambda) (n-1)^{-\frac{1}{2}+2\epsilon} \norm{\omega^{\frac{1}{2}}\psi}^2_{L^2}.
		\label{eq:TodBound}
	\end{align}
	In the last inequality, we used that $\omega$ is elliptic and applied Corollary \ref{cor:RelativeBoundednessPseudor}. The operator norms of the derivatives of $o_{j,\Lambda}-o_j$ are bounded by similar estimates; hence, $o_{j,\Lambda}-o_j \in S(f(\Lambda) (n-1)^{-1/4+\epsilon},\mathfrak{B}(D(\omega^{1/2}),L^2))$. From the Cauchy--Schwarz inequality and the Calderon--Vaillancourt Theorem it follows that
	\begin{align}
		\norm{(T_{\mathrm{od},\Lambda}-T_{\mathrm{od}})\Psi^{(n-1)}}^2_{\mathfrak{H}^{(n-1)}}
		&\leq (n-1) \sum_{j=1}^{n-1} \norm{\Op(o_{j,\Lambda}-o_j)\Psi^{(n-1)}}^2_{\mathfrak{H}^{(n-1)}} \notag \\
		%&\leq f(\Lambda) (n-1) \sum_{j=1}^{n-1} (n-1)^{-\frac{1}{2}+2\epsilon} \norm{\omega_j^\frac{1}{2}\Psi^{(n-1)}}^2_{\mathfrak{H}^{(n-1)}} \notag \\
		&\leq f(\Lambda) (n-1)^{\frac{1}{2}+2\epsilon} \sum_{j=1}^{n-1} \scp{\Psi^{(n-1)}}{\omega_j \Psi^{(n-1)}}_{\mathfrak{H}^{(n-1)}} \notag \\
		&= f(\Lambda) \norm{N^{\frac{1}{4}+\epsilon}\diff\Gamma(\omega)^{\frac{1}{2}}\Psi^{(n-1)}}^2_{\mathfrak{H}^{(n-1)}},
		\label{eq:MPTod}
	\end{align}
	where $\omega_j = \mathbb{1}^{\otimes(j-1)} \otimes \omega \otimes \mathbb{1}^{\otimes(n-1-j)}$. Summing over $n$ yields the desired inequality.
\end{proof}

\begin{prop}
	For every $\epsilon>0$, $T_\mathrm{d}$ is continuous as an operator from $H_0^\epsilon$ to $\mathfrak{H}$. Moreover, $T_{\mathrm{d},\Lambda}$ converges in norm to $\mathrm{T}_\mathrm{d}$ as an operator from $D(H_0^\epsilon)$ to $\mathfrak{H}$.
	\label{prop:Td}
\end{prop}

\begin{proof}
	We prove that $\norm{(T_{\mathrm{d},\Lambda} - T_{\mathrm{d}})\Psi}_\mathfrak{H} \leq f(\Lambda) \norm{H_0^\epsilon \Psi}_\mathfrak{H}$ for every $\Psi \in D(H_0^\epsilon)$. 
	
	The action of $T_{\mathrm{d},\Lambda} - T_{\mathrm{d}}$ on the subspace $\mathfrak{H}^{(n-1)}$ is described by the pseudo-differential operator $\Op(d_\Lambda - d)$ with
	\begin{align}
	 	&(d_\Lambda-d)(X,\Xi,\hat{x}_n,\hat{\xi}_n) \notag \\
	 	&\ \ \ = \frac{1}{2} \int_{\mathbb{R}^3} \frac{|\sigma(\omega^{-\frac{1}{2}})(X,\xi_n)|^2 (1-|\widehat{\rho}(\xi_n/\Lambda)|^2) }{\sigma(H_0^{-1})(X,\Xi-\xi_n,\hat{x}_{n},X,\xi)^{-1}[\sigma(K)(X,\xi_{n}) + \sigma(\omega)(X,\xi_{n})]} \notag \\ 
	 	&\ \ \ \ \ \ \ \ \ [\sigma(K)(X,\xi_{n}) + \sigma(\omega)(X,\xi_{n}) -\sigma(H_0^{-1})(X,\Xi-\xi_n,\hat{x}_{n},X,\xi)^{-1}] \diff \xi_n.
	\end{align}
 	From Lemma \ref{lem:SymbolH0-1} below it follows that $\sigma(H_0^{-1})^{-1} = \sigma(H_0) + S(\jap{\Xi})$; hence,
 	\begin{align}
 		&|\sigma(K)(X,\xi_{n}) + \sigma(\omega)(X,\xi_{n}) - \sigma(H_0^{-1})(X,\Xi-\xi_n,\hat{x}_{n},X,\xi)^{-1}| \notag \\ 
 		&= |-\Xi\cdot g(X)\Xi + \xi_n \cdot g(X)\Xi + \Xi \cdot g(X) \xi_n + S(\jap{\xi_n}) - \sum_{j=1}^{n-1} \sigma(\omega)(x_j,\xi_j) + S(\jap{\Xi-\xi_n})| \notag \\
 		&\leq C [|\Xi|^2 + 2|\xi_n||\Xi| + \Omega(\xi) + \jap{\Xi-\xi_n}]
 	\end{align}
 	and
 	\begin{align}
 		|(d_\Lambda-d)(X,\Xi,\hat{x}_n,\hat{\xi}_n)| \leq C \int_{\mathbb{R}^3} \frac{|\Xi|^2 + 2|\xi_n||\Xi| + \Omega(\xi) + \jap{\Xi-\xi_n}}{(\jap{\Xi-\xi_n}^2 + \Omega(\xi))\jap{\xi_n}^3} |\zeta_\Lambda(\xi_n)| \diff\xi_n.
 		\label{eq:EstimateD}
 	\end{align}
	For every $0<\epsilon<1$, $\jap{\xi_n}^2 = 1 + |\xi_n|^2 \geq |\xi_n|^{2(1-\epsilon)}$. We use this inequality to bound the first two summands in \eqref{eq:EstimateD} by
 	\begin{align}
 		\int_{\mathbb{R}^3} \frac{(|\Xi|^2 + 2|\xi_n||\Xi|) |\zeta_\Lambda(\xi_n)|}{(\jap{\Xi-\xi_n}^2 + \Omega(\xi))\jap{\xi_n}^3} \diff\xi_n 
 		&\leq C \int_{\mathbb{R}^3} \frac{(|\Xi|^2 + 2|\xi_n||\Xi|)|\zeta_\Lambda(\xi_n)|}{(|\Xi-\xi_n|^2 + 1)|\xi_n|^{3-\epsilon}} \diff \xi_n \notag \\
 		&\leq f(\Lambda) \left( \frac{|\Xi|^2}{|\Xi|^{2-2\epsilon}} + \frac{2|\Xi|}{|\Xi|^{1-2\epsilon}} \right) \leq f(\Lambda) |\Xi|^{2\epsilon},
 	\end{align}
 	for every $\epsilon >0$ sufficiently small. In the second inequality, we applied Lemma \ref{lem:IntegralEstimate}.
 	
 	To bound the third summand in \eqref{eq:EstimateD}, we use the same inequality for $\jap{\xi_n}^2$ and apply Corollary \ref{cor:IntegralEstimate}:
 	\begin{align}
 		\int_{\mathbb{R}^3} \frac{\Omega(\xi) |\zeta_\Lambda(\xi_n)|}{(\jap{\Xi-\xi_n}^2 + \Omega(\xi))\jap{\xi_n}^3} \diff\xi_n
 		\leq f(\Lambda) \Omega(\hat{\xi}_n)^\epsilon.
 	\end{align}
 	Lastly, the fourth summand in \eqref{eq:EstimateD} is bounded by $f(\Lambda)$:
 	\begin{align}
 		\int_{\mathbb{R}^3} \frac{\jap{\Xi-\xi_n} |\zeta_\Lambda(\xi_n)|}{(\jap{\Xi-\xi_n}^2 + \Omega(\xi))\jap{\xi_n}^3} \diff\xi_n \leq f(\Lambda).
 	\end{align} 	
 	It follows that
 	\begin{align}
 		|(d_\Lambda - d)(X,\Xi,\hat{x}_n,\hat{\xi}_n)| &\leq f(\Lambda) (|\Xi|^{2\epsilon} + \Omega(\hat{\xi}_n)^\epsilon + 1) \notag \\
 		&\leq f(\Lambda) [(\jap{\Xi}^2 + \Omega(\hat{\xi}_n))^\epsilon + 1] = f(\Lambda) [M(\Xi,\hat{\xi}_n)^\epsilon + 1].
 	\end{align}
 	The operator norms of the derivatives of $d_\Lambda - d$ are bounded by similar estimates; hence, $d_\Lambda - d \in S(M^\epsilon + 1)$. We conclude that, for every $\Psi^{(n-1)} \in \mathfrak{H}^{(n-1)} \cap D(H_0^\epsilon)$,
 	\begin{align}
 		\norm{(T_{\mathrm{d},\Lambda} - T_{\mathrm{d}})\Psi^{(n-1)}}_{\mathfrak{H}^{(n-1)}} &\leq f(\Lambda) \norm{ (M(\Xi,\hat{\xi}_n)^\epsilon +1) \widehat{\Psi}^{(n-1)}}_{\mathfrak{H}^{(n-1)}} \notag \\
 		&\leq f(\Lambda) \norm{(H_0^\epsilon+1) \Psi^{(n-1)}}_{\mathfrak{H}^{(n-1)}}.
 	\end{align}
 	Summing over $n$ yields the desired inequality.
\end{proof}

\begin{lem}
	Let $\sigma(H_0)$ be the symbol of $H_0 \vert_{\mathfrak{H}^{(n)}}$ and $\sigma(H_0^{-1})$ the symbol of $H_0^{-1} \vert_{\mathfrak{H}^{(n)}}$. Then $\sigma(H_0^{-1}) = \sigma(H_0)^{-1} + S(M^{-2}\jap{\Xi})$.
	\label{lem:SymbolH0-1}
\end{lem}

\begin{proof}
	We assume that the symbols are in the standard quantisation. The proof for arbitrary quantisations is similar. The symbol $\sigma(H_0)$ is given by
	\begin{align}
		\sigma(H_0)(X,\Xi,x,\xi) = \sigma(K)(X,\Xi) + \sum_{j=1}^{n} \sigma(\omega)(x_j,\xi_j)
	\end{align}
	with 
	\begin{align}
		\sigma(K)(X,\Xi) &= \Xi \cdot g(X)\Xi + S^1, \\
		\sigma(\omega)(x,\xi) &= \sqrt{\xi \cdot g(x) \xi + m^2} + S^0.
	\end{align}
	From the functional calculus it follows that the symbol $\sigma(H_0^{-1})$ in leading order is equal to $\sigma(H_0)^{-1}$. To obtain the order of the remainder, we compute the Moyal product
	\begin{align}
		\sigma(H_0) \# \sigma(H_0)^{-1} = 1 + \partial_\Xi \sigma(H_0) D_X \sigma(H_0)^{-1} + \sum_{j=1}^{n} \partial_{\xi_j} \sigma(H_0) D_{x_j} \sigma(H_0)^{-1} + R,
	\end{align}
	where $R$ is a symbol of lower order than the preceding terms. Clearly, $\partial_\Xi \sigma(H_0) = \partial_\Xi \sigma(K) \in S(\jap{\Xi}) \subset S(M\jap{\Xi}^{-1})$ and $D_X \sigma(H_0)^{-1} = \sigma(H_0)^{-2} D_X \sigma(K) \in S(M^{-2}\jap{\Xi}^2)$. Thus, $\partial_\Xi \sigma(H_0) D_X \sigma(H_0)^{-1} \in S(M^{-1}\jap{\Xi})$. 
	
	Furthermore, $\partial_{\xi_j} \sigma(H_0) D_{x_j} \sigma(H_0)^{-1} \in S(M^{-1}) \subset S(M^{-1}\jap{\Xi})$. This proves the claim because $\sigma(H_0)^{-1} = \sigma(H_0^{-1}) \# \sigma(H_0) \# \sigma(H_0)^{-1} = \sigma(H_0^{-1}) + S(M^{-2}\jap{\Xi})$.
\end{proof}

\begin{prop}
	The operator $T_R$ is continuous as an operator from $D(N^{1/2})$ to $\mathfrak{H}$. Moreover, $T_{R,\Lambda}$ converges in norm to $\mathrm{T}_R$ as an operator from $D(\diff\Gamma(\omega)^\epsilon N^{1/2})$ to $\mathfrak{H}$ for every $\epsilon > 0$.
	\label{prop:TR}
\end{prop}

\begin{proof}
	We prove that $\norm{(T_{R,\Lambda} - T_R) \Psi}_\mathfrak{H} \leq f(\Lambda)\norm{N^{1/2}\diff\Gamma(\omega)^\epsilon \Psi}_\mathfrak{H}$ by bounding all summands of $T_{R,\Lambda} - T_R$ separately.
	
	%\begin{enumerate}[wide, labelwidth=!, labelindent=0pt]
	\item $\mathbf{T_{\tilde{R}} - T_{\tilde{R},\Lambda}}$ and $\mathbf{T_{R_1} - T_{R_1,\Lambda}}$: All the terms in this part of $T_{R,\Lambda} - T_R$ have a similar structure and are treated similarly as the remainders we encountered in Section~\ref{sec:MappingPropertiesOfG}. Exemplary, we prove the desired inequality for
	\begin{align}
		&a(\tilde{u}_{\Lambda,X}) H_0^{-1} a^*(u_{\Lambda,X}) - a(\tilde{u}_{X}) H_0^{-1} a^*(u_X) \notag \\ 
		&\ \ \ = a(\tilde{u}_{\Lambda,X} - \tilde{u}_X) H_0^{-1} a^*(u_{\Lambda,X}) - a(\tilde{u}_X) H_0^{-1} a^*(u_X - u_{\Lambda,X}).
	\end{align}
	The first summand is a relatively bounded operator with respect to $N^{1/2}$ because, by Proposition \ref{prop:GmapsDomainToItself}, for every $\Psi \in D(N^{1/2})$,
	\begin{align}
		\norm{a(\tilde{u}_{\Lambda,X} - \tilde{u}_X) H_0^{-1} a^*(u_{\Lambda,X})\Psi}_\mathfrak{H} &\leq \norm{\omega^{-\frac{1}{2}}(\tilde{u}_{\Lambda,X} - \tilde{u}_X)}_\mathfrak{h} \norm{\diff\Gamma(\omega)^{\frac{1}{2}}H_0^{-1}a^*(u_{\Lambda,X})\Psi }_\mathfrak{H} \notag \\
		&\leq C\norm{\rho_\Lambda - \delta}_{H^{-2}} \norm{(N+1)^{\frac{1}{2}}\Psi}_\mathfrak{H}.
	\end{align}	
	Similarly, the second summand is relatively bounded with respect to $N^{1/2}\diff\Gamma(\omega)^\epsilon$ because, for every $\Psi \in D(N^{1/2}\diff\Gamma(\omega)^\epsilon)$,
	\begin{align}
		\norm{a(\tilde{u}_X) H_0^{-1} a^*(u_X - u_{\Lambda,X})\Psi}_{\mathfrak{H}} &\leq \norm{\omega^{-\frac{1}{2}}\tilde{u}_X}_{\mathfrak{h}} \norm{\diff\Gamma(\omega)^{\frac{1}{2}} H_0^{-1} a^*(u_X - u_{\Lambda,X})}_{\mathfrak{H}} \notag \\
		&\leq f(\Lambda) \norm{(N+1)^\frac{1}{2}\diff\Gamma(\omega)^\epsilon \Psi}_{\mathfrak{H}}
	\end{align}
	according to Proposition \ref{prop:GmapsDomainToItself}. 	
%	Another term we consider is
%	\begin{align}
%		\frac{1}{\Lambda} \norm{a(\omega^{-\frac{1}{2}} \tau_{\Lambda,X}) \Op(R_1) a^*(u_{\Lambda,X})}_{\mathfrak{H}} \leq \frac{1}{\Lambda} \norm{\omega^{-1}\tau_{\Lambda,X}}_\mathfrak{h} \norm{\Psi}_\mathfrak{H}.
%	\end{align}
%	The factor is asymptotically equivalent to $\Lambda^{-1/2}$ for large $\Lambda$; hence converges to 0.
		
	\item $\mathbf{T_{R_{\mathrm{od}}} - T_{R_{\mathrm{od}},\Lambda}}$: The first off-diagonal remainder is treated like the leading order term we discussed in Proposition \ref{prop:Tod}. The additional derivative in \eqref{eq:Rod1st} ensures that the symbol defining the pseudo-differential operator $T_{R_{\mathrm{od},1}} - T_{R_{\mathrm{od},1},\Lambda}$ is an element of $S(f(\Lambda),\mathfrak{B}(L^2,L^2))$. %and not only of $S(1,\mathfrak{B}(L^2,D(\omega^{1/2})))$ as in the proof of Proposition \ref{prop:Tod}. 
	Following the same steps as in the proof of Proposition \ref{prop:Tod}, we confirm that the first off-diagonal remainder satisfies the desired inequality.
		
	The second off-diagonal remainder \eqref{eq:Rod2nd} is a finite sum of operator-valued pseudo-differential operators as in \eqref{eq:TodLambda} with symbols
	\begin{align}
		&q_{j,\Lambda}(X,Y,\Xi, \hat{x}_{j,n}, \hat{\xi}_{j,n}) \psi(x_j) = \frac{1}{2}\sum_{k,l=1}^3 \int_{\mathbb{R}^6} \e^{\I\scp{x_j}{\xi_j}} \widehat{\rho}(\xi_n/\Lambda) \widehat{\psi}(\xi_n) \notag \\ 
		&\ \ \ \bigg\{ D_{\xi_j}^{(k)} D_{\xi_j}^{(l)} h_{j,\Lambda}^{kl}(X,Y,\Xi-\xi_j-\xi_n,\hat{x}_n,X,\xi) \widehat{\rho}(\xi_j/\Lambda) \notag \\
		&\ \ \ + \frac{2}{\Lambda} D_{\xi_j}^{(k)} h_{j,\Lambda}^{kl}(X,Y,\Xi-\xi_j-\xi_n,\hat{x}_n,X,\xi) D_{\xi_j}^{(l)} \widehat{\rho}(\xi_j/\Lambda)   \notag \\
		&\ \ \ + \frac{1}{\Lambda^2} h_{j,\Lambda}^{kl}(X,Y,\Xi-\xi_j-\xi_n,\hat{x}_n,X,\xi) D_{\xi_j}^{(k)} D_{\xi_j}^{(l)} \widehat{\rho}(\xi_j/\Lambda) \bigg \} \diff \xi_n \diff \xi_j,
		\label{eq:Rest2ndOd}
	\end{align}
	where
	\begin{align}
		&h_{j,\Lambda}^{kl}(X,Y,\Xi,\hat{x}_n,X,\xi) = \sigma(\omega^{-\frac{1}{2}})(Y,\xi_j) R_{\mathrm{od},2}^{kl}(X,\Xi,\hat{x}_{n},X,\xi) \overline{\sigma}(\omega^{-\frac{1}{2}})(X,\xi_{n}).
	\end{align}
	We numerate the three summands in \eqref{eq:Rest2ndOd} by $q_{j,\Lambda}^{(0)}$, $q_{j,\Lambda}^{(1)}$, and $q_{j,\Lambda}^{(2)}$. From the Cauchy--Schwarz inequality and \eqref{eq:InequalityRhoSquared} it follows that
	\begin{align}
		&|(q_{j,\Lambda}^{(0)}-q_j^{(0)})(X,Y,\Xi, \hat{x}_{j,n}, \hat{\xi}_{j,n}) \psi(x_j)| \notag \\ 
		&\leq C \int_{\mathbb{R}^3} \int_{\mathbb{R}^3} \frac{\jap{\xi_n}^{-\frac{1}{2}} \jap{\xi_j}^{-\frac{5}{2}}  |1-\widehat{\rho}(\xi_j/\Lambda)\widehat{\rho}(\xi_n/\Lambda)|}{\jap{\Xi-\xi_j-\xi_n}^2 + \Omega(\xi)} |\widehat{\psi}(\xi_n)| \diff \xi_n \diff \xi_j \notag \\ 
		&\leq C \left( \int_{\mathbb{R}^3} \bigg| \int_{\mathbb{R}^3} \frac{\jap{\xi_n}^{-\frac{1}{2}} \jap{\xi_j}^{-\frac{5}{2}} (|\zeta_\Lambda(\xi_j)| + |\zeta_\Lambda(\xi_n)|)}{\jap{\Xi-\xi_j-\xi_n}^2 + \Omega(\xi)} \diff \xi_j \bigg|^2 \diff \xi_n \right)^{1/2} \norm{\psi}_{L^2}.
	\end{align}
	The prefactor in front of $\norm{\psi}_{L^2}$ is bounded by $f(\Lambda)$:
	\begin{align}
		&\int_{\mathbb{R}^3} \bigg| \int_{\mathbb{R}^3} \frac{\jap{\xi_n}^{-\frac{1}{2}} \jap{\xi_j}^{-\frac{5}{2}} (|\zeta_\Lambda(\xi_n)| + |\zeta_\Lambda(\xi_j)|)}{\jap{\Xi-\xi_j-\xi_n}^2 + \Omega(\xi)} \diff \xi_j \bigg|^2 \diff \xi_n \notag \\
		&\ \ \ \leq \int_{\mathbb{R}^3} \left( \int_{\mathbb{R}^3} \frac{\jap{\xi_n}^{-1}\jap{\xi_j}^{-3}}{\jap{\Xi-\xi_j-\xi_n}^2 + \Omega(\xi)} \diff \xi_j \right) \left( \int_{\mathbb{R}^3} \frac{\jap{\xi_j}^{-2} (|\zeta_\Lambda(\xi_n)| + |\zeta_\Lambda(\xi_j)|)}{\jap{\Xi-\xi_j-\xi_n}^2 + \Omega(\xi)} \diff \xi_j \right) \diff \xi_n \notag \\
		&\ \ \ \leq C \int_{\mathbb{R}^3} \int_{\mathbb{R}^3} \frac{\jap{\xi_n}^{-1}\jap{\xi_j}^{-3} (f(\Lambda) \Omega(\hat{\xi}_j)^\epsilon + |\zeta_\Lambda(\xi_n)|)\Omega(\hat{\xi}_j)^{-\frac{1}{2}}}{\jap{\Xi-\xi_j-\xi_n}^2 + \Omega(\xi)} \diff \xi_j \diff \xi_n \leq f(\Lambda).
	\end{align}
	Furthermore, for $\beta = 1$ or $\beta = 2$,
	\begin{align}
		|(q_{j,\Lambda}^{(\beta)}-q_j^{(\beta)})(X,Y,\Xi, \hat{x}_{j,n}, \hat{\xi}_{j,n}) \psi(x_j)| = |q_{j,\Lambda}^{(\beta)}(X,Y,\Xi, \hat{x}_{j,n}, \hat{\xi}_{j,n}) \psi(x_j)| \notag \\ 
		\leq \frac{1}{\Lambda^\beta} \int_{\mathbb{R}^3} \bigg| \int_{\mathbb{R}^3} \frac{\jap{\xi_n}^{-\frac{1}{2}} \jap{\xi_j}^{-\frac{5}{2}+\beta} |\widehat{\tau}_\beta(\xi_j/\Lambda)||\widehat{\rho}(\xi_n/\Lambda)| |\widehat{\psi}(\xi_n)|}{\jap{\Xi-\xi_j-\xi_n}^2 + \Omega(\xi)} \diff \xi_j \bigg|^2 \diff \xi_n,
		\label{eq:qLambda12}
	\end{align}
	where $\tau_1(x) = x\rho(x)$ and $\tau_2(x) = |x|^2 \rho(x)$. Because $\tau_1$ and $\tau_2$ are Schwartz functions, $|\widehat{\tau}_\beta(\xi_j/\Lambda)| \leq C \Lambda^{\beta-\epsilon} \jap{\xi_j}^{-\beta + \epsilon}$ for every $\epsilon >0$; hence, \eqref{eq:qLambda12} is bounded by $\Lambda^{-\epsilon}$ times a constant.
	
	We argue along similar lines to prove that $x_j^\alpha (q_{j,\Lambda}-q_j)(X,Y,\Xi, \hat{x}_{j,n}, \hat{\xi}_{j,n}) \psi(x_j)$ is bounded by $f(\Lambda) \norm{\psi}_{L^2}$ for every multi-index $\alpha \in \mathbb{N}^3$. It follows that the operator norm of $q_{j,\Lambda} - q_j$ is bounded by $f(\Lambda)$ (cf. the proof of Theorem \ref{thm:L2Continuity}). The operator norms of the derivatives of $q_{j,\Lambda} - q_j$ are bounded by similar techniques; hence, $q_{j,\Lambda} - q_j \in S(f(\Lambda),\mathfrak{B}(L^2,L^2))$ and the second off-diagonal remainder satisfies the desired estimate.
		
	\item $\mathbf{T_{R_\mathrm{d}} - T_{R_\mathrm{d},\Lambda}}$: The remainder in the diagonal part is treated like the leading order term we discussed in Proposition \ref{prop:Td}. The additional derivative in \eqref{eq:Rd} ensures that the symbol defining the pseudo-differential operator is an element of $S(1)$.
	%\end{enumerate}
	%Hence, all parts of $T_R$ are relatively bounded and so is $T_R$.
\end{proof}

%We use the mapping properties of $T$ established in the previous three propositions to show that $T$ is infinitesimally bounded with respect to the free IBC Hamiltonian $H_0' = (1-G)^*H_0(1-G)$. This guarantees that the IBC Hamiltonian $H_0'+T$ is self-adjoint on $D(H)=D(H_0')$.

\begin{thm}
	The operator $T$ is infinitesimally bounded with respect to $H_0' = (1-G)^*H_0(1-G)$, and $T_\Lambda + E_\Lambda(X)$ is infinitesimally bounded with respect to $(1-G_\Lambda)^*H_0(1-G_\Lambda)$ uniformly in $\Lambda$.
	\label{prop:InfinitesimalBoundForT}
\end{thm}

\begin{proof}
	We prove the Proposition for $T$. The proof for $T_\Lambda + E_\Lambda(X)$ is similar. The strategy is to show that $TG$ and $T(1-G)$ are infinitesimally bounded with respect to $H_0'$. Then $T = T(1-G)+TG$ is infinitesimally bounded with respect to $H_0'$.
	
	We select an $0<\epsilon<1/4$. From Propositions \ref{prop:Tod}, \ref{prop:Td}, and \ref{prop:TR} it follows that
	\begin{align}
		\norm{T\Psi}_\mathfrak{H} %&\leq C (\norm{H_0^\epsilon \Psi}_\mathfrak{H} + \norm{N^{\frac{1}{4}+\epsilon}\diff\Gamma(\omega)^{\frac{1}{2}}\Psi}_\mathfrak{H} + \norm{(N+1)^{\frac{1}{2}}\Psi}_{\mathfrak{H}} + \norm{\Psi}_{\mathfrak{H}}) \notag \\
		&\leq C (\norm{H_0^\epsilon \Psi}_\mathfrak{H} + \norm{N^{\frac{1}{4}+\epsilon}\diff\Gamma(\omega)^{\frac{1}{2}}\Psi}_\mathfrak{H} + \norm{\Psi}_{\mathfrak{H}}).
	\end{align}
	Hence, by the boundedness of the inverse operator of $(1-G)^*$ (see Corollary \ref{cor:BoundedInverse1-G}), 
	\begin{align}
		\norm{T(1-G)\Psi}_\mathfrak{H} &\leq C(\norm{(1-G)^*H_0^\epsilon (1-G) \Psi}_\mathfrak{H} \notag \\
		&\ \ \ + \norm{(1-G)^*N^{\frac{1}{4}+\epsilon}\diff\Gamma(\omega)^{\frac{1}{2}}(1-G)\Psi}_\mathfrak{H} + \norm{\Psi}_\mathfrak{H})
	\end{align}
	For any $\gamma > 0$, $H_0^{2\epsilon} \leq \gamma H_0^2 + C_\gamma$ by Young's inequality. Likewise, $N^{1/2+2\epsilon}\diff\Gamma(\omega) \leq H_0^{3/2+2\epsilon} \leq \gamma H_0^2 + C_\gamma$. It follows that, for every $\Psi \in D(H)$,
	\begin{align}
		\norm{T(1-G)\Psi}_\mathfrak{H} &\leq 2\gamma \norm{(1-G)^*H_0 (1-G) \Psi}_\mathfrak{H} + C_\gamma \norm{\Psi}_\mathfrak{H} \notag \\
		&= 2\gamma \norm{H_0' \Psi}_\mathfrak{H} + C_\gamma \norm{\Psi}_\mathfrak{H},
	\end{align}
	that is, $T(1-G)$ is infinitesimally bounded with respect to $H_0'$. 
	
	Next, we observe that by Proposition \ref{prop:GmapsDomainToItself}, Proposition \ref{prop:H0GBounded}, and Young's inequality, for every $\gamma > 0$,
	\begin{align}
		\norm{TG\Psi}_\mathfrak{H} &\leq C (\norm{H_0^\epsilon G \Psi}_\mathfrak{H} + \norm{N^{\frac{1}{4}+\epsilon}\diff\Gamma(\omega)^{\frac{1}{2}}G\Psi}_\mathfrak{H} + \norm{\Psi}_\mathfrak{H}) \notag \\
		&\leq C(\norm{N^{\frac{3}{4}+\epsilon}\Psi}_\mathfrak{H} + \norm{\Psi}_\mathfrak{H})
		\leq \gamma \norm{N\Psi}_\mathfrak{H} + C_\gamma \norm{\Psi}_\mathfrak{H}.
	\end{align}
	We conclude that $TG$ is infinitesimally bounded with respect to $H_0'$ because
	\begin{align}
		\norm{N\Psi}_\mathfrak{H} &\leq C (\norm{N(1-G)\Psi}_\mathfrak{H} + \norm{\Psi}_\mathfrak{H}) \notag \\ &\leq C (\norm{(1-G)^*H_0 (1-G)\Psi}_\mathfrak{H} + \norm{\Psi}_\mathfrak{H})
	\end{align}
	according to Corollary \ref{cor:BoundedInverse1-G}, eq. \eqref{eq:1-GBoundedInverse}.
\end{proof}

\section{Regularity of Domain Vectors}
\label{sec:RegularityDomainVectors}

In this section, we analyse the regularity of vectors in the domain $D(H)$ of the IBC Hamiltonian $H$. The first observation we make is that $\Psi \in D(H)$ is as regular as $G\Psi$ because $\Psi = G\Psi + (1-G)\Psi$ and $(1-G)\Psi \in D(H_0)$ is regular by assumption. In the following two theorems, we prove that $D(H)$ is contained in $D(H_0^p)$ for every $p<1/2$ but that $D(H) \cap D(H_0^{1/2}) = \{0\}$. 

%The latter implies that the IBC Hamiltonian is not a perturbation of the free Nelson Hamiltonian.

\begin{thm}
	For every $p < 1/2$, the domain $D(H)$ of the IBC Hamiltonian $H$ is contained in $D(H_0^p)$ and in $D(\diff\Gamma(\omega)^{1/2})$.
	\label{thm:DomainH}
\end{thm}

\begin{proof}
	For every $\Psi \in D(H)$, $(1-G)\Psi \in D(H_0) \subset D(H_0^p)$; hence, $\Psi \in D(H_0^p)$ if and only if $G\Psi \in D(H_0^p)$. Moreover, $D(H) \subset (1-G)^{-1}D(N) \subset D(N) \subset D(N^{1/2})$ because $(1-G)^{-1}$ preserves $D(N)$ by Corollary \ref{cor:BoundedInverse1-G}. Thus, by Proposition \ref{prop:H0GRelativelyBounded},
	\begin{align}
		\norm{H_0^p G \Psi}_\mathfrak{H} \leq \norm{(N+1)^{\frac{1}{2}}\Psi}_\mathfrak{H} \leq \norm{(H+1) \Psi}_\mathfrak{H} < \infty.
		\label{eq:Regularity}
	\end{align}
	 This confirms that $G\Psi \in D(H_0^p)$. The proof for $D(H) \subset D(\diff\Gamma(\omega)^{1/2})$ is the same if we replace Proposition \ref{prop:H0GRelativelyBounded} with Proposition \ref{prop:H0GBounded}.
\end{proof}

%Interesting is that the above theorem is sharp in the sense that $D(H)$ and $D(H_0^{1/2})$ have trivial intersection. %This makes clear once again that it is not possible to treat the interacting field in the Nelson model as a perturbation of the free Hamiltonian, neither on the level of operators nor on the level of forms.

\begin{thm}
	The domain $D(H)$ of the IBC Hamiltonian $H$ and $D(H_0^{1/2})$ intersect trivially, that is, $D(H) \cap D(H_0^{1/2}) = \{ 0 \}$.
\end{thm}

\begin{proof}
	It suffices to prove that $G$ maps no nonzero vector $0 \neq \tilde{\Psi} \in D(H)$ into $D(H_0^{1/2})$ because $\tilde{\Psi} \in D(H_0^{1/2})$ if and only if $G\tilde{\Psi} \in D(H_0^{1/2})$. If $\tilde{\Psi} \in D(H)$ is nonzero, then there is at least one $n\in \mathbb{N}^*$ such that $\Psi \equiv \tilde{\Psi}^{(n-1)} \neq 0$.
	
	According to Lemma \ref{lem:DecompositionG}, a bounded operator $B$ exists such that
	\begin{align}
		H_0^{\frac{1}{2}}G\Psi = -\frac{1}{\sqrt{n}} \sum_{j=1}^n \Op(m_j)\Psi + B\Psi.
	\end{align}	
	The summand $B\Psi$ is an element of $\mathfrak{H}$ because $B$ is bounded. We prove that the first summand is not an element of $\mathfrak{H}$. To prove this, we demonstrate that
	\begin{align}
		\frac{1}{n}\int \bigg|\sum_{j=1}^n \Op(m_j)\Psi(X,x)\bigg|^2 \diff X \diff x &= \frac{1}{n} \int \bigg|\sum_{j=1}^n \mathcal{F}[\Op(m_j)\Psi](\Xi,\xi)\bigg|^2 \diff \Xi \diff\xi
		\label{eq:RegularityVector}
	\end{align}
	diverges. We set $U = B_r(0) \times B_R(0) \times B_r(0)^{n-1} \subset \mathbb{R}^3 \times \mathbb{R}^{3n}$, where $B_r(0)$ is the open ball with centre $0$ and radius $r$. Clearly,
	\begin{align}
		\frac{1}{n} \int \bigg|\sum_{j=1}^n \mathcal{F}[\Op(m_j)\Psi](\Xi,\xi)\bigg|^2 \diff \Xi \diff\xi \geq \frac{1}{n} \int_U \bigg|\sum_{j=1}^n \mathcal{F}[\Op(m_j)\Psi](\Xi,\xi)\bigg|^2 \diff \Xi \diff\xi.
	\end{align}
	From $(a+b)^2 \geq a^2/2 - b^2$ and the Cauchy--Schwarz inequality it follows that\footnote{The inequality $(a+b)^2 \geq a^2/2 - b^2$ is equivalent to $(a+2b)^2 \geq 0 $.}
	\begin{align}
		&\frac{1}{n} \bigg|\sum_{j=1}^n \mathcal{F}[\Op(m_j)\Psi](\Xi,\xi) \bigg|^2 \notag \\
		&\ \ \ \geq \frac{1}{2n} |\mathcal{F}[\Op(m_1)\Psi](\Xi,\xi)|^2 - \sum_{j=2}^n |\mathcal{F}[\Op(m_j)\Psi](\Xi,\xi)|^2.
		\label{eq:RegularitySummands}
	\end{align}
	We recall that $\mathcal{F}[\Op(m_j)\Psi](\Xi,\xi)$ is the Fourier transform in the variables $(X,\hat{x}_j)$ of
	\begin{align}
		\sum_{j=1}^n \int \e^{\I\scp{X-Y}{\Xi}} \e^{\I\scp{\hat{x}_j-\hat{y}_j}{\hat{\xi}_j}} \frac{\e^{-\I \scp{X}{\xi_j}} \sigma(\omega^{-\frac{1}{2}})(Y,\xi_j) \Psi(Y,\hat{y}_j)}{\sqrt{2}\sigma(H_0^{-1})(Y,\Xi-\xi_j,\hat{x}_j,Y,\xi)^{-1}} \diff Y\diff\Xi \diff \hat{y}_j \diff\hat{\xi}_j;
	\end{align}
	see \eqref{eq:H0GLeadingOrder}. The sum over $2\leq j\leq n$ in \eqref{eq:RegularitySummands} has a finite integral over $U$ because, for every $2\leq j \leq n$,
	\begin{align}
		&\int_U |\mathcal{F}[\Op(m_j)\Psi](\Xi,\xi)|^2 \diff\Xi \diff\xi
		=  \int_U |\mathcal{F}[\Op(\mathcal{F}^{-1}[\mathcal{F}(m_j)\mathbb{1}_{|\cdot| \leq r}])\Psi](\Xi,\xi)|^2 \diff\Xi \diff\xi \notag \\
		&\ \ \ \leq \int |\mathcal{F}[\Op(\mathcal{F}^{-1}[\mathcal{F}(m_j)\mathbb{1}_{|\cdot| \leq r}])\Psi](\Xi,\xi)|^2 \diff\Xi \diff\xi.
		\label{eq:RegularityIntegral2n}
	\end{align}
	The symbol $\mathcal{F}^{-1}[\mathcal{F}(m_j)\mathbb{1}_{|\cdot| \leq r}]$ is an element of $S(1,\mathfrak{B}(\mathbb{C},L^2))$ for every $r>0$. Due to the cut-off $\mathbb{1}_{|\cdot| \leq r}$, the operator norm of $\mathcal{F}^{-1}[\mathcal{F}(m_j)\mathbb{1}_{|\cdot| \leq r}]$ is bounded by a constant times
	\begin{align}
		\sup_{\Xi} \int_{|\xi_j|\leq r} \frac{\jap{\xi_j}^{-1}}{\jap{\Xi-\xi_j}^2 + 1} \diff\xi_j.
	\end{align}
	Similarly, the operator norms of the derivatives of $\mathcal{F}^{-1}[\mathcal{F}(m_j)\mathbb{1}_{|\cdot| \leq r}]$ are bounded. It follows from the Calderon--Vaillancourt Theorem that the integral \eqref{eq:RegularityIntegral2n} is finite for every $2\leq j\leq n$ and every $r>0$.
	
	For the summand with $m_1$ in \eqref{eq:RegularitySummands}, we use the same equality as in the first line of \eqref{eq:RegularityIntegral2n} with $r$ replaced by $R$. However, this time, we bound the operator norm of the symbol $\mathcal{F}^{-1}[\mathcal{F}(m_1)\mathbb{1}_{|\cdot| \leq R}]$ from below by a constant times
	\begin{align}
		\int_{|\xi_1| \leq R} \frac{\jap{\xi_1}^{-1}}{\jap{\Xi-\xi_1}^2 + \Omega(\xi)} \diff\xi_1.
	\end{align}
	This is possible because the symbols $\sigma(\omega^{-\frac{1}{2}})$ and $\sigma(H_0^{-1/2})$ are elliptic. It follows from Corollary \ref{cor:RelativeBoundednessPseudor} that, up to finite term $A$,
	\begin{align}
		&\int_U |\mathcal{F}[\Op(m_1)\Psi](\Xi,\xi)|^2 \diff\Xi \diff\xi \notag \\
		&\ \ \ \geq C \int_{B_r(0)^{n-1}} |\widehat{\Psi}(\Xi,\hat{\xi}_1)|^2 \left(\int_{|\xi_1| \leq R} \frac{\jap{\xi_1}^{-1}}{\jap{\Xi-\xi_1}^2 + \Omega(\xi)} \diff\xi_1\right) \diff \Xi \diff\hat{\xi}_1 + A.
		\label{eq:DivergentIntegral}
	\end{align}
	In the integration area of the integral on the right side, $|\Xi|\leq r$ and $\Omega(\hat{\xi}_1) \leq (n-1)\jap{r}$; thus,
	\begin{align}
		\jap{\Xi-\xi_1}^2 + \Omega(\xi) \leq \jap{\Xi-\xi_1}^2 + \jap{\xi_1} + C_r \leq C_r (\jap{\xi_1}^2 + \jap{\xi_1})
	\end{align}
	by Peetre's inequality (cf. Theorem \ref{thm:Peetre}). Hence, \eqref{eq:DivergentIntegral} is bounded from below by a constant times
	\begin{align}
		\int_{B_r(0)^{n-1}} |\widehat{\Psi}(\Xi,\hat{\xi}_1)|^2 \diff \Xi \diff\hat{\xi}_1 \int_{|\xi_1|\leq R} \frac{\jap{\xi_1}^{-1}}{\jap{\xi_1}^2 + \jap{\xi_1}} \diff\xi_1.
		\label{eq:RegularityDivergence}
	\end{align}
	We fix $r<\infty$ such that $\int_{B_r(0)^{n-1}} |\widehat{\Psi}(\Xi,\hat{\xi}_1)|^2 \diff \Xi \diff\hat{\xi}_1 > 0$. This is possible because $\Psi \neq 0$. Then, \eqref{eq:RegularityDivergence} diverges in the limit $R\to \infty$. This confirms that \eqref{eq:RegularityVector} is divergent.
\end{proof}

%\begin{cor}
%	The IBC Hamiltonian is not a perturbation of the free Nelson Hamiltonian.
%\end{cor}

\section{Comparison with the Renormalised Hamiltonian}
\label{sec:ProofIBCHamiltonian}

We verify that the IBC Hamiltonian we constructed in this chapter is equivalent to the renormalised Nelson Hamiltonian from Chapter \ref{ch:RemovalCutOff}. To prove this, we demonstrate that $H_\Lambda + E_\Lambda(X)$ converges in the norm resolvent sense to the IBC Hamiltonian $H$. From the uniqueness of the limit it then follows that the IBC Hamiltonian is indeed the renormalised Nelson Hamiltonian.

\begin{thm}
	The operator $H_\Lambda + E_\Lambda(X)$ converges in the norm resolvent sense to the IBC Hamiltonian $H$ as $\Lambda \to \infty$.
\end{thm}

\begin{proof}	
	According to the resolvent formula, the difference in the resolvents between $H = H_0' + T$ and $H_\Lambda + E_\Lambda(X)$ is equal to
	\begin{align}
		&(H_\Lambda + E_\Lambda(X) + \I)^{-1} - (H+\I)^{-1} \notag \\ 
		%&\ \ \ = (H_\Lambda + E_\Lambda(X) + \I)^{-1} (H - H_\Lambda - E_\Lambda(X))(H+\I)^{-1} \notag \\
		&\ \ \ = (H_\Lambda + E_\Lambda(X) + \I)^{-1} (G_\Lambda - G)^*H_0(1-G) (H+\I)^{-1}  \notag \\
		&\ \ \ \ \ \ + (H_\Lambda + E_\Lambda(X) + \I)^{-1} (1-G_\Lambda)^* H_0 (G_\Lambda-G) (H+\I)^{-1}  \notag \\
		&\ \ \ \ \ \ + (H_\Lambda + E_\Lambda(X) + \I)^{-1} (T - T_\Lambda - E_\Lambda(X)) (H+\I)^{-1}.
		\label{eq:ResolventDifference}
	\end{align}
	%We prove that every summand converges in norm to 0.	
	By Proposition \ref{prop:InfinitesimalBoundForT}, the operator $T$ is relatively bounded with respect to $H_0' = (1-G)H_0(1-G)$. Hence, $H_0(1-G) (H+\I)^{-1}$ is a bounded operator, and the first summand converges in norm to $0$ because $G_\Lambda \to G$ in norm by Proposition \ref{prop:H0GBounded}.
%	\begin{align}
%		&\sup_{\norm{\Psi}=1} \norm{(H_\Lambda + E_\Lambda(X) + \I)^{-1} (G_\Lambda - G)^*H_0(1-G) (H+\I)^{-1}\Psi} \notag \\ 
%		&\ \ \ \leq \sup_{\norm{\Psi}=1} C \norm{(G_\Lambda - G) (H+\I)^{-1}\Psi} \notag \\ 
%		&\ \ \ \leq \sup_{\norm{\Psi}=1} f(\Lambda) \norm{(N+1)^{\frac{1}{2}}(H+\I)^{-1} \Psi} \leq f(\Lambda) \stackrel{\Lambda\to\infty}{\longrightarrow} 0.
%	\end{align}
	Similarly, we argue that the second summand in \eqref{eq:ResolventDifference} converges to 0.
	
	To obtain convergence of the third summand, we use $N \leq H_\Lambda + E_\Lambda(X)$ and the operator bound on $T_\Lambda + E_\Lambda(X) - T$ from Subsection \ref{ssec:MappingPropertiesOfT}:
	\begin{align}
		&\norm{(H_\Lambda + E_\Lambda(X) + \I)^{-1} (T_\Lambda + E_\Lambda(X) - T) (H+\I)^{-1} \Psi}_{\mathfrak{H}} \notag \\ 
		&\ \ \ \leq f(\Lambda) \norm{(\diff\Gamma(\omega)^{1/2} + H_0^\epsilon + \diff\Gamma(\omega)^\epsilon) (H+\I)^{-1} \Psi}_{\mathfrak{H}}
	\end{align}
	The operator $(\diff\Gamma(\omega)^{1/2} + H_0^\epsilon + \diff\Gamma(\omega)^\epsilon) (H+\I)^{-1}$ is bounded by Proposition \ref{thm:DomainH}. Hence, also the third summand converges in norm to $0$.
\end{proof}

\begin{cor}
	The IBC Hamiltonian is equivalent to the renormalised Nelson Hamiltonian from Chapter \ref{ch:RemovalCutOff}.
\end{cor}

\appendix
\chapter{Some Functional Analysis}

\setcounter{section}{1}

We collect some important and well-known theorems and inequalities from functional analysis. 

\begin{thm}[Schwartz kernel theorem]
	Let $X,Y \subset \mathbb{R}^d$ be open and $\mathcal{D}(X) = C_\mathrm{c}^\infty(X)$ the space of smooth compactly supported functions on $X$. Then every distribution $K\in \mathcal{D}'(X\times Y)$ defines a linear map $A:\mathcal{D}(X) \to \mathcal{D}'(Y)$ satisfying
	\begin{align}
		\scp{Au}{v} = K(u \otimes v)
		\label{eq:DualPairing}
	\end{align}
	for all $u \in \mathcal{D}(X)$, $v\in \mathcal{D}(Y)$. Conversely, every continuous linear map $A:\mathcal{D}(X) \to \mathcal{D}'(Y)$ defines a unique distribution $K \in \mathcal{D}'(X\times Y)$ such that \eqref{eq:DualPairing} holds.
	\label{thm:SchwartzKernel}
\end{thm}

\begin{thm}[Schur test]
	Let $X,Y$ be two measurable spaces, $(B, \norm{\cdot})$ a Banach space, and $T:L^2(Y,B) \to L^2(X,B)$ the integral operator with kernel $K: X\times Y \to B$. If a constant $C>0$ exists such that
	\begin{align}
		\sup_x \int_Y \norm{K(x,y)} \diff y \leq C, \ \sup_y \int_X \norm{K(x,y)} \diff x \leq C,
	\end{align}
	then $T$ is a bounded operator. Moreover, the operator norm of $T$ is bounded by $C$.
	\label{thm:SchurTest}
\end{thm}

\begin{thm}[Cotlar--Stein]
	Let $H_1$, $H_2$ be two Hilbert spaces and $A_j \in \mathfrak{B}(H_1,H_2)$ a bounded operator for every $j\in \mathbb{N}$. If a constant $C>0$ exists such that
	\begin{align}
		\sup_j \sum_{k=1}^{\infty} \norm{A_j A_k^*}^{\frac{1}{2}}_{\mathfrak{B}(H_2)} \leq C, \ \sup_j \sum_{k=1}^{\infty} \norm{A_j^* A_k}^{\frac{1}{2}}_{\mathfrak{B}(H_1)} \leq C,
	\end{align}
	then $\sum_{j=1}^n A_j$ converges in the strong operator topology to a bounded operator $A$ as $n\to\infty$. Moreover, the operator norm of $A$ is bounded by $C$.
	\label{thm:CotlarStein}
\end{thm}

%\begin{proof}
%	Using the Cauchy-Schwarz inequality, we obtain
%	\begin{align}
%		\norm{Tf(x)}^2 &\leq \left(\int_Y \norm{K(x,y)}^{\frac{1}{2}}\norm{K(x,y)}^{\frac{1}{2}}\norm{f(y)} \diff y\right)^2 \notag \\ 
%		&\leq C \int_Y \norm{K(x,y)}\norm{f(y)}^2 \diff y.
%	\end{align}
%	From this it follows that
%	\begin{align}
%		\norm{Tf}_{L^2(X,B)}^2 &= \int_X \norm{Tf(x)}^2 \diff x \notag \\
%		&\leq C \int_X \int_Y \norm{K(x,y)}\norm{f(y)}^2 \diff y \diff x \notag \\
%		&\leq C^2 \int_Y \norm{f(y)}^2 \diff y = C^2 \norm{f}^2_{L^2(Y,B)}.
%	\end{align}
%	Hence, $T$ is a bounded operator with operator norm at most $C$.
%\end{proof}

\begin{thm}[Peetre's inequality]
	Let $t\in \mathbb{R}$ and set $\jap{x} = \sqrt{1+|x|^2}$. Then, for every $x,y\in \mathbb{R}^d$,
	\begin{align}
		\jap{x}^t \leq 2^{|t|} \jap{y}^t \jap{x-y}^{|t|}.
	\end{align}
	\label{thm:Peetre}
\end{thm}

\begin{defn}
	The \textbf{symmetric rearrangement} $A^*$ of a Lebesgue measurable set $A \subset \mathbb{R}^d$ is the ball centred at the origin whose Lebesgue measure is the same as that of~$A$. 
\end{defn}

\begin{defn}
	Let $f$ be a measurable nonnegative function. Then
	\begin{align}
		f^*(x) = \int_0^\infty \mathbb{1}_{\{y \in \mathbb{R}^d \mid f(y)>t\}^*}(x) \diff t
	\end{align}
	defines the \textbf{symmetric decreasing rearrangement} $f^*$ of $f$.
\end{defn}

\begin{exmp}
	Let $\Lambda \geq 0$, $p>0$, and $f_\Lambda(x) = (1-\chi_\Lambda(x)) |x|^{-p}$, where $\chi_\Lambda(x)$ is equal to $1$ if $|x| \leq \Lambda$, and $0$ otherwise. The level sets $\{f_\Lambda > t \} = \{ y \in \mathbb{R}^d \mid f(y) > t\}$ contain all $y \in \mathbb{R}^d$ with $\Lambda < |y| < t^{-1/p}$ (i.e. $\{f_\Lambda > t \} = B_{t^{-1/p}}\backslash B_\Lambda$). The Lebesgue measure of $B_{t^{-1/p}}\backslash B_\Lambda$ equals $\mathrm{vol}(B_1) (t^{-d/p} - \Lambda^d)$, where $\mathrm{vol}(B_1)$ is the volume of the $d$-dimensional unit ball. Thus, $\{f_\Lambda > t \}^* = B_{(t^{-d/p} - \Lambda^d)^{1/d}}$ and
	\begin{align}
		f_\Lambda^*(x) = \int_0^\infty \mathbb{1}_{B_{(t^{-d/p} - \Lambda^d)^{1/d}}}(x) \diff t = \int_0^{(|x|^d + \Lambda^d)^{-\frac{p}{d}}} \diff t = (|x|^d + \Lambda^d)^{-\frac{p}{d}}.
	\end{align}
	\label{exmp:SDR1}
\end{exmp}

\begin{exmp}
	If $g_y(x) = g(x-y)$ is the translation of a symmetric decreasing function $g$, then the symmetric decreasing rearrangement of $g_y$ is $g_y^* = g_0 = g$.
	\label{exmp:SDR2}
\end{exmp}

\begin{thm}[Hardy--Littlewood inequality]
	If $f,g$ are nonnegative measurable functions on $\mathbb{R}^d$ and vanishing at infinity, then
	\begin{align}
		\int_{\mathbb{R}^d} f(x)g(x) \diff x \leq \int_{\mathbb{R}^d} f^*(x) g^*(x) \diff x.
	\end{align}
	\label{thm:HardyLittlewood}
\end{thm}

\begin{prop}[\cite{schmidt2019}]
	For all $\nu, \sigma \geq 0$ and $\alpha, \gamma > 0$, an $\epsilon_0 > 0$ exists such that, for all $d \in (\nu + \sigma, \nu + \sigma + \alpha \gamma) \cap \mathbb{N}$, $\Lambda, \Omega \geq 0$, $\Xi \in \mathbb{R}^d$, and $0 < \epsilon \leq \epsilon_0$,
	\begin{align}
		\int_{\mathbb{R}^d} \frac{|\xi|^{-\nu} |\Xi - \xi|^{-\sigma} (1-\chi_\Lambda(\xi)) }{(|\Xi - \xi|^\gamma + |\xi| + \Omega)^\alpha} \diff \xi \leq C \Omega^{-\alpha + (d-\nu-\sigma)/\gamma + \epsilon(\Lambda)} \Lambda^{-\epsilon(\Lambda)},
		\label{eq:IntegralEstimate}
	\end{align}
	where $\epsilon(\Lambda)$ is equal to $\epsilon$ if $\Lambda > 0$, and $0$ if $\Lambda = 0$.
	\label{prop:IntegralEstimate}
\end{prop}

\begin{proof}	
	For $0< \epsilon < \alpha$, the denominator of the integrand is bounded by
	\begin{align}
		(|\Xi - \xi|^\gamma + |\xi| + \Omega)^{-\alpha} &= (|\Xi - \xi|^\gamma + |\xi| + \Omega)^{-\alpha+\epsilon(\Lambda)} (|\Xi - \xi|^\gamma + |\xi| + \Omega)^{-\epsilon(\Lambda)} \notag \\
		&\leq (|\Xi - \xi|^\gamma + \Omega)^{-\alpha+\epsilon(\Lambda)} |\xi|^{-\epsilon(\Lambda)}.
	\end{align} 
	Thus, the integral under consideration in \eqref{eq:IntegralEstimate} is bounded by $\int_{\mathbb{R}^d} f_\Lambda g_\Xi$, where $f_\Lambda(\xi) = (1-\chi_\Lambda(\xi))|\xi|^{-\nu-\epsilon(\Lambda)}$ and $g_\Xi(\xi) = (|\Xi - \xi|^\gamma +\Omega)^{-\alpha + \epsilon(\Lambda)} |\Xi - \xi|^{-\sigma}$. By Example \eqref{exmp:SDR1}, $f_\Lambda^*(\xi) = (|\xi|^d + \Lambda^d)^{-p/d}$ with $p=\nu+\epsilon(\Lambda)$ and, by Example \eqref{exmp:SDR2}, $g_\Xi^* = g_0$.
	
	If $\Lambda = 0$ and $\nu = 0$, then $f_\Lambda = f_0 = 1$ and from the Hardy--Littlewood inequality it follows that
	\begin{align}
		\int_{\mathbb{R}^d} f_0 g_\Xi = \int_{\mathbb{R}^d} \sqrt{g_\Xi} \sqrt{g_\Xi} \leq \int_{\mathbb{R}^d} g_0 = \int_{\mathbb{R}^d} \frac{|\xi|^{-\sigma}}{(|\xi|^\gamma + \Omega)^\alpha} \diff \xi.
	\end{align}
	Similarly, if $\Lambda = 0$ and $\nu > 0$, then, due to $f_0 = f_0^*$, 
	\begin{align}
		\int_{\mathbb{R}^d} f_0 g_\Xi \leq \int_{\mathbb{R}^d} f_0 g_0 = \int_{\mathbb{R}^d} \frac{|\xi|^{-\sigma-\nu}}{(|\xi|^\gamma + \Omega)^\alpha} \diff \xi.
	\end{align}
	Lastly, if $\Lambda > 0$, we obtain the following estimate:
	\begin{align}
		\int_{\mathbb{R}^d} f_\Lambda g_\Xi \leq \int_{\mathbb{R}^d} f_\Lambda^* g_0 = \int_{\mathbb{R}^d} |\xi|^{-\sigma} \frac{(|\xi|^d + \Lambda^d)^{-\frac{\nu + \epsilon}{d}}}{(|\xi|^\gamma + \Omega)^{\alpha - \epsilon}} \diff \xi \notag \\
		\leq \int_{\mathbb{R}^d} |\xi|^{-\sigma-\nu} \frac{\Lambda^{-\epsilon}}{(|\xi|^\gamma + \Omega)^{\alpha - \epsilon}} \diff \xi.
	\end{align}
	We combine these three inequalities and substitute $\xi$ with $\eta = \xi/\Omega^{1/\gamma}$:
	\begin{align}
		\int_{\mathbb{R}^d} f_\Lambda g_\Xi &\leq \int_{\mathbb{R}^d} |\xi|^{-\sigma-\nu} \frac{\Lambda^{-\epsilon(\Lambda)}}{(|\xi|^\gamma + \Omega)^{\alpha-\epsilon(\Lambda)}} \diff \xi \notag \\
		&\leq \Omega^{-\alpha+(d-\nu-\sigma)/\gamma + \epsilon(\Lambda)} \Lambda^{-\epsilon(\Lambda)} \int_{\mathbb{R}^d} \frac{|\eta|^{-\sigma-\nu}}{(|\eta|^\gamma + 1)^{\alpha-\epsilon(\Lambda)}} \diff \eta.
	\end{align}
	The last integral is finite if $d\in (\nu +\sigma, \nu + \sigma + \alpha\gamma)$ and if $\epsilon$ is sufficiently small.
\end{proof}

\begin{cor}
	Let $\zeta$ be a bounded continuous function on $\mathbb{R}^d$ with $\zeta(0) = 0$ and $\lim_{|\xi| \to \infty} \zeta(\xi) = 1$. Set $\zeta_\Lambda(\xi) = \zeta(\xi/\Lambda)$ for $\Lambda > 0$ and $\zeta_0(\xi) = 1$. Then, under the assumptions of Proposition \ref{prop:IntegralEstimate},
	\begin{align}
		\int_{\mathbb{R}^d} \frac{|\xi|^{-\nu} |\Xi - \xi|^{-\sigma} |\zeta_\Lambda(\xi)| }{(|\Xi - \xi|^\gamma + |\xi| + \Omega)^\alpha} \diff \xi \leq f(\Lambda) \Omega^{-\alpha + (d-\nu-\sigma)/\gamma + \epsilon(\Lambda)},
	\end{align}
	where $f$ is a continuous function on $[0,\infty)$ that converges to 0 as $\Lambda \to \infty$.
	\label{cor:IntegralEstimate}
\end{cor}

\begin{proof}
	The case $\Lambda = 0$ is covered by Proposition \ref{prop:IntegralEstimate}; hence, we assume that $\Lambda > 0$. We denote the integrand without $|\zeta_\Lambda(\xi)|$ by
	\begin{align}
		h_\Xi(\xi) = \frac{|\xi|^{-\nu} |\Xi - \xi|^{-\sigma}}{(|\Xi - \xi|^\gamma + |\xi| + \Omega)^\alpha}.
	\end{align}
	The function $\zeta$ is continuous and vanishes in 0; thus, for every $\delta > 0$, there is an $N \in \mathbb{N}$ such that $|\zeta_\Lambda(\xi)| = |\zeta(\xi/\Lambda)| \leq \delta$ for all $|\xi| \leq \Lambda / N$. We split the integral under consideration as follows:
	\begin{align}
		&\int_{\mathbb{R}^d} |\zeta_\Lambda(\xi)| h_\Xi(\xi) \diff \xi \notag \\
		&\ \ \ = \int_{\mathbb{R}^d} (1-\chi_{\Lambda / N}(\xi))|\zeta_\Lambda(\xi)| h_\Xi(\xi) \diff \xi + \int_{\mathbb{R}^d} \chi_{\Lambda / N}(\xi)|\zeta_\Lambda(\xi)| h_\Xi(\xi) \diff \xi.
	\end{align}
	In the first summand, we bound $|\zeta_\Lambda(\xi)|$ by constant independent of $\Lambda$ and apply Proposition \ref{prop:IntegralEstimate}:
	\begin{align}
		\int_{\mathbb{R}^d} (1-\chi_{\Lambda / N}(\xi))|\zeta_\Lambda(\xi)| h_\Xi(\xi) \leq C\Omega^{-\alpha + (d-\nu-\sigma)/\gamma + \epsilon(\Lambda/N)} (\Lambda/N)^{-\epsilon(\Lambda/N)}.
	\end{align}
	In the second summand, we bound $|\zeta_\Lambda(\xi)|$ by $\delta$ and $\chi_{\Lambda / N}$ by 1. Then, we apply Proposition \ref{prop:IntegralEstimate} with $\Lambda = 0$:
	\begin{align}
		\int_{\mathbb{R}^d} \chi_{\Lambda / N}(\xi)|\zeta_\Lambda(\xi)| h_\Xi(\xi) \diff \xi \leq \delta \int_{\mathbb{R}^d} h_\Xi(\xi) \diff \xi \leq \delta C \Omega^{-\alpha + (d-\nu-\sigma)/\gamma}.
	\end{align}
	Thus, we arrive at the following inequality:
	\begin{align}
		\int_{\mathbb{R}^d} h_\Xi(\xi) |\zeta_\Lambda(\xi)| \diff \xi \leq C ((\Lambda/N)^{-\epsilon(\Lambda/N)} + \delta) \Omega^{-\alpha + (d-\nu-\sigma)/\gamma + \epsilon}
	\end{align}
	The prefactor converges to $\delta$ as $\Lambda \to \infty$. As $\delta>0$ is arbitrary, the prefactor converges to 0 as $\Lambda \to \infty$.
\end{proof}

\begin{lem}
	Let $\zeta$ be as in Corollary \ref{cor:IntegralEstimate}. Then, for every $\Xi \in \mathbb{R}^3$ and $\nu \in (1,3)$, an $\epsilon_0 > 0$ exists such that, for every $0< \epsilon \leq \epsilon_0$,
	\begin{align}
		\int_{\mathbb{R}^3} \frac{|\zeta_\Lambda(\xi)| |\xi|^{-\nu}}{|\Xi-\xi|^2+1} \diff \xi \leq \frac{f(\Lambda)}{|\Xi|^{\nu - 1 - \epsilon}},
	\end{align}
	where $f$ is a continuous function on $[0,\infty)$ that converges to 0 as $\Lambda \to \infty$.
	\label{lem:IntegralEstimate}
\end{lem}

\begin{proof}
	First, we bound $|\zeta_\Lambda(\xi)|$ by a constant independent of $\Lambda$ and divide the integration area into $\{\xi \in \mathbb{R}^3 \mid |\Xi-\xi|\leq |\Xi|/2\}$ and $\{\xi \in \mathbb{R}^3 \mid |\Xi-\xi| > |\Xi|/2\}$.
	%\begin{enumerate}[wide, labelwidth=!, labelindent=0pt]
		If $|\Xi-\xi| \leq |\Xi|/2$, then $|\xi|\geq |\Xi|/2$ and
		\begin{align}
			\int_{|\Xi-\xi| \leq \frac{|\Xi|}{2}} \frac{|\xi|^{-\nu}}{|\Xi-\xi|^2+1} \diff \xi &\leq \frac{2^\nu}{|\Xi|^\nu} \int_{|\Xi-\xi| \leq \frac{|\Xi|}{2}} \frac{1}{|\Xi-\xi|^2+1} \diff \xi \notag \\
			&= \frac{2^\nu}{|\Xi|^\nu} \int_{|\xi| \leq \frac{|\Xi|}{2}} \frac{1}{|\xi|^2+1} \diff \xi \leq \frac{C}{|\Xi|^{\nu - 1}}.
		\end{align}
		If $|\Xi-\xi| > |\Xi|/2$, scaling $\xi$ by $|\Xi|$, we obtain
		\begin{align}
			\int_{|\Xi-\xi| > \frac{|\Xi|}{2}} \frac{|\xi|^{-\nu}}{|\Xi-\xi|^2+1} \diff \xi &\leq 
			\frac{1}{|\Xi|^{\nu-1}} \int_{\left|\frac{\Xi}{|\Xi|}-\xi\right| > \frac{1}{2}} \frac{1}{|\Xi/|\Xi|-\xi|^2|\xi|^\nu} \diff \xi \notag \\
			&\leq \frac{C}{|\Xi|^{\nu-1}}.
		\end{align}	
		The last integral is independent of $\Xi$ due to the rotation invariance of the Lebesgue measure. 	
	%\end{enumerate} 
	
	Next, we bound the integral under consideration by  
	\begin{align}
		\int_{\mathbb{R}^3} \frac{|\zeta_\Lambda(\xi)| |\xi|^{-\nu}}{|\Xi-\xi|^2+1} \diff \xi \leq 4 \jap{\Xi}^2 \int_{\mathbb{R}^3} |\zeta_\Lambda(\xi)| \jap{\xi}^{-2} |\xi|^{-\nu} \diff \xi \leq f(\Lambda) \jap{\Xi}^2
	\end{align}
	according to Peetre's inequality (see Theorem \ref{thm:Peetre}). The claim follows because, for nonnegative constants $a,b,c$, the inequalities $a\leq b$ and $a\leq c$ imply $a \leq b^{1-\epsilon} c^{\epsilon}$.
\end{proof}

\backmatter
\bibliographystyle{acm}
\bibliography{literature}

\end{document}